\documentclass[]{article}

\usepackage{packages}
\usepackage{notations}

\author[1]{\large Guillaume Barraquand}
\author[2]{\large Alexandre Krajenbrink}
\author[1]{\large Pierre Le Doussal}
\affil[1]{\small Laboratoire de Physique de l'\'Ecole Normale Sup\'erieure, ENS, Universit\'e PSL,  CNRS, Sorbonne Universit\'e, Universit\'e de Paris, 75005 Paris, France}
\affil[2]{\small SISSA and INFN, via Bonomea 265, 34136 Trieste, Italy}

\usetikzlibrary{matrix,decorations.pathreplacing,calc, positioning,automata,decorations.markings,arrows}
\usetikzlibrary{shapes}
\usetikzlibrary{patterns}
\usetikzlibrary{shapes.multipart}
\usetikzlibrary{arrows.meta, petri, topaths}

\tikzset{
    partial ellipse/.style args={#1:#2:#3}{
        insert path={+ (#1:#3) arc (#1:#2:#3)}
    }
}

\tikzset{middlearrow/.style={
        decoration={markings,
            mark= at position 0.55 with {\arrow[thick]{#1}} ,
        },
        postaction={decorate}
    }
}

\tikzset{
  coordsys/.pic={
    \draw[->] (0,0) -- ++(6mm,0pt);
    \draw[->] (0,0) -- ++(0pt,6mm);
  },
  myarr/.style={decoration={
      markings,
      mark=between positions 0 and 1 step 4mm with {\arrow{stealth}},
    },
    postaction=decorate
  },
}

\pgfkeys{
  /pgf/decoration/.cd,
  open/.is choice,
  open/left/.is if = openbraceleft,
  open/right/.is if = openbraceright,
  open/both/.style = {open/left = true, open/right = true},
  open/none/.style = {open/left = false, open/right = false},
  open/.initial = none
}

\pgfdeclaredecoration{open brace}{initial}
{
  \state{initial}[width=+\pgfdecoratedremainingdistance,next state=final]
  {
    \pgfpathmoveto{\pgfpointorigin}
    \ifopenbraceright
    \pgfpathmoveto{\pgfqpoint{0pt}{.5\pgfdecorationsegmentamplitude}}
    \else
    \pgfpathcurveto
    {\pgfqpoint{.15\pgfdecorationsegmentamplitude}{.3\pgfdecorationsegmentamplitude}}
    {\pgfqpoint{.5\pgfdecorationsegmentamplitude}{.5\pgfdecorationsegmentamplitude}}
    {\pgfqpoint{\pgfdecorationsegmentamplitude}{.5\pgfdecorationsegmentamplitude}}
    \fi
    {
      \pgftransformxshift{+\pgfdecorationsegmentaspect\pgfdecoratedremainingdistance}
      \pgfpathlineto{\pgfqpoint{-\pgfdecorationsegmentamplitude}{.5\pgfdecorationsegmentamplitude}}
      \pgfpathcurveto
      {\pgfqpoint{-.5\pgfdecorationsegmentamplitude}{.5\pgfdecorationsegmentamplitude}}
      {\pgfqpoint{-.15\pgfdecorationsegmentamplitude}{.7\pgfdecorationsegmentamplitude}}
      {\pgfqpoint{0\pgfdecorationsegmentamplitude}{1\pgfdecorationsegmentamplitude}}
      \pgfpathcurveto
      {\pgfqpoint{.15\pgfdecorationsegmentamplitude}{.7\pgfdecorationsegmentamplitude}}
      {\pgfqpoint{.5\pgfdecorationsegmentamplitude}{.5\pgfdecorationsegmentamplitude}}
      {\pgfqpoint{\pgfdecorationsegmentamplitude}{.5\pgfdecorationsegmentamplitude}}
    }
    {
      \pgftransformxshift{+\pgfdecoratedremainingdistance}
      \ifopenbraceleft
      \pgfpathlineto{\pgfqpoint{0pt}{.5\pgfdecorationsegmentamplitude}}
      \else
      \pgfpathlineto{\pgfqpoint{-\pgfdecorationsegmentamplitude}{.5\pgfdecorationsegmentamplitude}}
      \pgfpathcurveto
      {\pgfqpoint{-.5\pgfdecorationsegmentamplitude}{.5\pgfdecorationsegmentamplitude}}
      {\pgfqpoint{-.15\pgfdecorationsegmentamplitude}{.3\pgfdecorationsegmentamplitude}}
      {\pgfqpoint{0pt}{0pt}}
      \fi
    }
  }
  \state{final}
  {}
}
\newcommand{\PhaseDiagram}{
    \begin{tikzpicture}[scale=0.64, every node/.style={scale=1, inner sep=0pt}]
    \draw[thick, ->,>=stealth]   (0,0) -- (7,0) node[above=0.2] (A) { $A$ };
    \draw[thick] (0,0) -- (-2.5,-2.5);
    \node[state,minimum size=5pt, fill=black, regular polygon, regular polygon sides=3, shape border rotate=180] (GOEcritical2) at (0,8.95) {~} ;
    \draw[dotted]   (7,0) -- (9,0) node[above left =0.15] (Ainfty) {$+\infty$};
    \draw[loosely dotted]   (9,-2.5) -- (9,9)  ;
    \draw[dash dot]   (-2.5,0) node[left=0.2] (criticalB) {$B=-\frac{1}{2}$} -- (0,0) ;

    \draw[thick, ->,>=stealth]   (0,0) -- (0,7) node[right=0.2] (B) {$B$ };
    \node[state,minimum size=5pt, fill=black, regular polygon, regular polygon sides=3, shape border rotate=90] (GOEcritical) at (8.95,0) {~} ;
    \draw[dotted]   (0,7) -- (0,9) node[ below right =0.15] (Binfty) {$+\infty$};
    \draw[dash dot]   (0,-2.5) node[below=0.2] (criticalA) {$A=-\frac{1}{2}$} -- (0,0) ;
    \draw[loosely dotted]   (-2.5,9) -- (9,9) ;
    
    \draw[thick,  dashed]   (-4,4) -- (3.5,-3.5) node[below=0.2, align=center] (stationary) {$A+B+1=0$\\ \textit{Stationary correlations}};

    \node (GSE) at (3.5,3.5) {\textbf{GSE}};
    \node (Unbound) at (5.4,5.7) [state,ellipse,minimum width=65pt] {\textit{Unbound}};

    \node (BoundWall) at (-3,5.7) [state,ellipse,minimum width=60pt] {\textit{Bound to wall}}; 
    \node (BoundBrownian) at (5.6,-1.75) [state,ellipse,minimum width=50pt] {\textit{Bound to Brownian}};

    \node[rotate=90] (GOE1) at (0,3.5) {\textbf{GOE}};
    \node (GOE2) at (3.5,0) {\textbf{GOE}};
    
    \node[align=center, scale=0.8] (BrowDiri) at (9,4.5) {Brownian IC\\ with  Dirichlet \\ boundary condition};
    
    \draw [<-,>=stealth] (9.2,0) -- (9.7,0) node[scale=0.8, right=0.2] (GOE3) {GOE};
    \draw [<-,>=stealth] (9.2,-1) -- (9.7,-1) node[scale=0.8, right=0.2] (Gaussian3) {Gaussian};
    \draw [<-,>=stealth] (9.2,1) -- (9.7,1) node[scale=0.8, right=0.2] (GSE3) {GSE};

    \node[align=center, scale=0.8] (droplet) at (4.5, 9) {Droplet IC};
    
    \draw [<-,>=stealth] (0,9.2) -- (0,9.7) node[scale=0.8, above=2pt] (GOE4) {GOE};
    \draw [<-,>=stealth] (-1.6,9.2) -- (-1.6,9.7) node[scale=0.8, above=2pt] (Gaussian4) {Gaussian};
    \draw [<-,>=stealth] (1.1,9.2) -- (1.1,9.7) node[scale=0.8, above=2pt] (GSE4) {GSE};

    \draw[thick, gray] (-.8,-.8) -- (-.8,.8) -- (.8,.8) -- (.8,-.8) -- cycle;
    \pattern[black!10, opacity=0.5, pattern=north east lines] (-.8,-.8) -- (-.8,.8) -- (.8,.8) -- (.8,-.8) -- cycle;
    \end{tikzpicture}  
}

\newcommand{\ZoomedDiagram}{
\begin{tikzpicture}[scale=0.5]
\draw[->,>=stealth, thick]   (0,0) -- (7,0) node[below=0.1]  {$\Arescaled$};
\draw[dotted, thick]   (7,0) -- (8,0) node[ right =0.1] {GOE};
\draw[dash dot, thick]   (-2.5,0) node[left=0.2] {$\Brescaled=0$} -- (0,0) ;

\draw[->,>=stealth, thick]   (0,0) -- (0,7) node[left=0.1] {$\Brescaled$ };
\draw[dotted, thick]   (0,7) -- (0,8) node[ above =0.1] {GOE};
\draw[dash dot, thick]   (0,-2.5) node[below=0.2] {$\Arescaled=0$} -- (0,0);

\draw[thick]  (0,0) -- (-2.5,-2.5);

\draw[thick,  dashed]   (-4,4) -- (3.5,-3.5) node[below=0.1, align=center] (stationary) {$\Arescaled+\Brescaled=0$\\ \textit{Stationary correlations}};

\draw (4,4) node[] {\textbf{$F^{(\Arescaled, \Brescaled)}$}};
\fill (0,0) circle(0.2);
\draw[<-] (0.3,0.2) to[bend right] (1.2,0.8) node[anchor=south, align=center] {$ F^{(0,0)} $};

\draw[->] (6,6) -- (8,8) node[anchor=south west] {GSE};
\draw (6.5,6.5) node[anchor=west, scale=0.8] {$\Arescaled, \Brescaled\to+\infty$};

\draw[->] (6,4) -- (8,4) node[anchor=west] {$F^{(\Brescaled)}$};
\draw (7,4) node[anchor=north, scale=0.8] {$\Arescaled \to\infty$};

\draw[->] (4,6) -- (4,8) node[anchor=south] {$F^{(\Arescaled)}$};
\draw (4,7) node[anchor=south, scale=0.8, rotate=90] {$\Brescaled \to\infty$};
    
\draw[thick, gray,opacity=0.2] (-5.5,-5.5) -- (-5.5,5.5) -- (5.5,5.5) -- (5.5,-5.5) -- cycle;
\pattern[black!10, opacity=0.2, pattern=north east lines] (-5.5,-5.5) -- (-5.5,5.5) -- (5.5,5.5) -- (5.5,-5.5) -- cycle;
\end{tikzpicture}
}

\newcommand{\LogGammaPoly}{
      \begin{tikzpicture}[scale=0.8]
  \draw[->, thick, >=stealth', gray] (0,0) -- (10, 0);
  \draw[->, thick, >=stealth', gray] (0,0) -- (0, 5.5);
  \draw[thick, gray]  (0,0) -- (5.5,5.5);
  \draw (-.5, -.2) node{{\footnotesize $(1,1)$}};
  \draw (9.6, 5.3) node{{\footnotesize $(n,m)$}};
  \draw[->, >=stealth'] (10,2.5) node[anchor=west]{{\footnotesize $w_{i,j} \sim \mathrm{Gamma}^{-1}(\alpha_i + \alpha_j), \,\, i>j$}} to[bend right] (8,2);
  \draw[->, >=stealth'] (2.5,4.5) node[anchor=south]{{\footnotesize $w_{i,i} \sim \mathrm{Gamma}^{-1}(\diag + \alpha_i) $}} to[bend right] (3,3);
  \clip (0,0) -- (9.5,0) -- (9.5,5.5) -- (5.5, 5.5) -- (0,0);
  \draw[dotted, gray] (0,0) grid (10,6);
  
  \draw[ultra thick] (0,0) -- (1,0) -- (2,0) -- (2,1) -- (2,2) -- (3,2) -- (4,2) -- (4,3) -- (5,3) -- (6,3) -- (6,4) -- (7,4) -- (8,4) -- (8,5) -- (9,5);
  \end{tikzpicture}
}

\begin{document}

\title{\bf \Large Half-space stationary Kardar-Parisi-Zhang equation}
\date{\small \today}

\maketitle

\begin{abstract}
We study the solution of the Kardar-Parisi-Zhang (KPZ) equation for the stochastic growth
of an interface of height $h(x,t)$ on the positive half line, equivalently the free energy of the continuum directed polymer in a half space with a wall at $x=0$. The boundary condition $\partial_x h(x,t)|_{x=0}=A$
corresponds to an attractive wall for $A<0$, and leads to the binding of the polymer to the wall below the critical value $A=-1/2$. 
Here we choose the initial condition $h(x,0)$ to be a Brownian motion in $x>0$ with drift $-(B+1/2)$.
When $A+B \to -1$, the solution is stationary, i.e. 
$h(\cdot,t)$ remains at all times a Brownian motion with the same drift, 
up to a global height shift $h(0,t)$. We show that the distribution of this height shift is invariant under the exchange of parameters $A$ and $B$. For any $A,B > - 1/2$, we provide an exact formula characterizing the distribution of $h(0,t)$
at any time $t$, using two methods: the replica Bethe ansatz and a discretization called the
log-gamma polymer, for which moment formulae were obtained. 
We analyze its large time asymptotics for various ranges of parameters $A,B$. 
 In particular, when $(A, B) \to (-1/2, -1/2)$, the critical stationary case, the fluctuations of the interface are governed by a universal distribution akin to the Baik-Rains distribution arising in stationary growth on the full-line. It can be  expressed in terms of a simple Fredholm determinant, or equivalently in terms of the Painlev\'e II transcendent. 
This provides an analog for the KPZ equation, of some of the results
recently obtained by Betea-Ferrari-Occelli in the context of stationary half-space last-passage-percolation. 
From universality, we expect that limiting distributions found in both models can be shown to coincide.
\end{abstract}

\vspace*{1cm}

{\hypersetup{linkcolor=black}
\setcounter{tocdepth}{2}
	\tableofcontents
}

\newpage

\section{Introduction} 

\noindent The Kardar-Parisi-Zhang equation \cite{KPZ} describes the stochastic dynamics of the height field, $h(x,t)$, of a growing interface in the continuum, as a function of time, or equivalently, of the free energy of a continuum directed polymer in a random potential as a function of its length. 
In one space dimension, $x \in \mathbb{R}$, it is at the center of a vast universality class, the KPZ class, which contains numerous solvable discrete models with the same large scale behavior. Examples of such solvable models include random permutations and the associated PNG growth model, interacting particle systems (TASEP, ASEP, and variants), the stochastic six-vertex model and other stochastic vertex models, and  several discrete models of directed polymers (DP). There has been many papers on the subject and we refer the reader to the reviews and lecture notes \cite{baik2001symmetrized, SpohnTASEPGSE, spohn2006exact, ferrari2010interacting,  quastel2012lectures, corwin2012kardar, corwin2014macdonald, borodin2012lectures, borodin2014integrable, quastel2015one, takeuchi2016appetizer, borodin2016lectures, quastel2017totally}. 
Exact results have also been
obtained for the KPZ equation itself, and its equivalent system, the continuous directed polymer \cite{we,dotsenko,dotsenko22,dotsenko33,spohnKPZEdge,spohnKPZEdge2,spohnKPZEdge3,amir2011probability,sineG,CLDflat,CLDflat2,PLDCrossoverDropFlat,crissingprob2,dotsenkoGOE,
Quastelflat,SasamotoStationary,BCFV,KPZFixedPoint}. These results 
have been obtained on the full line, $x \in \mathbb{R}$, and have allowed proving the existing conjectures for the 
scaling exponents of the height fluctuations, $\delta h \sim t^{1/3} \sim x^{1/2}$, and to
predict and classify the universal probability distributions which arise for
various initial conditions. Most notably, the droplet and flat initial conditions were shown to lead, at 
large time, to the Tracy-Widom distributions \cite{tracy1994level, tracy1996orthonormal}  for the height at one point (centered and scaled by $t^{1/3}$) 
associated respectively to the 
Gaussian unitary and orthogonal ensembles, GUE and GOE, of random matrix theory.
Some of these predictions have been successfully tested in 
experiments \cite{exp4,exp444,Takeuchi,TakeuchiCrossover,TakeuchiHHLReview,deNardisPLDTakeuchi}.

\bigskip

\noindent It is also interesting for applications \cite{TakeuchiItoPng} to study models in the KPZ class restricted to the half line $x \in \mathbb{R}_+$, for which fewer results are available at present. For some specific coupling to the boundary (at $x=0$) the solvability properties can sometimes be preserved. One can usually define a
parameter, that we will call $A$, and which will be defined below in Eq.~\eqref{eq:boundaryconditionSHE} for the KPZ equation, characterizing this coupling. For the KPZ equation, $A=0$ and $A=+\infty$ correspond respectively to Neumann and to Dirichlet boundary conditions. The half-line KPZ equation is equivalent to a continuum directed polymer on a half-space with a wall at $x=0$, which is repulsive for $A>0$, and attractive for $A<0$. A remarkable feature of the half-space problem is the existence of a critical value $A_c$ of the parameter $A$ at which a phase transition occurs. In the polymer language it corresponds to a binding of the polymer to the wall if the attraction is strong enough $A<A_c$. The existence of this transition
was predicted by Kardar in 1985 \cite{KardarTransition} for the continuum directed polymer, using the replica Bethe ansatz.
No prediction for the height distribution was obtained however.

\bigskip 

\noindent In mathematics, in a pioneering paper
in 1999, Baik and Rains \cite{baik2001symmetrized} proved the existence of a similar transition 
in the context of the longest increasing sub-sequences (LIS) of symmetrized  random permutations.
There it was shown that, in the unbound phase, for the droplet initial condition, and near $x=0$, the (scaled) height fluctuations obey the Tracy-Widom  distribution associated to the Gaussian symplectic ensemble (GSE).
In the bound phase $A<A_c$, the fluctuations are simply Gaussian. Exactly at the transition, $A=A_c$,
fluctuations also obey the Tracy-Widom distribution but the one associated to the Gaussian orthogonal ensemble (GOE).
These results were extended in other models, e.g. a GSE-GUE crossover
was shown as the endpoint position is varied towards the bulk in the PNG model \cite{sasamotohalfspace}. For the TASEP in a half-space, equivalent to 
last passage percolation in a half-quadrant \cite{SpohnTASEPGSE}, 
similar results were obtained in \cite{baik2018pfaffian,baik2018pfaffian22} using Pfaffian-Schur processes \cite{borodin2005eynard}, 
in particular concerning the crossover as the parameter controlling the boundary is varied simultaneously with the distance to the boundary. {All these half-space models discussed above can be studied via the framework of Pfaffian point processes and random matrix theoretic techniques (these models are free-fermionic). However these models do not converge to the KPZ equation. The models converging to the KPZ equation such as the asymmetric simple exclusion process (ASEP) or directed polymers are not directly related to Pfaffian point processes (they are non-free-fermionic). Among those, 
ASEP  
was studied in \cite{TWhalf, tracy2013asymmetric, halfASEPBarraquand} and 
the half-space log-gamma polymer
was studied in \cite{OConnelSymmetrized,barraquand2018half, bisi2020geometric}. }
\bigskip 

\noindent For the KPZ equation in the half-space, a solution for the one-point height distribution 
near the wall valid at any time, was obtained for $A=+\infty$, 
for the droplet initial condition using the replica Bethe
ansatz in \cite{GueudrePLD} (see also \cite{krajenbrink2018large}).
Another solution (also non-rigorous) 
was obtained for $A=0$ in \cite{borodin2016directed}, with related but 
different methods using nested contour integral representations. In both cases
the large time limit is found to be GSE Tracy-Widom , consistent with the general picture obtained
by Baik and Rains as discussed above. 
Next, taking the limit from a half-space ASEP, a rigorous 
solution for the critical case $A=- \frac{1}{2}$ was obtained \cite{halfASEPBarraquand},
leading to GOE Tracy-Widom  fluctuations. More recently, in \cite{AlexLD}, a solution valid for any time 
was found using the replica Bethe ansatz for any $A> -1/2$, 
and which leads to the GOE Tracy-Widom  distribution at the critical point $A = -1/2$.
In \cite{deNardisPLDTT} a solution from the RBA, taking into account
bound states, was obtained, in agreement with the results of \cite{AlexLD}.

\bigskip 

\noindent A remarkable property of the KPZ equation on the full line is that the
stationary measure is the Brownian motion in the sense that if the 
initial condition $h(x,0)$ is a two sided Brownian motion (with the appropriate amplitude) the
PDF of the height difference between two space points is time independent. This leaves 
a uniform shift $h(0,t)$ whose fluctuations, scaled by $t^{1/3}$, was shown to follow the so-called 
Baik-Rains distribution, which is universal over the KPZ class. For the KPZ equation, the solution for all time with Brownian initial condition was found using the RBA in \cite{SasamotoStationary} and
proved rigorously in \cite{BCFV}. Early investigations on stationary models in the KPZ universality class started with  \cite{baik2000limiting} in the context of the polynuclear growth model, which introduced the Baik-Rains distribution as the limiting distribution of height fluctuations. For a very similar model (TASEP), the spatial correlations were investigated in \cite{baik2010limit}. Outside the class of free fermionic models, besides \cite{SasamotoStationary, BCFV} that we have already discussed, let us also mention \cite{aggarwal2018current} which proved the one point convergence of ASEP height function towards the Baik-Rains distribution.

\bigskip 

\noindent The aim of the present paper is to address the same problem, but for the half-line. We
study the KPZ equation on the half-line with a boundary parameter $A$ and an initial condition 
chosen as a unit one-sided Brownian motion with a drift, which we denote for later convenience as $- (B + \frac{1}{2})$. The problem is thus determined by two parameters, $A$ and $B$ and we will
study the phase diagram in the $(A,B)$ plane. As we show, on the line $A+B+1=0$ the initial condition is stationary, in the same sense as above, i.e. the PDF of the height difference between two space points remains at all time the one of the same Brownian motion. We show furthermore that the distribution of the height shift, $h(0,t)$, is invariant under the exchange of parameters $A$ and $B$. 
We will use two methods to obtain the exact generating function which characterizes the distribution of 
$h(0,t)$, at any time $t$. The first one is the replica Bethe ansatz and is a generalization of the calculation presented in \cite{AlexLD}. The second one starts from a known moment formula for the so-called
log-gamma polymer \cite{seppalainen2012scaling}, and takes the continuum limit to the KPZ equation. The obtained formula is
valid for any point in the quadrant $A,B > -1/2$. We then study its large time limit, leading to the phase diagram of Fig.~\ref{fig1}. In the $(A,B)$ plane, the point $(A,B)= (-1/2,-1/2)$ plays a special role, as the system is critical with respect to the boundary, and at the same time stationary, and will be called critical stationary. 
At this point, we show that in the large time limit the fluctuations of the interface are governed by a universal distribution, that we express in terms of a simple Fredholm determinant, equivalently in terms of the 
Painlev\'e II transcendent. In a sense, it is the analog for the half-line problem,
of the Baik-Rains distribution for the full line. Inside the quadrant $A,B > -1/2$ the distribution 
is obtained as the GSE Tracy-Widom  distribution and on scales $A+\frac{1}{2}, B + \frac{1}{2} \sim t^{-1/3}$, 
there is a universal two-parameter crossover distribution that we obtain in the quadrant.
In the so-called bound phase $A<-1/2$ or $B<-1/2$ away from the critical stationary point (including along the line $A+B+1=0$), the fluctuations are expected to be Gaussian, except when $A=B<-1/2$ where we expect the fluctuations to be distributed as the maximum {eigenvalue of a $2\times 2$ GUE matrix}.  The crossover to this behavior is however beyond the scope of this paper. Finally, note that our results will be consistent in the limit $B \to +\infty$ with all previous results for the droplet initial condition, in particular with the ones in \cite{AlexLD} for $A\geq -1/2$. The level of mathematical rigor of all these results is discussed in Section \ref{sec:mathematical}. 

\bigskip 

\noindent It is important to mention that very recently Betea, Ferrari and Occelli \cite{betea2019stationary} studied stationary half-space KPZ growth 
for a discrete model, the last-passage-percolation with exponential weights (i.e. a zero-temperature polymer). They obtained a formula for the asymptotic height distribution, depending on several parameters controlling the distance to the boundary and the position on the line $A+B+1$ near the critical stationary point. We expect, from the universality within the KPZ class, that our present result and theirs should match. The Pfaffian formula of \cite[Theorem 2.7]{betea2019stationary} and our formulae \eqref{eq:criticalstationaryintro}, \eqref{eq:Fdeterminatalformulaintro} and \eqref{eq:FintermsofF1F2intro} look different. Although we also have a Pfaffian representation, Eq.~\eqref{eq:defFepsiloneta0}, the associated kernels are different and we do not have a general method to show the equivalence of the Fredholm Pfaffians. A similar issue is discussed in \cite[Section 4.2, Eq.~(74)]{AlexLD} . Note that our formula allows for a very easy numerical evaluation of the CDF, given below in Fig.~\ref{fig:plots}, and of the first moments, and allows
us to determine tail estimates.
\bigskip 

It would be interesting to study the distribution of $h(x,t)$ when $x$ is at a distance of order $t^{2/3}$ from $0$, and $A,B$ are scaled close to $-1/2$. This would correspond to varying the parameter $\eta$ in \cite{betea2019stationary}. However, while we can obtain some integral formula for the moments of $Z(x,t)$ in the case $x>0$, see \eqref{eq:momentsKPZhalfspaceBrownian} below, we do not expect that it can be rewritten as a Pfaffian formula and the asymptotic analysis would require to develop other methods.  We leave this for future consideration.

\subsection*{Outline} First in Section~\ref{sec:Model} we define the models 
and make a summary of the main results and formulae obtained in this paper. In the two following sections we compute the moments of the polymer partition sum, i.e. the exponential moments of the KPZ field, by two methods. In 
Section~\ref{sec:BA} we present the derivation using the Bethe ansatz. In 
Section~\ref{sec:loggamma} we obtain the moments starting from the log-gamma polymer,
and check that the two moment formulae coincide. In Section~\ref{sec:FP} 
we obtain, a Pfaffian formula for the Laplace transform generating function by summing up the moments.
This leads to our first result, valid for all times and any $A,B > -1/2$, 
for the generating function as a Fredholm Pfaffian in terms of a matrix kernel. The large time
limit of this formula, and of the matrix kernel, is studied in section 
\ref{sec:largeTlimit}. In Section~\ref{sec:scalarkernel} we extend the method described in 
\cite{krajenbrink2018large} to obtain a formula for the Laplace transform generating function in terms 
of a scalar kernel, valid for all times and $A,B > -1/2$. In Section~\ref{subsec:largetime} 
we perform the large time limit on this scalar kernel, which leads to a two parameter
family of interpolating kernels near the point $(A,B)=(-1/2,-1/2)$. From it we obtain various 
limits, including our formula for the critical stationary distribution. 

\subsection*{Acknowledgements}
AK and PLD acknowledge support from ANR grant ANR-17-CE30-0027-01 RaMaTraF. AK acknowledges support from ERC under Consolidator grant number 771536 (NEMO).

\section{Model and main results} 
\label{sec:Model} 

\subsection{Model}
\label{subsec:Model} 

In this paper we study the KPZ equation, which reads, in dimensionless units 
\be \label{eq:KPZ}
\partial_t h(x,t) = \partial_x^2 h(x,t) + (\partial_x h(x,t))^2 + \sqrt{2} \, \xi(x,t) 
\ee 
where $\xi(x,t)$ is the standard space-time white noise, with $\mathbb{E}[ \xi(x,t) \xi(x',t') ]=\delta(x-x') \delta(t-t')$.
One introduces, via the Cole-Hopf mapping, the directed polymer partition sum $Z(x,t) = e^{h(x,t)}$, where 
$h(x,t)$ is solution of the KPZ equation \eqref{eq:KPZ}. It satisfies the multiplicative noise stochastic heat equation (SHE)
\be \label{SHE}
\partial_t Z(x,t) = \partial_x^2 Z(x,t) + \sqrt{2} \, \xi(x,t) \, Z(x,t) 
\ee 
understood here with the Ito prescription. Equation \eqref{SHE} means that $Z(x,t)$ can be 
seen as a partition sum over continuum directed paths in the random potential 
$-\sqrt{2} \, \xi(x,t)$, with the endpoint at time $t$ fixed at position $x$.\\

\begin{definition}
We consider the SHE 
on the half-line $x \geqslant 0$ with boundary parameter $A$ and $(x,t)\mapsto Z(x,t)$ to be the solution to \eqref{SHE} (it can be shown that the solution is unique, see \cite{corwin2016open, parekh2019kpz123} and references therein)  with the boundary condition
\begin{equation}
\partial_x Z(x,t)_{\mid x=0}=A Z(0,t).
\label{eq:boundaryconditionSHE}
\end{equation}
and with the Brownian initial data, in presence of a drift $-1/2-B$
\be
Z(x,0)=e^{\mathcal{B}(x)-(1/2+B) x}
\ee 
where $\mathcal{B}(x)$ is the standard Brownian (i.e. with $\mathcal{B}(0)=0$). We will sometimes denote the solution by $Z_A^B(x,t)$ to emphasize the dependence in parameters $A,B$, or simply $Z(x,t)$ when parameters are clear from the context. 
\label{def:SHEboundarybrownian}
\end{definition}

We have shifted by $1/2$ the drift parameter to make more explicit a remarkable symmetry between parameters $A$ and $B$. Indeed, we show  (see  Claim~\ref{prop:symmetryAB}) that for any $A,B\in \R$, we have the equality in distribution 
$$ Z_A^{B}(x=0,t)= Z_B^{A}(x=0,t), \text{ for any }t>0.$$
When $B$ goes to $+\infty$, we recover a result recently proved in \cite[Theorem 1.1]{parekh2019}. 
	
	\bigskip 
	
	When $A+B+1=0$, the model defined in Definition~\ref{def:SHEboundarybrownian} is stationary in the sense that for any fixed time $t$, the spatial process $\lbrace Z(x,t)/Z(0,t)\rbrace_{x\geqslant 0}$ has the same distribution as $\lbrace Z(x,0)\rbrace_{x\geqslant 0}$, that is the exponential of a standard Brownian motion with drift $-1/2-B$. Equivalently, the distribution of the slope field $\partial_x h(x,t)$ is time-stationary. Let us explain where this condition $A+B+1=0$ comes from.  In Section~\ref{sec:loggamma}, we consider a discretization of the KPZ equation, the log-gamma directed polymer. We identify in Section~\ref{sec:stationarystructure} initial and boundary conditions for the log-gamma polymer which make increments of the partition function stationary in time, using a result from \cite{seppalainen2012scaling} which deals with the full-space case. The log-gamma polymer partition function converges weakly to the stochastic heat equation from Definition~\ref{def:SHEboundarybrownian} at high temperature, see details in Section~\ref{sec:convergenceloggammaSHE} (note that we provide only a sketch of proof of this result based on the combination of results from \cite{wu2018intermediate} and \cite{parekh2019}). Hence we may pass to the limit, and taking into account the precise scalings, we obtain that the spatial process $\lbrace Z(x,t)/Z(0,t)\rbrace_{x\geqslant 0}$ is stationary when $A+B+1=0$. 
	
\bigskip 
There may exist other initial conditions for the half-space KPZ equation (or equivalently the stochastic heat equation) such that the slope field $\partial_x h(x,t)$ is time-stationary. Indeed, the KPZ equation arises as a scaling limit of the height function of particle systems such as ASEP for which other stationary distributions exist (see Refs.~\cite{liggett1975ergodic, derrida1993exact,grosskinsky, derrida2004asymmetric, bryc2019limit}).

\bigskip 

Before presenting our main results, let us clarify the meaning of the boundary condition \eqref{eq:boundaryconditionSHE}. As a process in $x$, $Z(x,t)$ has the same regularity as a Brownian motion, hence $\partial_x Z(x,t)$ cannot be associated to a real value. To make sense of \eqref{eq:boundaryconditionSHE}, we say \cite{corwin2016open} that $Z(x,t)$ is a solution  of \eqref{SHE} if it satisfies 
\begin{equation}
Z(x,t) = \int_{0}^{\infty} p^A_t(x,y)Z(y,0)\mathrm d y + \int_0^{\infty}\mathrm d y \int_{0}^t p^A_{t-s}(x,y)Z(y,s)\xi(y,s) \mathrm d s,
\label{eq:mildform}
\end{equation}
where the last integral is an Ito integral and $p^A_t(x,y)$ is the heat kernel on the positive half-line (i.e. it solves the equation $\partial_t u = \Delta u$ with initial data $\delta_y$) that satisfies the boundary condition 
\begin{equation}
\partial_x p^A_t(x,y) \big\vert_{x=0} = A\;p^A_t(0,y), \quad \quad t>0, y>0. 
\end{equation}
The main consequence that we will use below is that  
\begin{equation}
\partial_{x_i} \mathbb E\left[ Z(x_1,t) \dots Z(x_n,t) \right] \Big\vert_{x_i=0} = A \; \mathbb E\left[ Z(x_1,t) \dots Z(x_n,t) \right] \Big\vert_{x_i=0}, \quad 1\leqslant i\leqslant n, 
\label{eq:boundaryconditionformoments}
\end{equation}
which can be obtained by replacing $Z(x_i,t)$  using \eqref{eq:mildform} inside the expectation and differentiating with respect to $x_i$. \\

In terms of directed polymers, $Z(x,t)$ can be represented as a partition sum over directed paths
\begin{equation}
Z(x,t)=\E_{\mathcal{B}}\left[\int_0^{+\infty}\rmd y \, e^{\mathcal{B}(y)-(B+1/2)y}\int_{x(0)=y}^{x(t)=x}\mathcal{D}x(\tau)e^{-\int_0^t \rmd \tau [\frac{1}{4}(\frac{\rmd x}{\rmd \tau})^2-\sqrt{2}\eta(x(\tau),\tau)+2A \delta(x(\tau))]}\right]\, ,
\end{equation}
where  $\eta$ is a space-time white noise and $\mathcal D$ denotes the ``measure on paths''; more precisely the path integral is defined as an expectation value over reflected Brownian bridges $x(\tau)\in \R_+$ (reflected at $x=0$) for a given realization of the Brownian initial condition $\mathcal{B}(y)$, followed by an expectation over the Brownian $\mathcal{B}$. The extra $\delta$ interaction ensures the proper boundary condition at $x=0$ for $Z(x,t)$, see \cite[Section 3.2]{borodin2016directed}.

\subsection{Presentation of the main results} 
\label{subsec:presentation} 

Our main results concern the height at $x=0$, $h(0,t)$. In general its large time 
behavior is expected to be
\be
h(0,t) \simeq v^{A,B}_{\infty} t + t^{\beta} \chi 
\ee 
where $\chi$ is an $\mathcal{O}(1)$ random variable, and $\beta$ the growth fluctuation exponent. 
In the quadrant $A,B \geq -1/2$, to which our
exact results are restricted, one has $\beta=1/3$ and $v^{A,B}_{\infty}= - \frac{1}{12}$. Hence to
present these results, we define everywhere the shifted variable
\be
H(t) = h(0,t) + \frac{t}{12} 
\ee 
Note however that in the so-called bound phase, which will not
be studied in great detail here, we expect a different value of $v^{A,B}_{\infty}$ with $\beta=1/2$ and different distributions for $\chi$ (see Section~\ref{sec:conjecturalidentity}).  
In the limit $B \to +\infty$, i.e. for the droplet initial condition, it was found \cite{deNardisPLDTT} that
$v^{A,+\infty}_{\infty} = - \frac{1}{12} + (A+\frac{1}{2})^2$ for $A \leq -1/2$. In the general $A,B$ case, we expect that 
$ v^{A,B}_{\infty} = - \frac{1}{12} + \Big(\min\big\lbrace A+\frac{1}{2} , B+\frac{1}{2}, 0\big\rbrace\Big)^2$, based on a heuristic argument presented in Section~\ref{sec:conjecturalidentity}.

\subsubsection{Finite time: Fredholm Pfaffian of matrix kernel}

Our main result valid for all time $t \geqslant 0$ and all $A,B > - \frac{1}{2}$ is that the following generating function defined for $\varsigma>0$
can be written as a Fredholm Pfaffian 
\begin{equation}
\mathbb{E}\left[ \exp(-\varsigma \invgamma e^{H(t)}) \right]=1+\sum_{n_s=1}^{+\infty} \frac{(-1)^{n_s}}{n_s!}\prod_{p=1}^{n_s} \int_\mathbb{R} \rmd r_p \frac{\varsigma}{\varsigma+e^{-r_p}}{\rm Pf}\left[ K(r_i,r_j)\right]_{n_s \times n_s}.
\label{eq:fredholmKPZ}
\end{equation}
Here $\invgamma$ is a random variable, with
an inverse gamma distribution of parameter $A+B+1$, see \eqref{pxi}, independent from $H(t)$,
which enters in the construction of the generating function. The kernel 
$K$ is matrix valued and represented by a $2\times 2$ block matrix with elements
\begin{equation}\label{eq:kernelAllA}
\begin{split}
&K_{11}(r,r')=\iint_{C^2} \frac{\mathrm{d}w}{2\I\pi}\frac{\mathrm{d}z}{2\I\pi}\frac{w-z}{w+z}\ratioGamma(w)\ratioGamma(z)\cos(\pi w)\cos(\pi z)e^{ -rw-r'z + t  \frac{w^3+z^3}{3} },\\
&K_{22}(r,r')=\iint_{C^2} \frac{\mathrm{d}w}{2\I\pi}\frac{\mathrm{d}z}{2\I\pi}\frac{w-z}{w+z}\ratioGamma(w)\ratioGamma(z)\frac{\sin(\pi w)}{\pi}\frac{\sin(\pi z)}{\pi}e^{ -rw-r'z + t  \frac{w^3+z^3}{3}},\\
&K_{12}(r,r')=\iint_{C^2}  \frac{\mathrm{d}w}{2\I\pi}\frac{\mathrm{d}z}{2\I\pi}\frac{w-z}{w+z}\ratioGamma(w)\ratioGamma(z)\cos(\pi w)\frac{\sin(\pi z)}{\pi}e^{-rw-r'z + t  \frac{w^3+z^3}{3}},\\
&K_{21}(r,r')=-K_{12}(r',r).
\end{split}
\end{equation}
where the dependence in parameters $A,B$ only appears in the function 
\be \label{GG} 
\ratioGamma(z) = \frac{\Gamma(A+\frac{1}{2}-z)}{\Gamma(A+\frac{1}{2}+z)}\frac{\Gamma(B+\frac{1}{2}-z)}{\Gamma(B+\frac{1}{2}+z)} \Gamma(2z) 
\ee
and the contour $C$ is an upwardly oriented vertical line parallel to the imaginary axis with real part between $0$ and $\min\lbrace A+\frac{1}{2},B+\frac{1}{2},1\rbrace $. The series in \eqref{eq:fredholmKPZ} can also be interpreted as a Fredholm Pfaffian, see Eq.~\eqref{eq:PfaffAllTimesAllA} and Appendix~\ref{app:2D_to_1D}. The kernel \eqref{eq:kernelAllA} has a similar structure as the kernel defining the GSE Tracy-Widom distribution \cite{tracy2005matrix}. It is not entirely obvious that the integrals over the $r_i$ in \eqref{eq:fredholmKPZ} are well-defined, but this is the case. Indeed, one can show that (i) all the entries of $K$ have exponential decay as $r,r'$ go too $+\infty$, using a standard contour shift argument, see e.g. \cite[Lemma 6.4]{barraquand2018half} (ii) all entries of $K$ grow at most polynomially with $\vert r\vert, \vert r'\vert$, which can be shown using a variant of \cite[Lemma 7.11]{barraquand2018half}. 

\subsubsection{Finite time: Result in terms of a scalar kernel} 

The matrix kernel in \eqref{eq:kernelAllA} has the structure of a Schur Pfaffian. Following 
 \cite{krajenbrink2018large} and Appendix~\ref{app:2D_to_1D}, we are able to express the generating function in terms of the Fredholm determinant of a scalar kernel
\begin{equation} \label{res-scalar0} 
\mathbb{E}\left[ \exp(-\varsigma \invgamma e^{H(t)}) \right]
=\sqrt{\mathrm{Det}(I- 
\bar{K}_{t,\varsigma})_{\mathbb{L}^2(\mathbb{R}_+)}}.
\end{equation}
were the kernel $\bar{K}_{t,\varsigma}$ defined for all $(x,y)\in \mathbb{R}_+^2$ as
\begin{equation}\label{eq:1DkernelAllTIntro}
\bar{K}_{t,\varsigma}(x,y)
=2\partial_x \iint_{C^2} \frac{\rmd w \rmd z}{(2\I\pi)^2 }\ratioGamma(z)\ratioGamma(w) \frac{\sin(\pi (z-w))}{\sin(\pi(z+w))}\varsigma^{w+z}e^{-xz-yw  + t  \frac{w^3+z^3}{3} }
\end{equation}
where the function $G(z)$ is defined in \eqref{GG}. Again this formula is valid for
for all time $t \geqslant 0$ and all $A,B > - \frac{1}{2}$.  In principle from formula \eqref{eq:fredholmKPZ} or \eqref{res-scalar0} 
the PDF of $H(t)$ for any time $t$ can be extracted, see e.g. \cite{GueudrePLD}
for the case $A,B=+\infty$. Here, we only extract the PDF's in the large time
limit, as we now discuss.

\begin{remark}
The kernel $\bar{K}_{t,\varsigma}$ can be extended by adding a fictitious variable so that the r.h.s of 
\eqref{res-scalar0} is a $\tau$-function of the Kadomtsev-Petviashvili (KP) equation, see Appendix
\ref{app:KP}.
\end{remark}

\begin{remark}
Other cases where it is possible to transform a matrix-valued kernel into a scalar kernel have been considered in the random matrix literature, see Refs.~\cite[Sections II -- III]{tracy1996orthonormal} and \cite[Section 3.1]{forrester2000painlev}.
\end{remark}

\subsubsection{Large time limit and phase diagram} 

The phase diagram in the $(A,B)$ plane in the large time limit is shown in Fig. 
\ref{fig1}. Qualitatively there are three regions. {In the region $A<-1/2$ with $A<B$, we expect, 
from the results of \cite{deNardisPLDTT}
for $B=+\infty$, 
that the polymer is ``bound to the wall'' and that the (scaled) height distribution
at large time is Gaussian (see also analogous results for other models in \cite{baik2001asymptotics}, \cite[Section 6]{baik2018pfaffian}, \cite[Section 8.1]{barraquand2018half}). By ``bound to the wall'', we mean that the polymer path spends most of its time at the boundary and does not significantly venture into the bulk, this phenomenon was predicted by Kardar \cite{KardarTransition} who studied the depinning of the polymer by the random environment. By symmetry the same can be expected for the region $B<-1/2$ with $B<A$,  which
corresponds to a polymer ``bound to the Brownian''. In the special case where $A=B<-1/2$, the nature of fluctuations is different, there is a competition between the boundary and the initial condition and we expect that the fluctuations have the same distribution as the largest eigenvalue of a $2\times 2$ GUE matrix, based on heuristic arguments presented in Section~\ref{sec:conjecturalidentity}. We have, however,  no exact formula for
the region $A<-1/2$ or $B<-1/2$ called the bound phase. Our exact results concern the third region, the quadrant $A,B \geq -1/2$.}\\

The first result is that for any fixed $A,B > -1/2$, the distribution of the height $H(t)$ converges
at large time to the GSE Tracy-Widom  distribution $F_4$
{\bea \label{GSE1} 
 \lim_{t\to \infty} \mathbb{P}\left(\frac{H(t)}{t^{1/3}}\leqslant s\right)=F_4(s), \quad \quad A,B>-1/2. 
\eea}

When $A=-1/2$, and for any $B > -1/2$ we find that the fluctuations are given the GOE Tracy-Widom  distribution.
{\bea \label{GOE11} 
  \lim_{t\to \infty} \mathbb{P}\left(\frac{H(t)}{t^{1/3}}\leqslant s\right)=F_1(s), \quad \quad A=-1/2, B>-1/2. 
\eea
By symmetry the same holds for $B=-1/2$ and any $A>-1/2$. These results are natural to expect. Indeed, when $B>-1/2$, the initial condition has a drift so negative that the asymptotics of the height function should be the same as for the narrow wedge initial data. The limiting distribution then depends on the value of the boundary parameter $A$ according to the Baik-Rains transition discussed in the Introduction, hence the GSE and GOE Tracy-Widom asymptotics.

\bigskip 

To study the ``critical stationary'' point $(A,B)=(-1/2,-1/2)$ we write\footnote{Note that the parameter $A$ plays the same role as in \cite{halfASEPBarraquand} and \cite{AlexLD} but was denoted $a$ in \cite{borodin2016directed} and $b$ in \cite{deNardisPLDTT}. The parameter that we denote $\Arescaled$ was denoted $\epsilon$ in \cite{AlexLD}.}
\be
A + \frac{1}{2} = t^{-1/3} \Arescaled, \quad  \quad B + \frac{1}{2} = t^{-1/3} \Brescaled,
\ee 
and consider the large time limit at fixed values of $\Arescaled,\Brescaled$. This corresponds to  zooming around the critical stationary point as represented in Fig.~\ref{fig:detailaroundstationarypoint}.
It is natural to expect -- and we indeed show -- that in the large time limit there is a two parameter family of CDFs 
$F^{(\Arescaled, \Brescaled)}(s)$, indexed by $\Arescaled,\Brescaled$ such that
\begin{equation} \label{res00} 
 \lim_{t\to \infty} \mathbb{P}\left(\frac{H(t)}{t^{1/3}}\leqslant s\right) := F^{(\Arescaled, \Brescaled)}(s).  
\end{equation}

\begin{figure}
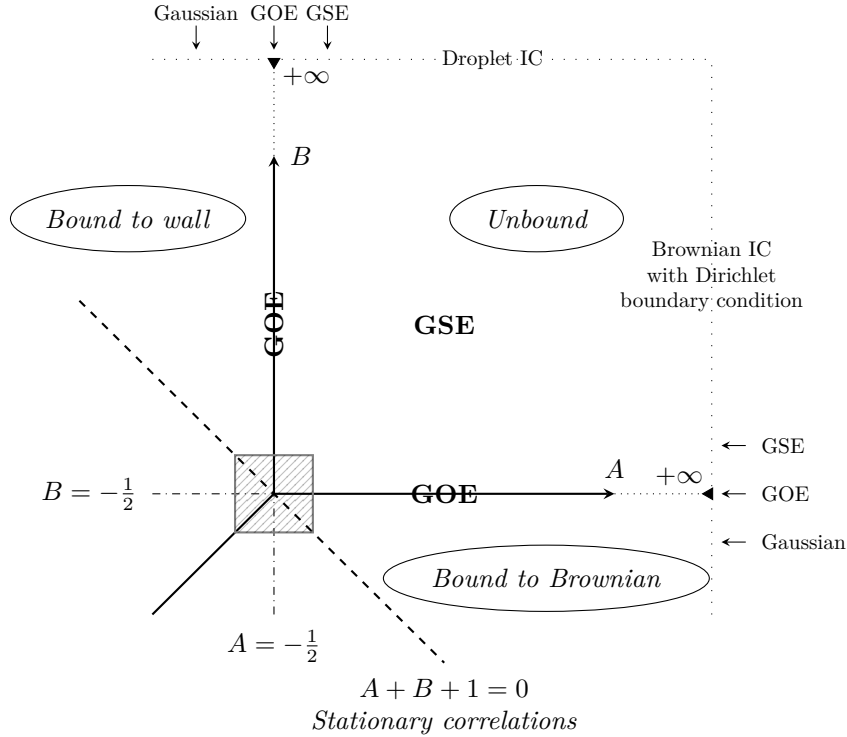

	\begin{center}
\PhaseDiagram
	\end{center}
\caption{Phase diagram indicating the distribution of height fluctuations at large time, as a function of the parameters $A,B$. The nature of fluctuations in the dashed area around $(A,B)=(-1/2,-1/2)$ is explained in Fig.~\ref{fig:detailaroundstationarypoint}.}
\label{fig1} 
\end{figure}

Here, we first obtain a Fredholm Pfaffian formula for the CDF $F^{(\Arescaled, \Brescaled)}(s)$ with $\Arescaled,\Brescaled>0$,
by taking the large time limit of the matrix kernel formula \eqref{eq:fredholmKPZ}.
It reads 
\begin{equation}
F^{(\Arescaled,\Brescaled)}(s) = \left(1+\frac{\partial_s}{\Arescaled+\Brescaled} \right){\rm Pf}(J-K^{(\Arescaled, \Brescaled)})_{\mathbb{L}^2(s,+\infty)}
\label{eq:defFepsiloneta0}
\end{equation}
where the large time matrix kernel $K^{(\Arescaled, \Brescaled)}$ is given
in \eqref{eq:largeTimeK}. Equivalently, a more convenient formula is obtained by taking the large time limit of the scalar kernel formula \eqref{res-scalar0}. We then obtain the CDF of the one-point KPZ height $H(t)$ in the critical region, in terms of a Fredholm determinant, for $\Arescaled,\Brescaled>0$ 
\be
\label{eq:MainResultsLargeTime}
F^{(\Arescaled, \Brescaled)}(s) =\left(1+\frac{\partial_s}{\Arescaled+\Brescaled} \right)\sqrt{\det(I-\bar K^{(\Arescaled, \Brescaled)})_{\mathbb{L}^2(s,+\infty)}}
\ee
where the scalar transition kernel $\bar K^{(\Arescaled,\Brescaled)}$ takes the form 
	\begin{equation}
\bar K^{(\Arescaled, \Brescaled)}(x,y) = \int_{0}^{+\infty} d \la A^{(\Arescaled, \Brescaled)}(x+\la)A^{(\Arescaled, \Brescaled)}(y+\la)\,\mathrm  -\frac 1 2 A^{(\Arescaled, \Brescaled)}(x) \int_{0}^{+\infty}\mathrm d\la \, A^{(\Arescaled, \Brescaled)}(y+\la)\, .
\label{eq:GLDform0}
	\end{equation}
Here the function $A^{(\Arescaled, \Brescaled)}(x)$ is defined by the integral representation
\begin{equation}
A^{(\Arescaled, \Brescaled)}(x) = \int \frac{\mathrm d z}{2\I\pi} \frac{\Arescaled +z}{\Arescaled -z}  \frac{\Brescaled+z}{\Brescaled -z} e^{-xz+ \frac{z^3}{3}},
\label{eq:defA0}
\end{equation} 
	where the contour is a upwardly oriented vertical line with real part between $0$ and $\min \lbrace \Arescaled, \Brescaled \rbrace $. Finally, introducing the operator $\hat A_s$ with kernel $\hat A_s(x,y)= A^{(\Arescaled, \Brescaled)}(x+y+s)$, the final and simplest expression for the cross-over CDF obtained from an algebraic manipulations of \eqref{eq:MainResultsLargeTime} is
\be
F^{(\Arescaled, \Brescaled)}(s) = \frac{1}{2} \left(1 + \frac{\partial_s}{\Arescaled+\Brescaled}\right)   \left(\det(I - \hat A_s)_{\mathbb{L}^2(\mathbb{R}_+)} + \det(I + \hat A_s)_{\mathbb{L}^2(\mathbb{R}_+)}\right)
\ee

It is clear on this formula that if $\Arescaled,\Brescaled \to +\infty$ simultaneously, then $A^{(\Arescaled, \Brescaled)}(x)$
converges to the standard Airy function, and $\bar K^{(\Arescaled, \Brescaled)}$ to the kernel 
associated to the GSE Tracy-Widom  distribution (in the form found in \cite{GueudrePLD}). This thus matches smoothly with the result \eqref{GSE1} valid for any fixed $A,B>-1/2$. Another interesting limit, that we call $F^{(\Arescaled)}(s)= \lim_{\Brescaled \to +\infty} F^{(\Arescaled,\Brescaled)}(s)$,
is the limit $\Brescaled \to + \infty$ at fixed $\Arescaled$, which corresponds to the droplet initial condition
in the critical region for the wall parameter. A formula for that CDF was obtained, for $\Arescaled\geq 0$, using the RBA in \cite{AlexLD}. It was conjectured to coincide with the GSE-GOE-Gaussian crossover introduced by Baik and Rains \cite{baik2001asymptotics}, see also \cite{baik2018pfaffian,baik2018pfaffian22} 
in the context of last passage percolation. This crossover was
also studied in the context of spiked models of random matrices from the GSE \cite{wang2009largest}. 
By the $A \leftrightarrow B$ symmetry,  the case $A\to+\infty$ is similar,  and we obtain the same distribution $F^{(\Brescaled)}(s)= \lim_{\Arescaled \to +\infty} F^{(\Arescaled,\Brescaled)}(s)$, which corresponds to the model with Brownian initial data in presence of an infinite repulsive wall (see \cite{parekh2019kpz123} for a more mathematical interpretation). This is consistent with Tracy-Widom  GOE fluctuations for any fixed 
$A=-1/2$ and $B>-1/2$ $(a=0,b=\infty)$ or fixed $B=-1/2$ and $A>-1/2$ ($b=0, a=\infty)$.

A more difficult limit, which we discuss now, is the stationary critical point 
$(A,B)=(-1/2,-1/2)$ corresponding to both $\Arescaled,\Brescaled \to 0$.

\begin{figure}
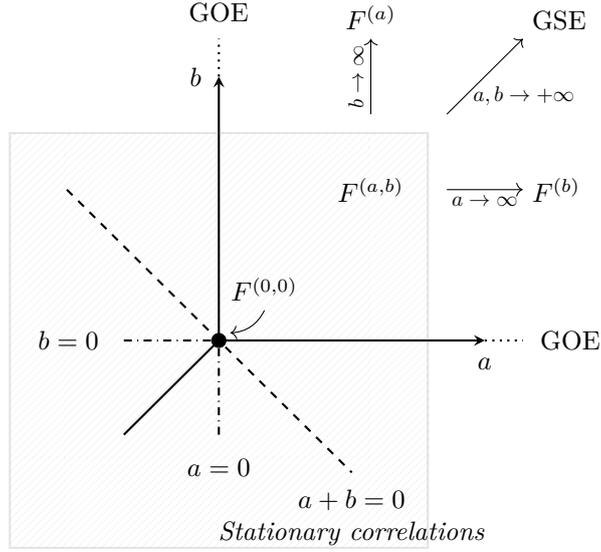

\begin{center}		
\ZoomedDiagram
\end{center}
\caption{Zoom into the vicinity of $(A,B) = (-1/2, -1/2) $. The distribution of height fluctuations at large time  is indicated as a function of parameters $\Arescaled=t^{1/3}(A+\frac 1 2), \Brescaled=t^{1/3}(B+\frac 1 2)$.}
\label{fig:detailaroundstationarypoint}
\end{figure}

\subsection{Stationary critical distribution}

At the point $(A,B)=(-1/2,-1/2)$ we have found a remarkable universal distribution, corresponding
to the CDF $F(s):=F^{(0,0)}(s)$. Taking the limit $\Arescaled, \Brescaled\to 0$ is delicate (as it is also to obtain the Baik-Rains
distribution) and we found that the representation of the kernel $\bar K^{(\Arescaled, \Brescaled)}$ as in \eqref{eq:GLDform0} was crucial. The limit is performed in Section~\ref{sec:criticalstationarycase}.  We have shown that
the limit is well defined, i.e. independent of the ratio $r=\Brescaled/\Arescaled$.  We have obtained the result in several equivalent forms. We recall that we are characterizing the CDF $F(s)$ such that
\be
\lim_{t\to \infty} \mathbb{P}\left(\frac{H(t)}{t^{1/3}}\leqslant s\right) = F(s), \quad \quad A=B=-1/2. 
\label{eq:criticalstationaryintro}
\ee

The first form is in terms of the sum of two Fredholm determinants. Defining the following two kernels acting on functions in $\mathbb{L}^2( \R_+)$, 
\be \label{eqs} 
\Ai_s(x,y) = {\Ai}(s + x + y), \quad \quad \widetilde{\Ai}_s(x,y) = {\Ai}(s + x + y) + \int_0^{+\infty} \rmd \lambda \, 
{\Ai}(s + x + \lambda), 
\ee 
then
\begin{equation}
F(s) = \partial_s  \left[  2 \det(I + \widetilde \Ai_s) + (s-2)\det(I + \Ai_s)\right].
\label{eq:Fdeterminatalformulaintro}
\end{equation}
The second form is expressed in terms of the CDF's of the GOE and GUE Tracy-Widom distributions $F_1$ and $F_2$ respectively,  as
\begin{equation}
F(s)=\partial_s\left[\frac{F_2(s)}{F_1(s)} \int_{-\infty}^s \rmd t \, \frac{F_1(t)^4}{F_2(t)^2} \right]. 
\label{eq:FintermsofF1F2intro}
\end{equation}
which is very reminiscent of the formula for the Baik-Rains distribution for the full space stationary problem (recalled in \eqref{eq:defBaikRains}).
The third form is expressed in terms of the Hastings-McLeod solution $q(s)$ to the Painlev\'e II equation as 
\begin{equation}
F(s)=\partial_s\left[e^{ - \frac{1}{2} \int_s^{+\infty} \mathrm d r [(r-s)q^2(r)-q(r)] } \int_{-\infty}^{s} \rmd r \, e^{-2\int_r^{+\infty} \rmd t \, q(t)} \right]. \label{pasdelabel} 
\end{equation}

The first moments and cumulants are given in the Table~\ref{table:moments of the distribution} and we plot in Fig.~\ref{fig:plots}, the CDF $F$ together with its derivative, the PDF. Finally, we computed and plotted in Appendix~\ref{app:F} the asymptotics of the CDF $F(s)$ 
\begin{itemize}
\item  for large positive $s$
\begin{equation}
\begin{split}
&1 - F(s) \\
& = \frac{s^{3/4}e^{-\frac{2 s^{3/2}}{3}}}{4\sqrt{\pi}} \left[1+\frac{139s^{-3/2}}{48  }-\frac{11423s^{-3}}{4608  }+\frac{3907027s^{-9/2}}{663552  }-\frac{2886147455s^{-6}}{127401984  }+o(s^{-6})\right]
 \end{split}
\end{equation}
\item and for large negative $s$
\begin{equation}
\begin{split}
F(s)&=2^{-203/48}e^{\zeta'(-1)/2} \exp \big[-\frac{\left| s\right| ^3}{24}-\frac{\left| s\right| ^{3/2}}{\sqrt{2}}+\frac{23}{16}\log\abs{s}+\frac{91}{8 \sqrt{2}
   \left| s\right| ^{3/2}}\\
   & \hspace*{5cm} -\frac{3957}{128 \, \abs{s} ^3}+\frac{28717}{128 \sqrt{2} \left| s\right| ^{9/2}}-\frac{469683}{512 \left| s\right|
   ^6}+o(s^{-6})\big]
   \end{split}
\end{equation}
\end{itemize}

\begin{table}[t!]
\begin{center}
\begin{tabular}{p{4cm} p{2cm} p{2cm} p{2cm}  p{2.4cm}}
   \hline
   Distribution &  Mean & Variance & Skewness & Excess kurtosis \\ [1ex]
\hline
&    &  & \\[0.5ex]
Half-space stationary & 0  & 1.649 & 0.266 & 0.134 \\[1ex]
Tracy-Widom $\beta=1$ & $-1.2065\dots$  & $1.6078\dots$& $0.2935\dots $ & $0.1652\dots $ \\[1ex]
Tracy-Widom $\beta=2$ & $-1.7711\dots$  & $0.8132\dots$& $0.2241\dots $ & $0.0934\dots $ \\[1ex]
Tracy-Widom $\beta=4$ & $-2.3069\dots$  & $0.5177\dots$& $0.1655\dots $ & $0.0492\dots $ \\[1ex]
Baik-Rains & 0  & $1.1504\dots $ & $0.3594\dots$ & $0.2892\dots$ \\[1ex]
   \hline 
\end{tabular}
   \caption{{Mean, variance, skewness and excess kurtosis} of the half-space critical stationary distribution and comparison with the Tracy-Widom and Baik-Rains distributions (see \cite[Section 9.4.1]{Bornemann2} and \cite{Halpin2014}).}

   \label{table:moments of the distribution}
\end{center}
\end{table}

As we mentioned above, it remains to be shown that our formula is equivalent to
the Fredholm Pfaffian formula obtained in \cite[Theorem 2.7]{betea2019stationary} (setting there $\delta=0$ and $u=0$) as expected from universality. 

\begin{figure}[t!]
\begin{center}
		\includegraphics[width=7.1cm]{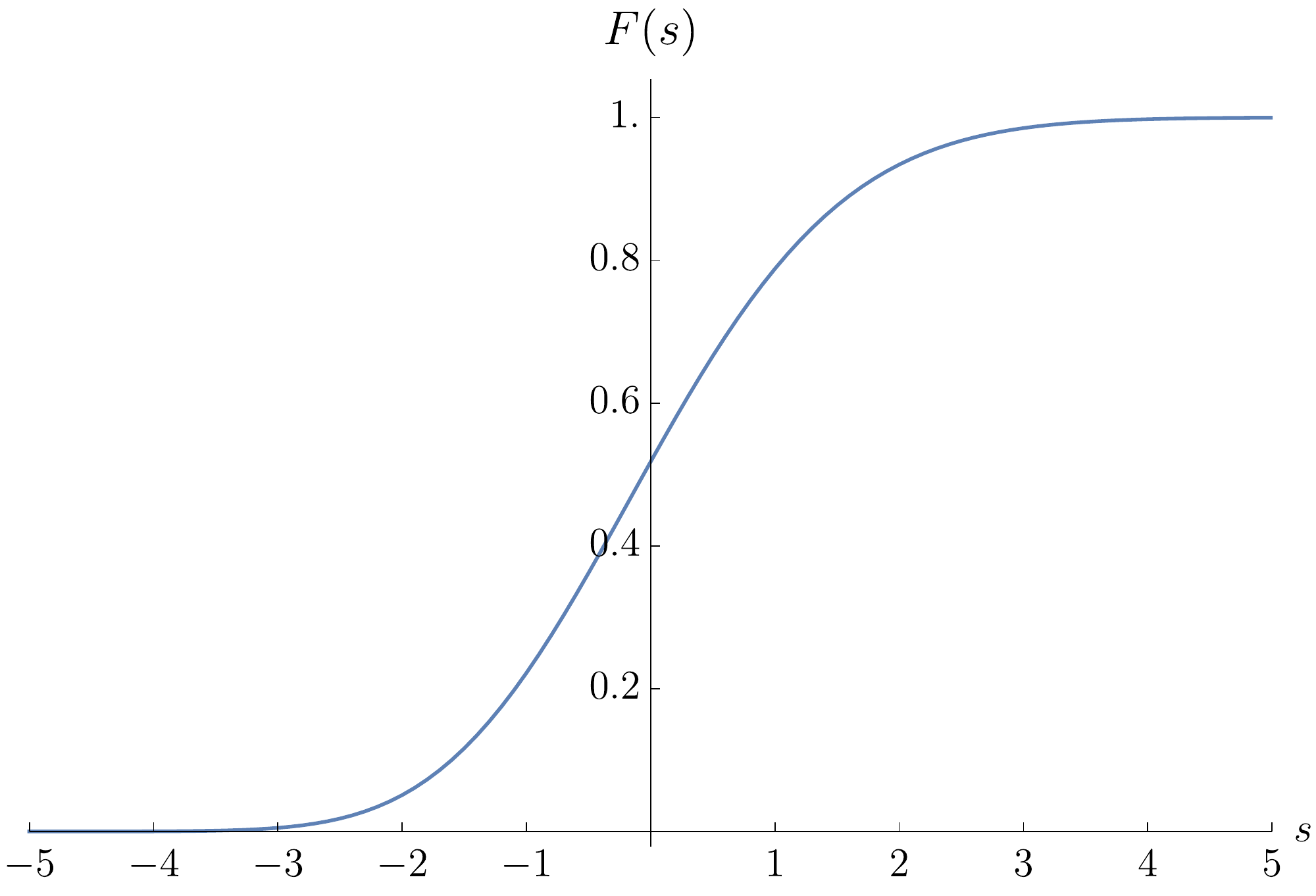}
		\includegraphics[width=7.1cm]{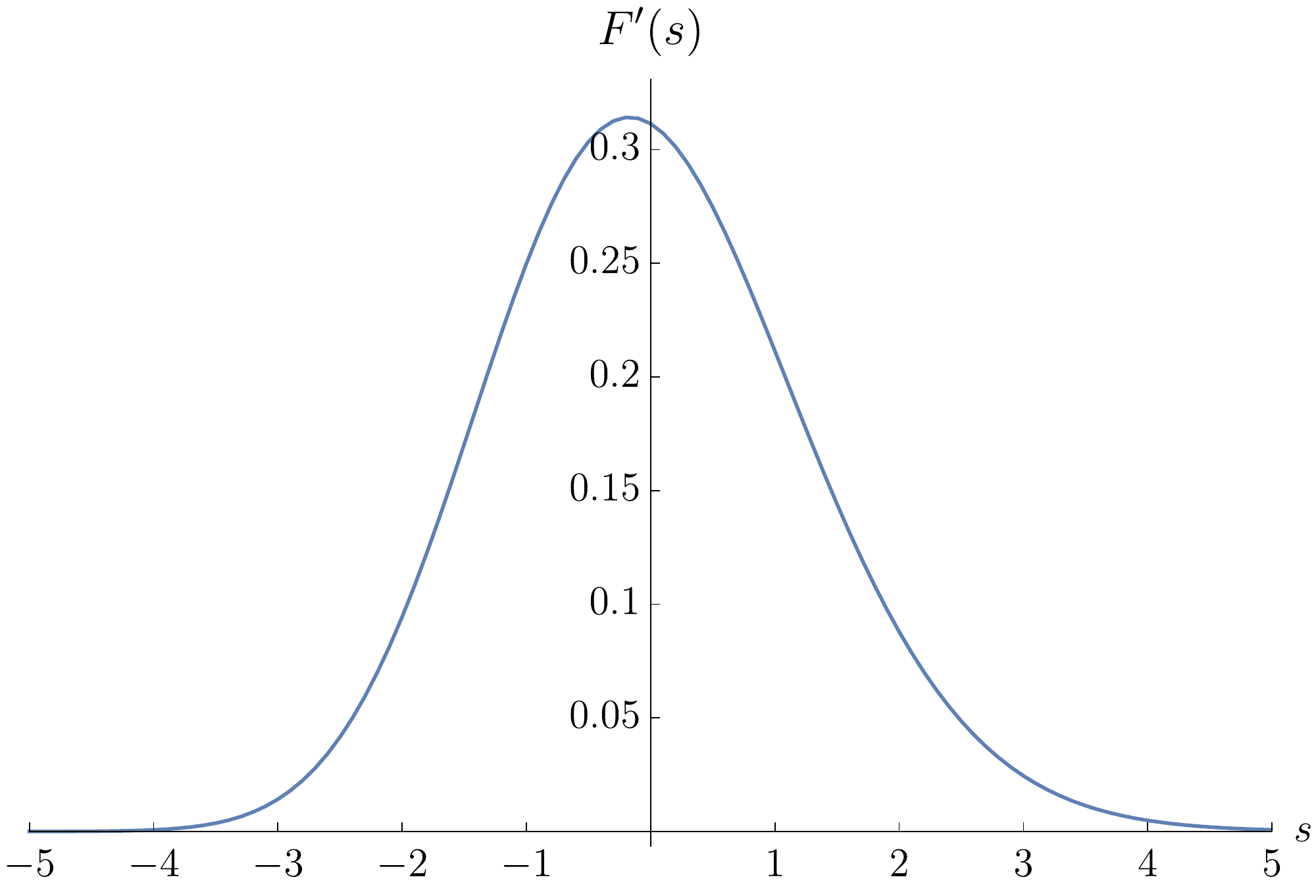}
\end{center}
\caption{\textbf{Left}: Critical stationary CDF $F$. \textbf{Right}: corresponding PDF. See Fig.~\ref{fig:plotsTails} in Appendix.~\ref{app:F} for the comparison with the asymptotics ($s\to \pm \infty$) in true and logarithmic scales.}
\label{fig:plots}
\end{figure}

\subsection{Mathematical aspects}
\label{sec:mathematical}

{The results presented in this article rely on a combination of physics and mathematics methods, but we focus in this article on  physics results and make clear here that most of our results are not proved according to the standards of rigor of the mathematics literature (in particular the results stated as Claims below).} It remains a challenge to turn the arguments that we present here into mathematical theorems. Let us comment further on these aspects for the mathematically inclined reader. 

\medskip 

The first difficulty from the mathematical point of view is that we cannot  rigorously characterize the distribution of the KPZ equation through its moments, because  they grow too fast to uniquely determine the distribution. This is why the moment generating series that we consider in Section~\ref{sec:FP} are actually divergent series, but the formal power series become convergent after certain manipulations and exchanges of series/integrals. For the full-space KPZ equation, it has been proved (see e.g. \cite{borodin2016directed}) that these manipulations lead to the correct answer. It could be possible to overcome this issue in our case by working on a model for which the moment problem is well-defined, and take a scaling limit to the KPZ equation. Such a strategy has been implemented for instance in  \cite{amir2011probability, borodin2014macdonald, borodin2012free, BCFV} for the full-space KPZ equation and in \cite{halfASEPBarraquand} for the half-space KPZ equation with $A=-1/2$ and droplet initial data. Another possible approach is provided by the framework of half-space Macdonald processes \cite{barraquand2018half} which allows to prove Laplace transform formulae despite the divergence of moments. 

\medskip 

The second obstacle is that in order to prove the results of Section~\ref{sec:BA} below, one would need to prove the completeness of the Bethe ansatz eigenfunctions. Actually, we present in Section~\ref{sec:loggamma} another approach to obtain the same moment formulae. It relies on rigorous formulae for the log-gamma polymer from \cite{barraquand2018half}, and we take a scaling limit to the KPZ equation. We obtain a nested contour integral formulae for the moments of $Z_A^B(x,t)$ in Claim~\ref{prop:nestedcontours}. Note that this formula allows to take $x>0$. Then, for $x=0$, we may move the contours together appealing to a combinatorial conjecture from Borodin, Bufetov and Corwin \cite[Conjecture 5.2]{borodin2016directed}, and a Pfaffian structure appears. Hence we see that assuming completeness of Bethe ansatz eigenfunctions or using this conjecture leads to the same moment formula. The results from \cite{borodin2014spectral, borodin2015spectral} suggest that the two problems are indeed related. 

\medskip 

Finally, the asymptotic analysis of Fredholm Pfaffians such as \eqref{eq:fredholmKPZ} is delicate, especially in the critical stationary regime. In particular, we have first performed a large time limit for positive $a,b$, and then let the parameters $a,b$ go to $0$. This allowed us to benefit from the structure of the kernel \eqref{eq:GLDform0}, which eventually lead to a very simple formula for the distribution $F(s)$. It would be interesting to find a generalization of the form \eqref{eq:GLDform0} at finite time, and prove that the limits commute, i.e. one can take first $A,B\to -1/2$ and then study the large time limit.

\section{Moments from the replica Bethe ansatz}

\label{sec:BA} 

\subsection{Quantum mechanics and Bethe ansatz}

In this Section we use the replica Bethe ansatz method to calculate the
integer moments of the partition sum. The equal time multi-point moments of the solution of the SHE, $Z (x,t)$, over the KPZ noise
can be expressed \cite{kardareplica} as a matrix element of the quantum mechanical evolution operator in imaginary time of the Lieb-Liniger model \cite{ll}

\be \label{momentsn} 
\mathbb{E} \left[ Z (x_1,t) \dots Z (x_n,t)\right] = \langle x_1 \dots x_n | e^{- t H_n} |\Psi(t=0) \rangle
\ee 
Here $H_n$ is the Hamiltonian of the Lieb Liniger model \cite{ll} for $n$ quantum particles
with attractive delta function interactions of strength $c=-\bar c <0$
\be
H_n = - \sum_{i=1}^n \partial_{x_i}^2 - 2 \bar c \sum_{1 \leqslant i < j \leqslant n} \delta(x_i-x_j) 
\ee 
with here an below, in our units $\bar c=1$. The initial state $\ket{\Psi(t=0)}$ is such that
\be
\mathbb{E}\left[ Z (x_1,0) \dots Z (x_n,0)\right]= \langle x_1 \dots x_n |  \Psi(t=0) \rangle
\ee
 Since here we are considering the Brownian initial condition
and interested in averages both over the Brownian and the KPZ noise
we must take the initial state $|\Psi(t=0) \rangle$ as
\be
\langle x_1 \dots x_n |  \Psi(t=0) \rangle = \Phi_0(x_1,\dots,x_n) := 
\mathbb{E}_{\mathcal{B}}\left[ \exp\left( \sum_{j=1}^n \mathcal{B}(x_j) -(B+1/2)x_j \right) \right]
\ee 
A simple calculation shows that $\Phi_0(x_1,\dots,x_n)$ is the fully symmetric
function which in the sector $0 \leqslant x_1<\dots \leqslant x_n$ takes the form
\be \label{initial} 
\Phi_0(x_1,\dots,x_n)  = \exp\left(\sum_{j=1}^n \frac{1}{2} (2n-2j+1)x_j -(B+1/2) x_j \right), 
\ee

We can now rewrite \eqref{momentsn} at coinciding points using the decomposition of
the evolution operator $e^{- t H_n}$ in terms of the eigenstates of $H_n$ as
\bea
\label{replicaevolution}
\mathbb{E}\left[Z (x,t)^n\right]
= \sum_\mu \Psi_\mu(x,\dots, x) \langle \Psi_\mu | \Phi_0 \rangle  \frac{1}{||\mu||^2} e^{- t E_\mu} 
\eea 
Here the un-normalized eigenfunctions of $H_n$ are denoted $\Psi_\mu$ (of norm denoted $||\mu||$) with eigenenergies $E_\mu$. Here we used the fact that only symmetric (i.e. bosonic) eigenstates contribute since the initial and final states are fully symmetric in the $x_i$. Hence the $\sum_\mu$ denotes a sum over all bosonic eigenstates of the Lieb-Liniger model,
also called delta Bose gas, and $\langle \Psi_\mu | \Phi_0 \rangle$ denotes the overlap, i.e. the
Hermitian scalar product of the initial state \eqref{initial} with the eigenstate $\Psi_\mu$.\\

We should remember now that $H_n$ is defined on the half-line $x \geqslant 0$.
The boundary condition
at the wall with parameter $A$ translates into the
same boundary condition for the wavefunctions
(in each of their coordinate). This half-line quantum mechanical problem can be solved by the Bethe ansatz
for $A=+\infty$, i.e. for Dirichlet boundary condition\cite{GaudinHardWall,TWhalf,BAhardwall} 
(see also section 5.1 of \cite{gaudin2014bethe}) and this fact was used in \cite{gueudre2012directed}.
It can also be solved for arbitrary $A$,\cite{gaudin2014bethe,Castillo,borodin2016directed,VanDiejen,VanDiejen22,EmsizComplete,GutkinSutherland,HeckmanOpdam}
which led to the moment formula in \cite{PLD1}, \cite{AlexLD} and \cite{deNardisPLDTT}. \\

From the Bethe ansatz the eigenstates $\Psi_\mu$ are thus Bethe states, i.e. 
superpositions of plane waves over all permutations $P$ of the $n$ rapidities $\lambda_j$ for $j\in[1,n]$ with an additional summation over opposite pairs $\pm \lambda_j$ due to the infinite hard wall. 
The bosonic (fully symmetric) eigenstates can be obtained everywhere from their expression in the sector $0\leqslant x_1 \leqslant \dots \leqslant x_n$, which reads
\be \label{wave}
\begin{split}
& \Psi_\mu(x_1,\dots,x_n) = \frac{1}{(2 \I )^{n}} \sum_{P \in S_n} \prod_{p=1}^n \left( \sum_{\varepsilon_p=\pm 1} \varepsilon_p  e^{\I\varepsilon_p x_p \lambda_{P(p)}} A[\varepsilon_1 \lambda_{P(1)},\varepsilon_2 \lambda_{P(2)}, \dots, \varepsilon_n \lambda_{P(n)}]  \right) \\
& A[\lambda_1,\dots,\lambda_n]= \prod_{n \geqslant \ell > k \geqslant 1} \left(1+ \frac{\I \bar c}{\lambda_\ell - \lambda_k}\right)\left(1+ \frac{\I \bar c}{\lambda_\ell + \lambda_k}\right) \prod_{\ell=1}^n \left(1+\I \frac{\lambda_\ell}{A}\right)
\end{split}
\ee
This wavefunction automatically satisfies both
\begin{enumerate}
\item The matching condition arising from the $\delta(x_i-x_j)$ interaction
\begin{equation}
\left(\partial_{x_{i+1}}-\partial_{x_i} +\bar{c}\right) \Psi_\mu(x_1,\dots,x_n)\mid_{x_{i+1}=x_i}=0
\end{equation}
\item The boundary condition $\partial_{x_i}\Psi_\mu(x_1,\dots,x_n)\big\vert_{x_i=0}=A \Psi_\mu(x_1,\dots,x_n)\big\vert_{x_i=0}$ for all $i \in [0,n]$.
\end{enumerate}
The allowed values for the rapidities $\lambda_i$, which parametrize the true physical eigenstates are determined by the Bethe equations arising from the boundary conditions at $x=L$ as discussed below. One will find that the normalized eigenstates $\psi_\mu=\Psi_\mu/||\mu||$ vanish as $(\lambda_i-\lambda_j)$ or $(\lambda_i+\lambda_j)$ when two rapidities become equal or opposite: hence the rapidities obey an exclusion principle.\\

The detailed Bethe equations, which determine the allowed values for the set of rapidities $\lbrace \lambda_j \rbrace $, depend on the choice of boundary condition at $x=L$. However, in the $L\to +\infty$ limit, these details do not matter. For simplicity we choose a hardwall at $x=L$. The Bethe equations then read
\begin{equation}
e^{2\I \lambda_j L}=\frac{A-\I \lambda_j}{A+\I \lambda_j}  \prod_{\ell \neq j}\frac{\lambda_j-\lambda_\ell-\I \bar{c}}{\lambda_j-\lambda_\ell+\I \bar{c}}\frac{\lambda_j+\lambda_\ell-\I \bar{c}}{\lambda_j+\lambda_\ell+\I \bar{c}}
\end{equation}

In the case of the infinite hardwall, these equations are also given in  Ref.~\cite{BAhardwall} and their solutions in the large $L$ limit were studied in  Ref.~\cite{ChineseBA}. The structure of the states for infinite $L$ is found similar to the standard case, i.e. the general eigenstates are built by partitioning the $n$ particles into a set of $n_s$  bound-states formed by $m_j \geqslant 1$ particles with $n=\sum_{j=1}^{n_s} m_j$.
Each bound state $\mu$ is indexed by a set of $\{ k_j, m_j\}_{j=1\dots n_s}$ where the $k_j$'s are real numbers. These states are {\it perfect strings} \cite{m-65} , i.e. a set of
rapidities 
\begin{equation}
\lambda^{j, a}=k_j +\frac{\I\bar c}2(m_j+1-2a)
\end{equation}
 where $a = 1,\dots ,m_j$ labels the rapidities within the string.  Such eigenstates have momentum and energy 
\begin{equation}
K_\mu=\sum_{j=1}^{n_s} m_j k_j, \qquad E_\mu=\sum_{j=1}^{n_s} m_j k_j^2-\frac{\bar c^2}{12} m_j(m_j^2-1).
 \end{equation}
 The ground-state corresponds to a single $n$-string with $k_1=0$. The difference with the standard case is that the states are now invariant by a sign change of any of the momenta $\lambda_j \to -\lambda_j$, i.e. $k_j \to -k_j$. From now on, we will denote the wavefunctions of the string states as $\Psi_{\{k_\ell,m_\ell\}}$.\\
 
It is important to note that although for $A=+\infty$ the strings are the only solutions of the Bethe equations at large $L$, for finite $A$ there are other solutions which correspond to so-called boundary bound states.
These solutions have been obtained and studied in details in \cite{deNardisPLDTT}. As we will see below we will not need them in this work.
 
\subsection{Moment formula} 

To calculate the $n$-th moments of $Z(x,t)$ from formula \eqref{replicaevolution}, we need to perform a summation over the eigenstates. For $A<+\infty$ these eigenstates contain both the string states and the boundary bound states mentioned above. Our strategy here will be similar to the one in 
\cite{AlexLD}, i.e. we will calculate the moments for $n < 2 A + 1$ 
which turn out to be sufficient to perform the analytic continuation in $n$ and obtain the
generating function for any $A> -1/2$, using a method
similar to the one in \cite{SasamotoStationary}. The nice feature is that when $n < 2 A+1$
there are no boundary bound states. To see that, consider the Table 1 in \cite{deNardisPLDTT} which contains the classification of the boundary bound states for this problem.
For $A > -1/2$, a $m$ particle boundary bound state must obey $m \geq \floor{2 A}+2 \geq 2A +1$.
On the other hand from $n < 2 A+1$, one must have $m \leq n < 2 A+1$, which excludes the bound state.
Hence we need to consider only the string states.

A formula for the inverse of the squared norm of an arbitrary string state was obtained
for $A< +\infty$ in \cite{PLD1} and \cite{deNardisPLDTT}, consistent with the results of
\cite{AlexLD}, as
\be
\begin{split}
& \|\mu\|^2  := \int_0^L \rmd x_1\dots \int_0^L \rmd x_n |\Psi_{\{k_\ell,m_\ell\}}(x_1,\dots, x_n)|^2 \\
\label{normformula}
& \frac{1}{||\mu||^2} = \frac{1}{n!} \bar c^{n-n_s} 2^{n_s}
\prod_{i=1}^{n_s} S_{k_i,m_i} H_{k_i,m_i}\prod_{1 \leqslant i<j \leqslant n_s} D_{k_i,m_i,k_j,m_j}  L^{-n_s} \\
& D_{k_1,m_1,k_2,m_2} =\left(
\frac{4 (k_1-k_2)^2 + (m_1-m_2)^2 c^2}{4 (k_1-k_2)^2 + (m_1+m_2)^2 c^2}\right) \times \left(\frac{4 (k_1+k_2)^2 + (m_1-m_2)^2 c^2}{4 (k_1+k_2)^2 + (m_1+m_2)^2 c^2}\right)\\
& S_{k,m} =  \frac{2^{2m-2}}{m^2}  \prod_{p=1}^{[m/2]} \frac{4 k^2 + c^2 (m-2 p)^2}{4 k^2+c^2 (m+1-2 p)^2}  \\
& H_{k,m} =  \prod_{a=1}^m \frac{A^2}{A^2+(k+\frac{\I \bar{c}}{2}(m+1-2a))^2}
\end{split}
\ee
with $S_{k,1}=1$. Note that we have only kept the leading term in $L$ as $L\to +\infty$. Inserting the norm formula \eqref{normformula} into \eqref{replicaevolution}, we obtain the starting formula for the integer moments of the partition sum with Brownian weight on the endpoint in the limit $L \to +\infty$
\be \label{start0}
\begin{split} 
\mathbb{E}\left[Z (x,t)^n\right]  & =  \sum_{n_s=1}^n \frac{ 2^{n_s} \bar c^n }{n_s! \bar c^{n_s} n! } 
\prod_{p=1}^{n_s} \sum_{m_p \geqslant 1}  \int_\mathbb{R} \frac{\rmd k_p}{2 \pi}  m_p S_{k_p,m_p}H_{k_p,m_p} e^{ (m_p^3-m_p) \frac{\bar c^2 t}{12} - m_p k_p^2 t }   \\
& \times \delta_{n,\sum_{j=1}^{n_s} m_j}
 \prod_{i<j}^{n_s} D_{k_i,m_i,k_j,m_j} \Psi_{\{k_\ell,m_\ell\}}(x,\dots, x) \, \langle  \Psi_{\{k_\ell,m_\ell\}} |
 \Phi_0 \rangle
 \end{split}
 \ee
Here the Kronecker delta enforces the constraint $\sum_{j=1}^{n_s} m_j=n$ with $m_j \geqslant 1$ and in the summation over states we used $\sum_{k_j} \to m_j L \int_\mathbb{R} \frac{\rmd k}{2\pi}$ which holds also here in the large $L$ limit: the momenta sums become continuous and one can use that the string momenta $m_j k_j$ correspond to free particles as in Refs.~\cite{CLR10,CLDflat2,CLDflat,gueudre2012directed,deNardisPLDTT}.\\

 We can simplify the factor $\Psi_{\{k_\ell,m_\ell\}}(x,\dots,x)$ in \eqref{start0}. For the
general Bethe state \eqref{wave} (before insertion of the string solution), the $x=0$ limit then reads
reads
\bea
&& \Psi_{\mu}(0,\dots,0) =  \frac{n!}{A^n} \prod_{j=1}^n \lambda_j 
\eea
Inserting the string solution we see that we can replace in \eqref{start0} at the wall $x=0$
\bea  \label{numerator} 
&&  \Psi_{\{k_\ell,m_\ell\}}(0,\dots,0)  =   \frac{n!}{A^n} \prod_{j=1}^{n_s} A_{k_j,m_j}   \\
&& A_{k,m} = \prod_{a=1}^m \left(k+ \I \frac{\bar c}{2} (m+1-2 a)\right) =
(-\I \bar c)^m \frac{\Gamma(\frac{1+m}{2} + \frac{\I k}{\bar c} )}{\Gamma(\frac{1-m}{2} + \frac{\I k}{\bar c} )}  
\eea

To obtain the $n$-th moment in \eqref{start0} we
 still need to calculate the overlap $\langle  \Psi_{\{k_\ell,m_\ell\}} |
 \Phi_0 \rangle$ where $\Phi_0$ is given in \eqref{initial}. In general it involves sums over permutations and leads to complicated expressions but in our case, a simple structure emerges akin to the one known in full-space for a few initial conditions (droplet, half-flat, Brownian). Here, as we find in the Appendix~\ref{app:overlap}, the  result in the half-space for Brownian initial conditions is quite simple
\begin{equation} \label{over0}
\langle  \Psi_\mu | \Phi_0 \rangle =\frac{n!}{A^n}  \frac{\Gamma(A+B+1)}{\Gamma(A+B-n+1)} \prod_{j=1}^n \frac{\lambda_j}{B^2+\lambda_j^2}
\end{equation}
This holds under the condition that the integral converge, that is $\frac{n}{2} < B+ \frac{1}{2}$, which
we will also assume from now on. 
Inserting the rapidities $\lambda_j$ of the string state one see that the denominator in the product in the overlap \eqref{over0} read
\begin{equation}
\begin{split}
E_{k,j}&=\prod_{a=1}^{m} \frac{1}{B^2+(k+\I \frac{\bar c}{2}(m+1-2a))^2}
\\
&=\frac{1}{\bar{c}^{2m}}\frac{\Gamma(\frac{1-m}{2}+\frac{B+\I k}{\bar{c}})}{\Gamma(\frac{1+m}{2}+\frac{B+\I k}{\bar{c}})}\frac{\Gamma(\frac{1-m}{2}+\frac{B-\I k}{\bar{c}})}{\Gamma(\frac{1+m}{2}+\frac{B-\I k}{\bar{c}})}
\end{split}
\end{equation}
while the numerator was already calculated in \eqref{numerator}. We can thus define 
$C_{k,j}=A^2_{k,j}E_{k,j}$ and putting all together we obtain the 
starting expression for the integer moments, denoting here and below $Z(0,t)=Z(t)$
\begin{equation}
\begin{split}
 \mathbb{E}&\left[Z (t)^n\right] =  \frac{\Gamma(A+B+1)}{\Gamma(A+B-n+1)} \sum_{n_s=1}^n \frac{2^{n_s} \bar c^n n! }{n_s! \bar c^{n_s} } 
\\
&\times \prod_{p=1}^{n_s}  \sum_{m_p \geqslant 1}\int_\mathbb{R} \frac{\rmd k_p}{2 \pi}m_p  C_{k_p,m_p}  S_{k_p,m_p}H_{k_p,m_p} e^{ (m_p^3-m_p) \frac{\bar c^2 t}{12} - m_p k_p^2 t }   \delta_{n,\sum_{j=1}^{n_s} m_j} \,  \prod_{i<j}^{n_s} D_{k_i,m_i,k_j,m_j} 
\end{split}
\end{equation}
where we recall the constraint $\sum_{j=1}^{n_s} m_j=n$. Let us
use $\bar c=1$ from now on. Denoting 
\begin{equation} \label{B} 
\begin{split}
B_{k,m} &= 4 m^2 C_{k,m} S_{k,m} H_{k,m}\\
&=\frac{2 k}{\pi} \sinh(2 \pi k) \Gamma(2 \I k +m) \Gamma(-2 \I k + m)\\
& \times  \frac{\Gamma(\frac{1-m}{2}+A+\I k)}{\Gamma(\frac{1+m}{2}+A+\I k)}\frac{\Gamma(\frac{1-m}{2}+A-\I k)}{\Gamma(\frac{1+m}{2}+A-\I k)} \frac{\Gamma(\frac{1-m}{2}+B+\I k)}{\Gamma(\frac{1+m}{2}+B+\I k)}\frac{\Gamma(\frac{1-m}{2}+B-\I k)}{\Gamma(\frac{1+m}{2}+B-\I k)}
\end{split}
\end{equation}
The starting formula for the moments is then 

 \be \label{znn}
 \begin{split}
& \mathbb{E}\left[Z (t)^n\right]=\frac{\Gamma(A+B+1)}{\Gamma(A+B-n+1)} \\
&   \sum_{n_s=1}^n \frac{   n! 2^{n_s} }{n_s! }  
\prod_{p=1}^{n_s} \sum_{m_p \geqslant 1} \int_\mathbb{R} \frac{\rmd k_p}{2 \pi} \frac{B_{k_p,m_p}}{4 m_p}  e^{ (m_p^3-m_p) \frac{ t}{12} - m_p k_p^2 t } \delta_{n,\sum_{j=1}^{n_s} m_j}
 \prod_{i<j}^{n_s} D_{k_i,m_i,k_j,m_j} 
 \end{split}
 \ee
where $B_{k,m}$ is given in \eqref{B} and $D_{k_i,m_i,k_j,m_j}$ is given in \eqref{normformula} and where we recall the constraint $\sum_{j=1}^{n_s} m_j=n$.

\subsection{Decorated moments}

As in Refs.~\cite{SasamotoStationary,SasamotoStationary2}, it is useful to 
eliminate the Gamma factor $\frac{\Gamma(A+B+1)}{\Gamma(A+B-n+1)}$ in 
\eqref{znn}. To this aim we introduce a random variable $\invgamma \sim \mathrm{Gamma}^{-1}(A+B+1)$, 
independent of the KPZ height, which is inverse gamma distributed with parameter $A+B+1$, in this case 
 \begin{equation} \label{pxi} 
p_\invgamma(x)=\frac{1}{\Gamma(A+B+1)}x^{-A-B-2}e^{-\frac{1}{x}}\Theta(x)
 \end{equation}
 The $n$-th moment of $\invgamma$ is given by
 \begin{equation}
\mathbb{E}[\invgamma^n]=\frac{\Gamma(A+B-n+1)}{\Gamma(A+B+1)}
 \end{equation}
 As we will see $ \mathbb{E}\left[\invgamma^n Z (t)^n\right]$ will serve as the basis to form a Fredhom Pfaffian.
 \subsection{Moments in terms of a Pfaffian}

An important identity, which makes the problem solvable in the end, is that the inverse norms of the states
can be expressed as a Schur Pfaffian. Introducing the reduced variables $X_{2p-1} = m_p + 2 \I k_p $ and $X_{2p} = m_p - 2 \I k_p$ for $p\in [1,n_s]$, the norm reads 
\be \label{pfid}
 \prod_{1 \leqslant i<j \leqslant n_s} D_{k_i,m_i,k_j,m_j}  = \prod_{j=1}^{n_s} \frac{m_j}{2 \I k_j} \underset{2 n_s \times 2 n_s}{\rm Pf} \left[ \frac{X_i-X_j}{X_i+X_j} \right]
\ee
where we recall that the Pfaffian of an anti-symmetric matrix $A$ of size $N \times N$ is defined by 
\begin{equation}
{\rm Pf}(A)=\sqrt{{\rm Det}(A)}=\sum_{\substack{\sigma\in S_N,\\ \sigma(2p-1)<\sigma(2p)}}{\rm sign}(\sigma)\prod_{p=1}^{N/2}A_{\sigma(2p-1),\sigma(2p)}
\end{equation}
and that the Schur Pfaffian is given by (see Ref.~\cite{Knuth_1995})
\begin{equation}
{\rm Pf} \left[ \frac{X_i-X_j}{X_i+X_j} \right]=\prod_{i<j}\frac{X_i-X_j}{X_i+X_j}\, .
\end{equation}
Hence the starting formula for the moments now becomes:
\be \label{znn2}
\begin{split}
&\mathbb{E} \left[\invgamma^n Z (t)^n\right]  =\\
&= \sum_{n_s=1}^n \frac{  n! }{n_s! }  
\prod_{p=1}^{n_s} \sum_{m_p \geqslant 1} \int_\mathbb{R}\frac{\rmd  k_p}{2 \pi} \frac{B_{k_p,m_p}}{4 \I k_p}   e^{ (m_p^3-m_p) \frac{ t}{12} - m_p k_p^2 t } \delta_{n,\sum_{j=1}^{n_s} m_j}
~ \underset{2 n_s \times 2 n_s}{\rm Pf}\left[\frac{X_i-X_j}{X_i+X_j} \right]
\end{split}
\ee

\section{Moments from the log-gamma polymer}
\label{sec:loggamma}
In this section we compute again the moments of the solution of the SHE, $Z(x,t)$, taking
a limit of a known formula for the moments of the partition function of the log-gamma polymer
on the half-quadrant square lattice.  This method  uses the convergence of the log-gamma polymer to the KPZ equation at high temperature and  a combinatorial conjecture of Borodin-Bufetov-Corwin \cite[Conjecture 5.2]{borodin2016directed}. We also discuss in Section~\ref{sec:symmetry} and Section~\ref{sec:conjecturalidentity} useful identities in distribution coming from symmetries in so-called half-space Macdonald processes \cite{barraquand2018half}. 

\subsection{Moment formula for the log-gamma polymer}

\begin{definition}[Half-space log-gamma polymer]
\label{def:loggammahalfspace}
Let  $\diag, \alpha_1, \alpha_2, \dots $ be real parameters  such that $\alpha_i+\diag>0$ for all $i \geqslant 1$ and $\alpha_i+\alpha_j>0$ for all $i\neq j\geqslant 1$. The half-space log-gamma polymer is a probability measure on up-right paths confined in the half-quadrant $ \lbrace (i,j)\in \Z_{>0}^2 : i\geqslant j \rbrace $ (see Figure~\ref{halfspaceloggamma}), where the probability of an admissible path $\pi$ between $(1,1)$ and $(n,m)$ is given by
	$$ \frac{1}{\mathcal Z(n,m)} \ \ \prod_{(i,j)\in \pi} w_{i,j},$$
	and where $\big(w_{i,j}\big)_{i\geqslant j}$ is a family of independent random variables such that for $i>j, w_{i,j}\sim \mathrm{Gamma}^{-1}(\alpha_i + \alpha_j)$ and $w_{i,i}\sim \mathrm{Gamma}^{-1}(\diag + \alpha_i)$. The notation  $\mathrm{Gamma}^{-1}(\theta)$ denotes the inverse of a Gamma distributed random variable with shape parameter $\theta$.
	The partition function $\mathcal Z(n,m)$ is given by
	$$\mathcal Z(n,m) = \sum_{\pi: (1,1) \to (n,m)} \prod_{(i,j)\in \pi} w_{i,j}. $$
\end{definition}
\begin{figure}
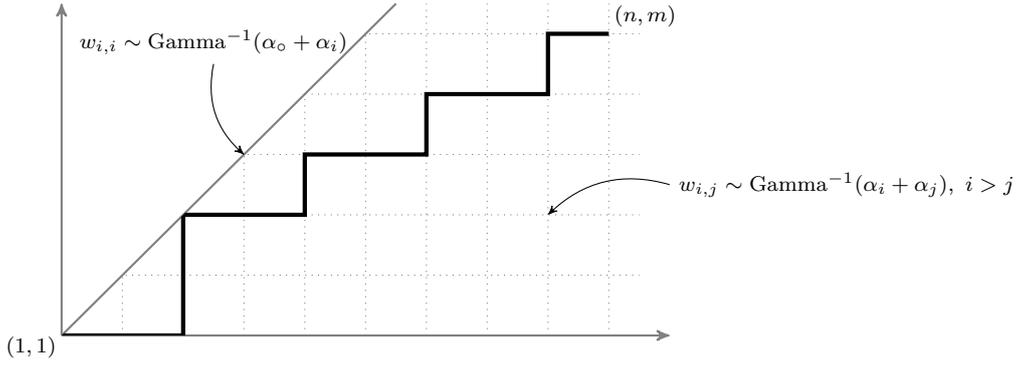

\begin{center}
\LogGammaPoly
\end{center}
	\caption{An admissible path in the half space log-gamma polymer model, that is a path proceeding by unit steps rightward and upward in the half quadrant $ \lbrace (i,j)\in \Z_{>0}^2 : i\geqslant j \rbrace $.}
	\label{halfspaceloggamma}
\end{figure}

The moments of the partition function were computed using half-space Macdonald processes in \cite{barraquand2018half}. 
\begin{proposition}[{\cite[Corollary 6.36]{barraquand2018half}}]
	\label{prop:momentsloggamma}
For $n\geqslant m$ and  $k\in\Z_{>0}$ such that  $k<  \min\lbrace \alpha_j+\frac 1 2 , \alpha_i+\diag \rbrace $,
\begin{multline}
\EE[\mathcal Z(n,m)^k] =\oint\frac{\mathrm{d}z_1}{2\I\pi}\cdots \oint\frac{\mathrm{d}z_k}{2\I\pi} \prod_{1\leqslant a<b\leqslant k} \frac{z_a-z_b}{z_a-z_b-1}\, \frac{z_a+z_b}{1+ z_a+z_b}\\ \times
\prod_{i=1}^{k}  \frac{2z_i}{z_i-\diag+1/2}\prod_{j=1}^n \left( \frac{1}{\alpha_j-z_i-1/2} \right) \prod_{j=1}^{m} \left(\frac{1}{z_i + \alpha_j-1/2}\right),
\label{eq:momentsZ}
\end{multline}
where the contours are such that for all $1\leqslant c\leqslant k$, the contour for $z_c$ encloses $\lbrace - \alpha_j+1/2\rbrace_{1\leqslant j\leqslant m}$  and $\lbrace z_{c+1}+1, \dots, z_k+1\rbrace$, and excludes the poles at $ \diag- 1/2$,  $ \alpha_j-1/2$ (for $1\leqslant j\leqslant n$) and $-1-z_{j}$ (for $j\neq c$).  Because the integrand decays at least quadratically at infinity, one may chose the contours to be all vertical lines such that the contour for the variable $z_i$ is $r_i+\I\R$ where 
$$\max_j\lbrace k - \alpha_j-1/2 \rbrace  < r_k+k-1 < \dots <r_2+2 < r_1 < \min_j \lbrace \diag-1/2, \alpha_j-1/2, 0 \rbrace.$$
\end{proposition}
Note that if $k>\alpha_i+\alpha_j$ or $k>\alpha_i+\diag$ for some $i<j$, the $k$-th moment of $Z(n,m)$ fails to exist.
\begin{remark}
One may also compute mixed moments of the partition functions at several points along a down-right path. 
\end{remark}

\subsection{Stationary structure for the log-gamma polymer} 
\label{sec:stationarystructure}

In this paragraph, we will need to assume $\diag +\alpha_1=0$. In order to do so, we need to consider a modified partition function where we have removed the weight $w_{1,1}$, i.e. we define 
\begin{equation}
\mathcal Z^{stat}(n,m) = \frac{\mathcal Z(n,m)}{w_{1,1}}. 
\end{equation}

Following \cite{seppalainen2012scaling}, we define horizontal and vertical increments of the partition function as 
\begin{equation}
U_{n,m} = \frac{\mathcal Z^{stat}(n,m)}{\mathcal Z^{stat}(n-1,m)}, \;\;V_{n,m} = \frac{\mathcal Z^{stat}(n,m)}{\mathcal Z^{stat}(n,m-1)}.
\end{equation}
The partition function satisfies the recurrence
\begin{equation}
\mathcal Z^{stat}(n,m) = w_{n,m}(\mathcal Z^{stat}(n-1,m)+\mathcal Z^{stat}(n,m-1)),
\end{equation}
where by convention we have assumed that $\mathcal Z^{stat}(n,m)=0$ if $(n,m)$ does not belong to the half quadrant $ \lbrace (i,j)\in \Z_{>0}^2 : i\geqslant j \rbrace $. From there, one may deduce a recurrence for the increments 
\begin{equation}
U_{n,m} = w_{n,m}\left( 1 +\frac{U_{n,m-1}}{V_{n-1,m}} \right),\;\;\;V_{n,m} =  w_{n,m}\left( 1 +\frac{V_{n-1,m}}{U_{n,m-1}} \right).
\label{eq:recurrenceforincrements}
\end{equation}
We will need the following lemma from \cite{seppalainen2012scaling} where the stationary structure for the full-space log-gamma polymer was introduced. 
\begin{lemma}[{\cite[Lemma 3.2]{seppalainen2012scaling}}]
Let $U,V,w$ be independent random variables. Let 
\begin{equation}
U'=w\left(1+\frac{U}{V} \right), \;\; V'=w\left(1+\frac{V}{U} \right), \;\; w'=\left( \frac{1}{U} +\frac{1}{V}\right)^{-1}. 
\end{equation}
If for some  $\alpha>0$ and $\theta\in (-\alpha, \alpha)$,  $U\sim \mathrm{Gamma^{-1}}(\alpha+\theta)$, $V\sim \mathrm{Gamma^{-1}}(\alpha-\theta)$, $w\sim \mathrm{Gamma^{-1}}(2\alpha)$, then the triples $(U,V,w)$ and $(U', V', w')$ have the same distribution. 
\label{lem:seppalainenstationary}
\end{lemma}
Coming back to our model $\mathcal Z^{stat}(n,m)$, when  $\diag+\alpha_1=0$ the model is stationary in the following sense. 
\begin{proposition}
Let $k\in \Z_{\geqslant 1}$. Assume that $\alpha_2=\alpha_3=\dots = \alpha>0$ and $\diag+\alpha_1=0$. Consider a down-right path in the lattice going through the points $\lbrace (n_i, m_i)\rbrace_{1\leqslant i\leqslant k}$, such that $(n_{i+1},m_{i+1})-(n_i,m_i)$  equals either $(0,-1)$ or $(1,0)$ (see Fig. \ref{fig:transfopaths}). We associate to this down-right path increments $\left\lbrace I_j \right\rbrace_{1\leqslant j\leqslant k-1}$ where 
$$I_j= \begin{cases}U_{n_{j+1}, m_{j+1}} \text{ when }n_{j+1}>n_{j},\\ 
V_{n_j, m_j} \text{  when } m_j>m_{j+1}.\end{cases}$$ 

Then the increments $\left\lbrace I_j \right\rbrace_{1\leqslant j\leqslant k-1}$ are all independent and distributed as $I_j\sim  \break \mathrm{Gamma^{-1}}(\alpha_1+\alpha)$ when $I_j$ is a horizontal $U$ increment, and  $I_j\sim  \mathrm{Gamma^{-1}}(\diag+\alpha)$ when $I_j$ is a vertical $V$ increment.  In particular, for any $m$, the increments $\left\lbrace U_{n,m} \right\rbrace_{n\geqslant m+1}$ are independent and distributed as $U_{n,m}\sim \mathrm{Gamma^{-1}}(\alpha_1+\alpha)$.
\label{prop:stationarity}
\end{proposition}

\begin{proof} 
The distribution of increments along the first row is completely constrained by the definition of the model. 
Indeed, we have that $\mathcal Z^{\rm stat}(n,1) = \prod_{i=2}^n w_{i,1}$, so that the increments along the first row are given by $U_{n,1}=w_{n,1}$ and the definition of the model implies that weights $w_{n,1}\sim \mathrm{Gamma^{-1}}(\alpha_1+\alpha)$ are independent.  Hence, for $m=1$, the increments $\left\lbrace U_{n,m}\right\rbrace_{n\geqslant 2}$ are independent and distributed as $\mathrm{Gamma^{-1}}(\alpha_1+\alpha)$ as claimed.  

\medskip 

In other terms, we have seen that the statement of the Proposition is true for the infinite path going through the points $(n,1)$ for all $n\geqslant 1$. We will show that the property is preserved under two types of local transformation of paths, depicted in Fig. \ref{fig:transfopaths}, that consist in 
\begin{enumerate}
	\item (boundary update) Lifting one unit upwards the starting point of the path along the boundary;
	\item (bulk update) Transforming a down-right step into a right-down step.  
\end{enumerate} 
\begin{figure}
\begin{center}
	      \begin{tikzpicture}[scale=0.8]
	\draw[->, thick, >=stealth', gray] (0,0) -- (10, 0);
	\draw[->, thick, >=stealth', gray] (0,0) -- (0, 5.5);
	\draw[thick, gray]  (0,0) -- (5.5,5.5);
	\clip (0,0) -- (9.5,0) -- (9.5,5.5) -- (5.5, 5.5) -- (0,0);
	\draw[dotted, gray] (0,0) grid (10,6);
	\draw[ultra thick] (4,4) --(5,4) -- (6,4) -- (6,3) -- (7,3) -- (7,2) -- (7,1) -- (8,1) -- (9,1) -- (9,0) -- (10,0);
	\draw[ultra thick, red] (5,5) -- (5,4);
	\draw[ultra thick, red] (7,2) -- (8,2) -- (8,1);
	\draw[->, red, thick] (7.2,1.2)  to[bend left] (7.8,1.8) ;
	\draw[->, red, thick] (4.3,4.1)  to[bend right] (4.9,4.7) ;
	\end{tikzpicture}
\end{center}
\caption{The two types of elementary local transformations of down-right paths considered in the proof of Proposition \ref{prop:stationarity}. The thick black path represents an arbitrary down-right path. The portions in red represent the local modifications of the path that we consider.}
\label{fig:transfopaths}
\end{figure}
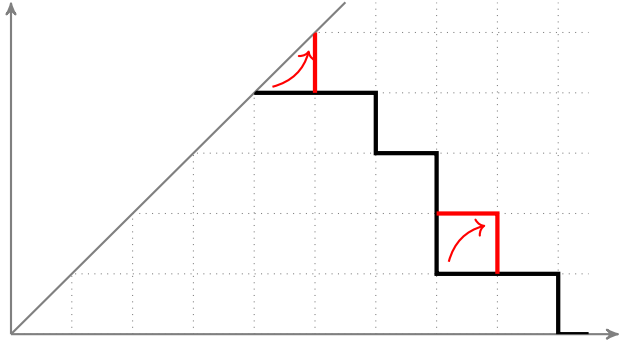
It is clear that any infinite down-right path can  be obtained by iteration of these local transformations, starting from the path going through the points $(n,1)$ for all $n\geqslant 1$. Furthermore, if the statement of the Proposition is true for any infinite down-right path, it is true as well for any subpath of the form $\lbrace (n_i, m_i)\rbrace_{1\leqslant i\leqslant k}$ as in the statement of the Proposition. Hence we only need to show that the distribution of increments is preserved under the two local moves. 

\bigskip 
The distribution of increments on the boundary is constrained by the definition of the model. We have that $\mathcal Z^{\rm stat}(n,n) = w_{n,n}\mathcal Z^{\rm stat}(n,n-1)$, so that $V_{n,n}= w_{n,n}$ and we recall that $w_{n,n}\sim  \mathrm{Gamma^{-1}}(\diag+\alpha_n)$ is independent from all other weights. Hence, for any  $n=m$, $V_{n,m}$ is distributed as $\mathrm{Gamma^{-1}}(\diag+\alpha_n)$ and is independent from the increments $\lbrace U_{n',m'}, V_{n',m'}\rbrace $ for $m'<m$ (since those increments are independent from $w_{n,n}$). Hence, after a boundary update, the distribution of increments is preserved. 

\bigskip 
In order to show that the property is preserved under bulk update, we use Lemma~\ref{lem:seppalainenstationary}. After a bulk update, the increments are updated according to \eqref{eq:recurrenceforincrements}, where $w_{n,m}$ is independent from the increments on the earlier path and distributed as $\mathrm{Gamma^{-1}}(2\alpha)$ (recall that we have assumed that $\alpha_2=\alpha_3=\dots = \alpha$). If $\diag+\alpha_1=0$, we may set $\theta= \alpha_1=-\diag$ and Lemma~\ref{lem:seppalainenstationary} implies that increments along the new path will be distributed as $U_{n,m}\sim \mathrm{Gamma^{-1}}(\alpha+\alpha_1)$, $V_{n,m}\sim \mathrm{Gamma^{-1}}(\alpha+\diag)$. This shows that the distribution of increments is preserved under bulk update. Because before the bulk update, the variables $(U, V, w)$ are independent from
the rest of the increments $I_j$ by induction, and the new random variables
$(U' , V' )$ are just measurable functions of $(U, V, w)$, the new variables are
also independent of the other increments $I_j$. This concludes the proof. 
\end{proof}

One consequence of the stationary structure is that we may compute the expectation of $\log  \mathcal Z^{\rm stat}(n,m)$. We assume that parameters $\alpha_i$ are chosen as in Proposition \ref{prop:stationarity}. Observe that $\log  \mathcal Z^{\rm stat}(n,m)$ is equal to the sum of the logarithms of increments of the partition function along any path from $(1,1)$ to $(n,m)$. These increments are not independent, so that their sum is a highly non trivial random variable, but we know the expectation of each increment. Since the vertical increments are distributed as  $\mathrm{Gamma^{-1}}(\diag+\alpha)$ and the horizontal increments are distributed as $\mathrm{Gamma^{-1}}(\alpha_1+\alpha)$, we have that (for $\diag+\alpha_1=0$)
\begin{equation}
\mathbb E\left[ \log \mathcal \mathcal Z^{\rm stat}(n,m) \right] = - (n-1) \psi(\alpha_1+\alpha) -(m-1) \psi(\diag+\alpha),
\label{eq:expectationlogZ}
\end{equation}
where we have used that $\mathbb E[\log (\mathrm{Gamma^{-1}}(\theta))] =-\psi(\theta) $ and $\psi$ is the digamma function.

\subsection{Convergence to the half-space KPZ equation}
\label{sec:convergenceloggammaSHE}
At high temperature (when the parameters of inverse gamma random variables go to infinity and space-time coordinates are rescaled appropriately), the partition function $\mathcal Z(n,m)$ converges to the multiplicative noise stochastic heat  equation on $\R_{\geqslant 0}$ with Robin type boundary condition \cite{wu2018intermediate}. \\

Although the convergence of discrete directed polymers to half-space KPZ equation was proved rigorously in \cite{wu2018intermediate} (based on the full-space analogous result in \cite{alberts2014intermediate}), we will rederive (heuristically) this convergence in order to adapt it to our units and initial condition (which is not covered in \cite{wu2018intermediate}). 
 Let us change coordinates and use more natural time and space coordinates $\uptau = n+m-2$ and $\varkappa = n-m$. The partition function $Z_d(\varkappa, \uptau) := \mathcal Z(n,m)$ satisfies the 
discrete version of the stochastic heat equation
\begin{equation}
Z_d(\varkappa,\uptau) = w_{\varkappa,\uptau} (Z_d(\varkappa-1,\uptau-1)+Z_d(\varkappa+1,\uptau-1)),\hspace{1cm} \varkappa >0
\label{eq:discreterecurrence}
\end{equation}
where $w_{\varkappa,\uptau}\sim\mathrm{Gamma}^{-1}(\gamma_{\varkappa, \uptau})$ with parameter $\gamma_{\varkappa,\uptau}=\alpha_n+\beta_m$ (independent for each $\varkappa,\uptau$). The boundary condition at $\varkappa=0$ is given by
\begin{equation}
Z_d(0,\uptau) = w_{0,\uptau} Z_d(1, \uptau-1).
\label{eq:boundarycondition}
\end{equation}

Let us renormalize $Z_d$ and define 
$Z_r(\varkappa, \uptau) = C^{-\uptau} Z_d(\varkappa, \uptau)$. The correct factor to use is such that $C^{\uptau}$ behaves asymptotically as the point to line partition function where weights would be replaced by their average. Hence we set $C=2\mathbb E [w_{\varkappa, \uptau}]$. Note that in the following, we will choose parameters so that $\mathbb E [w_{\varkappa, \uptau}]$ does not depend on $\uptau, \varkappa$ (except along the lines  $\uptau=\varkappa$ or $\varkappa=0$). We may rewrite \eqref{eq:discreterecurrence} as 
\begin{equation} 
\nabla_{\uptau} Z_r(\varkappa, \uptau) = \tfrac{1+\Brescaled_{\varkappa, \uptau}}{2} \Delta_\varkappa Z_r(\varkappa, \uptau-1) +\Brescaled_{\varkappa, \uptau} Z_r(\varkappa, \uptau-1),
\label{eq:discreteSHE}
\end{equation}
where $\Brescaled_{\varkappa, \uptau}= \frac{2 w_{\varkappa,\uptau}}{C}-1 $, $\nabla_{\uptau}$ is
the discrete time derivative and $ \Delta_\varkappa$ is the discrete Laplacian. 

Let us fix $\diag\in \R$, $\alpha_1\in \R$ and set $\alpha_i = 1/2 + \sqrt{n}/2$ for all $i\geqslant 2$ and use the scalings 
\begin{equation}
\uptau=n t/2, \quad \varkappa=\sqrt{n}x/2.
\label{eq:scalingsloggamma}
\end{equation}
 In this case, we may choose $C = 2/\sqrt{n}$ and the family of random variables $w_{\varkappa,\uptau}$ rescales to a white noise in the sense that $n\, \Brescaled_{\varkappa, \uptau} \Rightarrow \sqrt{2} \xi(x,t)$. 

\bigskip 

 At $\uptau = \varkappa$, we have that for large $\varkappa$,  
 \begin{equation}
 Z_r(\varkappa, \varkappa) = \mathrm{Gamma}^{-1}(\diag+\alpha_1) \times   e^{\mathcal B(x)  -\alpha_1 x } + o\left(\frac{1}{n}\right), 
 \end{equation}
where the inverse Gamma random variable (coming from $w_{1,1}$) and the Brownian motion $\mathcal B(x)$ are independent.  

It is then natural to  define the continuous limit 
 \begin{equation}
 Z_\infty(x,t) = \lim_{n\to\infty}  Z_r(\varkappa, \uptau), 
 \end{equation}
 so that $Z_{\infty}$ has the initial data 
$
Z_\infty(x,0) =  \mathrm{Gamma}^{-1}(\diag+\alpha_1) \times   e^{\mathcal B(x)  -\alpha_1 x }.
$
Under the scalings that we consider, the boundary condition \eqref{eq:boundarycondition}  becomes
\begin{equation}
Z_{\infty}(0,t) \approx \frac{w}{C}\, Z_{\infty}\left(\frac{2}{\sqrt{n}},t\right). 
\label{eq:limitboundarycondition}
\end{equation} 
Let us  take the average on both sides of \eqref{eq:limitboundarycondition}.  We use that $\mathbb E[\frac{w}{C}]= \frac{1}{1+\frac{2\diag -1}{\sqrt{n}}},$  and consider that the weight $w$ is independent from $Z_{\infty}(\frac{2}{\sqrt{n}},t)$, as this is true in \eqref{eq:boundarycondition}.  We obtain  
 \begin{equation}
\mathbb E\left[  Z_{\infty}(0,t)  \right] = \left( 1-\frac{2}{\sqrt{n} }  (\diag-1/2) \right)  \mathbb E \left[ Z_{\infty}\left(\frac{2}{\sqrt{n}},t\right) \right]+ o(1/\sqrt{n}), 
 \end{equation}
 which, by Taylor approximation, leads to 
\begin{equation}
\partial_x \mathbb E\left[ Z_\infty(x,t)\right]\Big\vert_{x=0} = (\diag -1/2)\mathbb E\left[Z_{\infty}(0,t)\right].
\label{eq:boundaryconditionexpected} 
\end{equation}
Note that one may also obtain the more general boundary condition \eqref{eq:boundaryconditionformoments} for mixed moments by multiplying both sides of \eqref{eq:limitboundarycondition} by $Z_{\infty}(x_2,t)\dots Z_{\infty}(x_n,t)$ before taking the average.
Finally, 
multiplying  \eqref{eq:discreteSHE} by $n$ we obtain, when taking formally the $n\to\infty$ limit, that $Z_{\infty}(x,t)$ should satisfy the SHE \eqref{SHE}.  Thus, we have arrived at the following. 
\begin{claim}[{Combining \cite{wu2018intermediate} and  \cite{parekh2019}.}]
Let $\mathcal Z(n,m)$ be the partition function of the log-gamma polymer (see Definition~\ref{def:loggammahalfspace}) where $\alpha_2=\alpha_3= \dots = \frac{\sqrt{n}}{2}+\frac{1}{2}$, $\diag = A+\frac{1}{2}$ and $\alpha_1 = B+\frac{1}{2}$.  Let $Z(x,t)$ be the solution of the multiplicative noise stochastic heat equation  from Definition~\ref{def:SHEboundarybrownian} with boundary parameter $A$ and initial drift $-1/2-B$. 
Fix $t>0,x\geqslant 0$. Then the family of random variables 
\begin{equation}
\left\lbrace \frac{ \mathcal Z \left(\frac{nt +\sqrt{n}x}{4} , \frac{nt -\sqrt{n}x}{4} \right) }{\left( \frac{2}{\sqrt{n}}\right)^{nt/2-2}} \right\rbrace_{t>0,x \geqslant 0}
\end{equation}
converges in distribution to $\mathrm{Gamma}^{-1}(A+B+1)\times Z(x,t)$ (in the space of continuous space-time trajectories), where the inverse Gamma random variable is independent from the process $Z(x,t)$. {Moreover, the partition function $\mathcal Z^{\rm stat}(n,m)$ from Section \ref{sec:stationarystructure} converges, under the exact same scalings,  to $Z(x,t)$ (not multiplied by $\mathrm{Gamma}^{-1}(A+B+1)$).}
\label{prop:convergencediscretecontinuous}
\end{claim}

Note that the derivation that we presented above is only heuristic. We will not provide a complete proof of this result, though we indicate where the needed arguments can be found. The convergence of the polymer partition function for the half-space log-gamma polymer to the multiplicative noise stochastic heat equation is proved in \cite{wu2018intermediate} using a chaos series representation of the polymer partition function, see in particular Section 5 therein. However, the setting of  \cite{wu2018intermediate} restricts to delta initial data $ B=+\infty$ (and Robin type boundary with arbitrary parameter $A$). The convergence for Brownian initial data with arbitrary parameter $B$ was proven in \cite[Theorem 2.2]{parekh2019}, though \cite{parekh2019} works only in the case of Dirichlet boundary condition, that is in the case $A=+\infty$. Hence, one needs to combine the arguments from \cite{wu2018intermediate} and \cite{parekh2019} to deduce this result. 

\bigskip 
Using the stationary structure from Section~\ref{sec:stationarystructure} together with Claim~\ref{prop:convergencediscretecontinuous}, we obtain the following. Let  $Z(x,t)$ be as in Claim~\ref{prop:convergencediscretecontinuous} and assume that $A+B+1=0$. Then, for any time $t>0$, $Z(x,t)/Z(0,t)$ is the exponential of a Brownian motion with drift $-B-1/2$.

\bigskip 
We may also compute the expectation of $h(x,t) = \log Z(x,t)$ in the stationary case when $A+B+1=0$. Using \eqref{eq:expectationlogZ} and plugging there the scalings of Claim \ref{prop:convergencediscretecontinuous}, we obtain that 
\begin{equation}
\mathbb E\left[ h(x,t)\right]= -\frac{t}{12} + \left(B+\frac 1 2\right)^2t-\left(B+\frac 1 2\right)x, \quad \quad A+B+1=0.
\end{equation}
In particular, when $A=B=-1/2$, we have that $\mathbb E\left[ h(0,t)\right]=-t/12$.

\subsection{Moments of the half-space KPZ equation with Brownian initial data}

Using the moment formula from Proposition~\ref{prop:momentsloggamma} and the convergence result from Claim~\ref{prop:convergencediscretecontinuous}, we obtain the following moment formula for the half-space stochastic heat equation $Z(x,t)$. Note that the formula is valid for any $x\geqslant 0$. 
\begin{claim} Let $Z(x,t)$ be the solution to the half-space stochastic heat equation (Definition~\ref{def:SHEboundarybrownian}) with Brownian initial data with drift  $-1/2-B$ and boundary parameter $A$. Assume that  $B>k-1$, and $A+B>k-1$. Then, we have
	\begin{multline}
	\EE[Z(x,t)^k] =2^k \frac{\Gamma(A+B+1)}{\Gamma(A+B+1-k)} \int_{r_1+\I\R}\frac{\mathrm{d}z_1}{2\I\pi}\cdots \int_{r_k+\I\R}\frac{\mathrm{d}z_k}{2\I\pi} \prod_{1\leqslant a<b\leqslant k} \frac{z_a-z_b}{z_a-z_b-1}\, \frac{z_a+z_b}{z_a+z_b-1}\\ \times
	\prod_{i=1}^k \frac{z_i}{z_i+A}  \frac{1}{B^2-z_i^2}    e^{tz_i^2 - xz_i},
	\label{eq:momentsKPZhalfspaceBrownian}
	\end{multline}
	where the contours are chosen so that $B>r_1>r_2+1>\dots, >r_k+k-1>\max\lbrace k-1,k-1-A\rbrace$. 
	\label{prop:nestedcontours}
\end{claim}
\begin{remark}
	One may also compute mixed moments of $ Z(x,t)$, that is $\mathbb E [Z(x_1,t)\dots Z(x_k,t)]$. 
\end{remark}

\begin{proof}
	We start from the moment formula given in Proposition~\ref{prop:momentsloggamma}. Under the scalings considered in Claim~\ref{prop:convergencediscretecontinuous}, the second line of \eqref{eq:momentsZ} becomes
	$$ \prod_{i=1}^k \frac{2z_i}{z_i-\diag +1/2}  \frac{1}{(\alpha_1-1/2)^2-z_i^2} \left( \frac{1}{\sqrt{n}/2 -z_i } \right)^{\frac{nt+\sqrt{n}x}{4}-1 }  \left( \frac{1}{\sqrt{n}/2  +z_i} \right)^{\frac{nt-\sqrt{n}x}{4}-1 }.$$
	Using dominated convergence, one readily obtains that  
	\begin{multline}
	\lim_{n\to\infty}\mathbb{E} \left( \frac{ \mathcal Z\left(\frac{nt +\sqrt{n}x}{4} , \frac{nt -\sqrt{n}x}{4} \right) }{\left( \frac{2}{\sqrt{n}}\right)^{nt/2-2}} \right)^k =
	 \int_{r_1+\I\R}\frac{\mathrm{d}z_1}{2\I\pi}\cdots \int_{r_k+\I\R}\frac{\mathrm{d}z_k}{2\I\pi} \\ \prod_{1\leqslant a<b\leqslant k} \frac{z_a-z_b}{z_a-z_b-1}\, \frac{z_a+z_b}{1+z_a+z_b}
	\prod_{i=1}^k \frac{2z_i}{z_i-A}  \frac{1}{B^2-z_i^2}    e^{tz_i^2 + xz_i}. 
	\label{eq:limitmoments}
	\end{multline}
Using the convergence in distribution from Claim~\ref{prop:convergencediscretecontinuous} the left hand side in \eqref{eq:limitmoments} converges to \break $m_k \mathbb E[Z(x,t)^k]$, where $m_k$ is the $k$-th moment of an inverse Gamma random variable with parameter $\diag+\alpha_1$  (the convergence in distribution implies the  convergence of moments modulo some tail bounds, which, using Markov inequality, can be proven from our explicit moment formulae. Details of this argument are provided in a similar case in \cite[Section 3]{ghosal2018moments}). It is well known that 
$m_k =\Gamma(\diag+\alpha_1-k)/ \Gamma(\diag+\alpha_1)$, so that 
\begin{equation}
\mathbb E[Z(x,t)^k] = \frac{\Gamma(A+B+1)}{\Gamma(A+B+1-k)} \times \text{R.H.S. of } \eqref{eq:limitmoments}
\end{equation}
	Finally, we have used the change of variables $\tilde z_i=-z_{k-i+1}$ to obtain the statement of the Proposition.
\end{proof}

\subsection{Symmetry between drift and boundary parameters} 
\label{sec:symmetry}
We now exploit one of the symmetries of the log-gamma polymer, arising from more general symmetries of so-called half-space Macdonald
processes \cite{barraquand2018half}, which, in the 
KPZ limit, extends a result of Parekh \cite{parekh2019}. 

\begin{claim}
Let us denote here by $Z_A^{B}(x,t)$ (with explicit dependence in the parameters $A,B$) the solution to the multiplicative noise SHE with boundary parameter $A$ and initial drift $-1/2-B$
(Definition~\ref{def:SHEboundarybrownian}).
 For any fixed $t>0$ and $A,B\in \R$, we have the equality in distribution
\begin{equation}
Z_A^{B}(0,t) = Z_{B}^{A}(0,t).
\end{equation}
\label{prop:symmetryAB}
\end{claim}
\begin{proof}
This is a consequence of \cite[Proposition 8.1]{barraquand2018half} which states that the law of the partition function $\mathcal Z(n,n)$ of the half-space log-gamma polymer is invariant under exchanging parameters $\diag$ and $\alpha_1$ (recall Definition~\ref{def:loggammahalfspace}), along with the convergence result from Claim~\ref{prop:convergencediscretecontinuous}. \\

However, \cite[Proposition 8.1]{barraquand2018half} assumes that $\diag+\alpha_1>0$ as in Definition~\ref{def:loggammahalfspace}, which would require the condition $A+B+1>0$. Let us explain why we do not need to assume  this condition. Recall Definition~\ref{def:loggammahalfspace} and let us denote the partition function by $\mathcal Z_{\diag}^{\alpha_1}(n,n)$,  where we indicate explicitly the dependence on parameters $\diag, \alpha_1$. The result from \cite[Proposition 8.1]{barraquand2018half} implies that for $\tilde\diag,\tilde \alpha_1$ such that   $\tilde\diag+\tilde \alpha_1>0$ we have the equality in distribution 
$$ \mathcal Z_{\tilde \diag}^{\tilde \alpha_1}(n,n) = \mathcal Z_{\tilde \alpha_1}^{\tilde \diag}(n,n).$$ 
Notice that by Definition~\ref{def:loggammahalfspace}, we have 
$$ \mathcal Z_{\tilde \diag}^{\tilde \alpha_1}(n,n) = w_{1,1}  Z_{\tilde \diag}^{\tilde \alpha_1, \mathrm{stat}}, \;\;\;\mathcal Z_{\tilde \alpha_1}^{\tilde \diag}(n,n) = w_{1,1}  Z_{\tilde \alpha_1}^{\tilde \diag , \mathrm{stat}}, $$
where $\mathcal Z_{\diag}^{\alpha_1, \mathrm{stat}}$ is defined as in Section~\ref{sec:stationarystructure} and $w_{1,1}$ has the same distribution in both cases. This implies that we have also the equality in distribution 
\begin{equation}
\mathcal Z_{\tilde \diag}^{\tilde \alpha_1 , \mathrm{stat}}(n,n) = \mathcal Z_{\tilde \alpha_1}^{\tilde \diag , \mathrm{stat}}(n,n).
\label{eq:equalitystationarypartitionfunctions}
\end{equation} 
The distribution of the random variable in \eqref{eq:equalitystationarypartitionfunctions} depends on parameters $\tilde\diag, \tilde \alpha_1, \alpha_2, \dots, \alpha_n$. The equality in distribution can be analytically extended to all parameters such that  $\alpha_i+\alpha_j>0$ for any $2\leqslant i<j\leqslant n$, $\tilde \alpha_1+\alpha_i>0$ for any $2\leqslant i\leqslant n$, and  $\tilde \diag+\alpha_i>0$ for any $2\leqslant i\leqslant n$. In particular, we do not require anymore that $\tilde \diag+\tilde\alpha_1>0$. Passing to the limit in \eqref{eq:equalitystationarypartitionfunctions} using Claim~\ref{prop:convergencediscretecontinuous}, we obtain the desired result. 
\end{proof}

\subsection{Another conjectural identity in law} 
\label{sec:conjecturalidentity}

{The symmetry between parameters $A$ and $B$ stated in Claim~\ref{prop:symmetryAB} relies on a similar property for the log-gamma polymer (symmetry between $\diag$ and $\alpha_1$) based on the theory of half-space Macdonald processes ans stated as \cite[Proposition 8.1]{barraquand2018half}. This result was stated in \cite{barraquand2018half} for the log-gamma polymer model where $\alpha_2 = \dots = \alpha_n$. However, the same property actually holds for general  half-space Macdonald process and for any choice of parameters, see \cite[Proposition 2.6]{barraquand2018half}, as long as we restrict to the partition function on the boundary. Furthermore, the law of $\mathcal Z(n,n)$ is symmetric with respect to permutation of the parameters $\alpha_i$. Thus, we claim that one can also exchange the roles of the parameter $\diag$ and the parameter $\alpha_2$. }\\

{When scaling parameters to the (multiplicative noise) stochastic heat equation in Claim~\ref{prop:convergencediscretecontinuous}, we set $\diag=A-\frac 1 2$, $\alpha_1=B-\frac 1 2$ and $\alpha_i=\frac{\sqrt{n}}{2}$. If we exchange parameters $\diag$ and $\alpha_2$, we expect that we will obtain the stochastic heat equation with boundary parameter equal to $+\infty$ (i.e. with Dirichlet boundary condition $Z(0,t)=0$) and initial condition given by 
\begin{equation}
Z(x,0) = \int_{0}^x \exp\left( \mathcal B_1(y) + \mathcal B_2(x)-\mathcal B_2(y) \right)\mathrm d y , 
\label{eq:initialconditiontwobrownians}
\end{equation} 
where $\mathcal B_1$ and $\mathcal B_2$ are independent Brownian motions with respective drifts $-(B+1/2)$ and $-(A+1/2)$. Let us call $Z_{\rm Dir}^{A,B}$ the solution to the heat equation with Dirichlet boundary condition and the initial condition given above in \eqref{eq:initialconditiontwobrownians}. We refer to \cite{parekh2019kpz123} regarding the exact meaning of the Dirichlet boundary condition in this context. Following similar arguments as in the proof of \cite[Theorem 1.1]{parekh2019kpz123}, we conjecture that for any $t>0$, we have the identity in distribution
\begin{equation}
Z_A^B(0,t) 
= \lim_{x\to 0} \frac{Z_{\rm Dir}^{A,B}(x,t)}{x}.
\label{eq:conjecturalidentity}
\end{equation} 
}

{The identity in law \eqref{eq:conjecturalidentity} allows to predict the law of large numbers for $h(0,t)= \log Z_A^B(0,t)$ in the bound phase. Recall that we expect that  $h(0,t)$ follows the asymptotics 
\be
h(0,t) \simeq v^{A,B}_{\infty} t + t^{\beta} \chi 
\label{eq:genericlimittheorem}
\ee 
where $\chi$ is an $\mathcal{O}(1)$ random variable, and $\beta$ the growth fluctuation exponent. Using \eqref{eq:conjecturalidentity}, $\log \frac{Z_{\rm Dir}^{A,B}(x,t)}{x}$  should also follow the same asymptotics. When $A<-1/2$ or $B<-1/2$, we see that $Z_{\rm Dir}^{A,B}(x,0)$ grows as $e^{\max\lbrace |A+\frac 1 2| , |B+\frac 1 2 | \rbrace x}$. Hence, the polymer  partition function will be dominated by paths from $(x,0)$ to $(0,t)$ where $x$ is of order $t$. More precisely, since the point to point free energy from $(x,0)$ to $(0,t)$ behaves asymptotically as $-\frac{t}{12}- \frac{x^2}{4t}$, the partition function will be dominated by paths leaving from $x=x^*_t$, where $x^*_t = {\rm argmax}_{x>0}( - \frac{x^2}{4t} + \max\lbrace A+\frac 1 2 , B+\frac 1 2  \rbrace x)
=2 \max\lbrace \vert A+\frac 1 2 \vert , \vert B+\frac 1 2 \vert \rbrace t$. Thus, we have that  for general $A,B$, 
$$v^{A,B}_{\infty} = - \frac{1}{12} + \Big(\min\big\lbrace A+\frac{1}{2} , B+\frac{1}{2}, 0\big\rbrace\Big)^2.$$

We can even predict the fluctuation exponent $\beta$ and the nature of fluctuations. When $A<-1/2$ or $B<-1/2$ with $A\neq B$ (say $A<B$ for simplicity), the initial condition for $Z_{\rm Dir}^{A,B}(x,t)$ in \eqref{eq:initialconditiontwobrownians} will essentially be the Brownian motion, i.e. $\mathcal B_2(x)$ with drift $-(A+1/2)$. The fluctuations of the initial condition at the optimal point $x^*_t$ are thus
$\approx \mathcal B_2(x^*_t)$ i.e. Gaussian on the scale $t^{1/2}$, and they  will dominate the fluctuations in the partition function, hence we find that $\beta=1/2$ and $\chi$ is Gaussian. The situation is completely similar when $B<A$. \\

{When $A=B<-1/2$, the situation is a bit more delicate since we cannot approximate the initial condition  \eqref{eq:initialconditiontwobrownians} by a Brownian motion. Instead, we notice that \eqref{eq:initialconditiontwobrownians} can be interpreted as the partition function in the O'Connell-Yor directed polymer model \cite{oconnell2001brownian}. For large values of $x$, it behaves as 
\begin{align} \log Z_{\rm Dir}^{A,B}(x,0) &= \log \int_{0}^x  \exp\left(\mathcal B_1(y) + \mathcal B_2(x)-\mathcal B_2(y) \right) \mathrm d y \\ &\approx  \left\vert B+\frac 1 2\right\vert x 
+
\sqrt{x} \max_{0\leqslant y\leqslant 1} \left\lbrace \mathcal B_1(y) + \mathcal B_2(1)-\mathcal B_2(y) \right\rbrace.
 \end{align}
It was proved \cite{baryshnikov2001gues, gravner2001limit}  that the latter quantity  behaves asymptotically as the largest eigenvalue of a $2\times 2$ GUE matrix in the scale $x^{1/2}$. We must now evaluate this
quantity at the optimal point $x=x^*_t = \vert 2 B+1\vert t$, hence we find that $\beta=1/2$ and that $\chi$ has the same distribution as the top eigenvalue of a $2\times 2$ GUE matrix, discussed for instance in \cite[Section  6]{oconnell2001brownian}.}

\subsection{Residue expansion}

In this section, we restrict to $x=0$ and denote $Z(0,t)=Z(t)$ as in the previous sections. The moment formula \eqref{eq:momentsKPZhalfspaceBrownian} is not convenient for asymptotic analysis because the contours are different for each variable, so that the complexity of the formula significantly increases with $k$. To overcome this issue, one has to deform the contours  to all lie on  a fixed vertical line and take into account the residues encountered during this contour deformation. This procedure was implemented in \cite{borodin2016directed}, but the computation of residues is very involved and the result relies on a conjectural combinatorial simplification. Applying this result \cite[Conjecture 5.2]{borodin2016directed} to the moment formula  \eqref{eq:momentsKPZhalfspaceBrownian}, we conjecture that for  $A>k-1$ and  $B>k-1$,  we have 
\begin{multline}
\EE[Z(t)^k] = 2^k \frac{\Gamma(A+B+1)}{\Gamma(A+B+1-k)} \sum_{\underset{\lambda=1^{m_1}2^{m_2}\dots}{\lambda\vdash k}} \frac{(-1)^{\ell(\la)}}{m_1!m_2!\dots} \int_{\I\R} \frac{\rmd w_1}{2\I\pi} \dots \int_{\I\R} \frac{\rmd w_{\ell(\la)}}{2\I\pi}\\
\times \prod_{j=1}^{\ell(\la)} \frac{(w_j+1/2)_{\la_j-1}}{4(w_j)_{\la_j}} \Pf \left[ \frac{u_i-u_j}{u_i+u_j} \right]_{i,j=1}^{2\ell(\la)} \\
\times E(w_1, w_1+1, \dots, w_1+\lambda_1-1, w_2, \dots, w_2+\lambda_2-1, \dots, w_{\ell(\la)}, \dots,w_{\ell(\la)} + \lambda_{\ell(\la)} -1),
\end{multline}
where we use the Pochhammer notation for rising factorials $(w)_{\la} = w(w+1)\dots (w+\la-1)$, we define variables $u_i$ for $1\leqslant i\leqslant 2\ell(\la)$ as 
\begin{equation}
(u_1, \dots, u_{2\ell(\la)}) = (-w_1+ \tfrac{1}{2}, w_1 - \tfrac{1}{2} + \lambda_1,
-w_2+ \tfrac{1}{2}, w_2 - \tfrac{1}{2} + \lambda_2,\dots,
-w_{\ell(\la)}+ \tfrac{1}{2}, w_{\ell(\la)} - \tfrac{1}{2} + \lambda_{\ell(\la)}),
\end{equation}
and where 
\begin{equation}
E(z_1, \dots, z_k) = \prod_{i=1}^k \frac{e^{tz_i^2} }{B^2-z_i^2}    \sum_{\sigma\in BC_k} \sigma\left( \prod_{1\leqslant j<i\leqslant k}  \frac{z_i-z_j-1}{z_i-z_j} \frac{z_i+z_j-1}{z_i+z_j}  \prod_{i=1}^k \frac{z_i}{z_i+A} \right).
\end{equation}
It turns out that the symmetrization can be performed using \cite[Equation (54)]{borodin2016directed}, relying on the theory of BC-symmetric Hall-Littlewood polynomials \cite{venkateswaran2015symmetric},  and one finds 
\begin{equation}
E(z_1, \dots, z_k) = 2^k k! \prod_{i=1}^k \frac{e^{tz_i^2} }{B^2-z_i^2} \frac{z_i^2}{z_i^2-A^2}.
\end{equation}
\begin{remark}
	It is now apparent that the moment formulae are invariant with respect to the transformation $(A,B)\mapsto(B,A)$. 
\end{remark}

Performing explicitly the evaluation into strings, we obtain
\begin{multline}
\EE[Z(t)^k] =  4^k k! \frac{\Gamma(A+B+1)}{\Gamma(A+B+1-k)} \sum_{\underset{\lambda=1^{m_1}2^{m_2}\dots}{\lambda\vdash k}} \frac{(-1)^{\ell(\la)}}{m_1!m_2!\dots} \int_{\I\R} \frac{\rmd w_1}{2\I\pi} \dots \int_{\I\R} \frac{\rmd w_{\ell(\la)}}{2\I\pi}\\
\times \Pf \left[ \frac{u_i-u_j}{u_i+u_j} \right]_{i,j=1}^{2\ell(\la)}  \prod_{j=1}^{\ell(\la)}  \frac{e^{t \mathtt{G}(w_i+\la_i)}}{e^{t \mathtt{G}(w_i)}}   \frac{(w_j+1/2)_{\la_j-1}}{4(w_j)_{\la_j}}  \left( \frac{\Gamma(w_i+\lambda_i)}{\Gamma(w_i)} \right)^2\\ \times \frac{\Gamma(B-w_j-\la_j+1)\Gamma(B+w_j)}{\Gamma(B-w_j+1)\Gamma(B+w_j+\la_j)}\frac{\Gamma(w_j-A)\Gamma(w_j+A)}{\Gamma(w_j+\la_j-A)\Gamma(w_j+\la_j+A)},
\end{multline}
where 
\begin{equation*}
\mathtt{G}(w) = \frac{w^3}{3}-\frac{w^2}{2}+\frac{w}{6},
\end{equation*}
so that 
\begin{equation}
\mathtt{G}(w+\ell)-\mathtt{G}(w) = w^2+(w+1)^2 + \dots + (w+\ell-1)^2. 
\end{equation}

Note that $B^2-z^2 = (\pm B-z)(\pm B+z)$ and $z^2-A^2 = - (\pm A-z)(\pm A+z)$, so that many choices are possible for the evaluation of $E(z_1, \dots, z_k)$ into strings. The most convenient choice seems to be the following  formula.  

\begin{claim}[{based on \cite[Conjecture 5.2]{borodin2016directed}}]
For  $A>k-1$ and  $B>k-1$, we have 
\begin{multline}
\EE[Z(t)^k] = 4^k k! \frac{\Gamma(A+B+1)}{\Gamma(A+B+1-k)} \sum_{\underset{\lambda=1^{m_1}2^{m_2}\dots}{\lambda\vdash k}} \frac{(-1)^{\ell(\la)}}{m_1!m_2!\dots} \int_{\I\R} \frac{\rmd w_1}{2\I\pi} \dots \int_{\I\R} \frac{\rmd w_{\ell(\la)}}{2\I\pi}\\
\times \Pf \left[ \frac{u_i-u_j}{u_i+u_j} \right]_{i,j=1}^{2\ell(\la)}  \prod_{j=1}^{\ell(\la)}  \frac{e^{t \mathtt{G}(w_i+\la_i)}}{e^{t \mathtt{G}(w_i)}}   \frac{(w_j+1/2)_{\la_j-1}}{4(w_j)_{\la_j}} \frac{\Gamma(-w_j+1)\Gamma(w_j+\lambda_j)}{\Gamma(-w_j-\lambda_j+1)\Gamma(w_j)} \\ \times \frac{\Gamma(B-w_j-\la_j+1)\Gamma(B+w_j)}{\Gamma(B-w_j+1)\Gamma(B+w_j+\la_j)}\frac{\Gamma(A-w_j-\la_j+1)\Gamma(A+w_j)}{\Gamma(A-w_j+1)\Gamma(A+w_j+\la_j)}.
\label{eq:conjecture} 
\end{multline}
\label{conj:momentsPfaffian}
\end{claim}
Comparing with the formula \eqref{znn2} obtained from the replica Bethe ansatz, we
see that \eqref{eq:conjecture} and \eqref{znn2} agree after the substitutions
\begin{equation}
w_j \to \I k_j + \frac{1-m_j}{2}, \quad  \quad \ell(\lambda) \to n_s, \quad  \quad \lambda_j \to m_j.
\end{equation}
Indeed, under this change of variables,
we have 
\begin{equation}
(u_1,u_2,\dots,u_{2\ell(\la)-1} u_{2\ell(\la)}) = (X_2,X_1, \dots, X_{2n_s}, X_{2n_s-1}) 
\end{equation}
where $X_j = m_p + 2 \I k_p$for $j=2 p -1$, $X_j=
	m_p - 2 \I k_p$ for $j=2 p$. Thus we have that 
	\begin{equation} (-1)^{\ell(\la)} \Pf \left[ \frac{u_i-u_j}{u_i+u_j} \right]_{i,j=1}^{2\ell(\la)} = \underset{2 n_s \times 2 n_s}{\rm Pf} \left[ \frac{X_i-X_j}{X_i+X_j} \right].\end{equation}
One may also check that  $e^{\mathtt{G}(w_i+\la_i)-\mathtt{G}(w_i)} = e^{ (m_j^3-m_j) \frac{ t}{12} - m_j k_j^2 t }.$ Using the reflection formula for the Gamma function, we may write the hyperbolic sine in the definition of $B_{k,m}$ in \eqref{B} as a product of Gamma functions so that we have (dropping indices of variables $k_j,m_j,\lambda_j, w_j$)
\begin{equation}\frac{B_{k,m}}{4 \I k} = \frac{1}{2}\frac{\Gamma(2w-2\la -1)\Gamma(1-2w)}{\Gamma(2w+\la)\Gamma(1-2w-\la)}\frac{\Gamma(B-w-\la+1)\Gamma(B+w)}{\Gamma(B-w+1)\Gamma(B+w+\la)}\frac{\Gamma(A-w-\la+1)\Gamma(A+w)}{\Gamma(A-w+1)\Gamma(A+w+\la)}.\end{equation}
Using the duplication formula for the Gamma function and the reflection formula a few times, we arrive at 
\begin{multline} \frac{B_{k,m}}{4 \I k} =  2^{2\la} \frac{(w+1/2)_{\la-1}}{4(w)_{\la}} \frac{\Gamma(-w+1)\Gamma(w+\lambda)}{\Gamma(-w-\lambda+1)\Gamma(w)} \\ \times \frac{\Gamma(B-w-\la+1)\Gamma(B+w)}{\Gamma(B-w+1)\Gamma(B+w+\la)}\frac{\Gamma(A-w-\la+1)\Gamma(A+w)}{\Gamma(A-w+1)\Gamma(A+w+\la)}.\end{multline}
Thus, we have shown that \eqref{eq:conjecture} and \eqref{znn2} agree as claimed.

\section{Generating function in terms of a Fredholm Pfaffian}
\label{sec:FP} 

{ We will now write the moment generating function of $Z(t)$. We define, for $\varsigma>0$,
\begin{equation}
g(\varsigma) =\mathbb{E} \left[ \exp(- \varsigma e^{\frac{t}{12}} \invgamma Z(t)) \right].
\end{equation}
Ignoring the fact that the summation over $n$ cannot be exchanged with the expectation due to the divergence of moments, we will consider the following formal power series 
$$ 1 + \sum_{n=1}^\infty \frac{(- \varsigma e^{\frac{t}{12}})^n}{n!} \mathbb{E} \left[\invgamma^n Z(t)^n\right],$$
that we will again denote by $g(\varsigma)$. 
\begin{remark}
Following \cite{SasamotoStationary2},  we may informally write 
 \begin{equation}
\mathbb{E}[\invgamma^n]\varsigma^n =\frac{\Gamma(A+B-n+1)}{\Gamma(A+B+1)}\varsigma^n=\frac{\Gamma(A+B-\varsigma\partial_{\varsigma}+1)}{\Gamma(A+B+1)}\varsigma^n
 \end{equation}
 so that 
 \begin{equation}
 \mathbb{E} \left[ \exp(- \varsigma e^{\frac{t}{12}}  Z(t)) \right] =\frac{\Gamma(A+B+1)}{\Gamma(A+B+1-\varsigma\partial_{\varsigma})} \mathbb{E} \left[ \exp(- \varsigma e^{\frac{t}{12}} \invgamma Z(t)) \right],
 \label{eq:removingweights}
 \end{equation}
 where the operator $\frac{\Gamma(A+B+1)}{\Gamma(A+B+1-\varsigma\partial_{\varsigma})}$ should be understood in the following sense: If 
 $$  \frac{\Gamma(A+B+1)}{\Gamma(A+B+1-z)} = \sum_{n=0}^{+\infty}a_nz^n,$$
 in a neighborhood of $z=0$, then 
$$  \frac{\Gamma(A+B+1)}{\Gamma(A+B+1-\varsigma\partial_{\varsigma})} = \sum_{n=0}^{+\infty}a_n(\varsigma\partial_{\varsigma})^n.$$
Equation~\eqref{eq:removingweights} is a positive temperature analogue of \cite[Eq. (4.5)]{ferrari2006scaling} which we will use in the next section. This type of arguments can be traced back to the work of Baik and Rains \cite{baik2000limiting}. 
\label{rem:remark5.1}
\end{remark} 
\begin{remark} 
 Alternatively, one may use the density of the inverse Gamma distribution and write 
 \begin{equation} g(\varsigma) = \int_{0}^{+\infty} \frac{\mathrm d u}{\Gamma(A+B+1)}u^{-A-B-2}e^{-1/u}\mathbb E\left[ \exp(- \varsigma u e^{\frac{t}{12}}  Z(t))  \right].\label{eq:multiplicativeconvolution}\end{equation} 
We may rewrite this equation using the following representation of the Bessel function
 \begin{equation}
 \int_{0}^{+\infty}\rmd u \, u^{-\nu} e^{1/u-x^2 u} = 2 x^{\nu} K_{\nu}(2 x),
 \end{equation}
  for $x>0$, where $K_{\nu}$ denotes the modified Bessel $K$ function.
 Exchanging the expectation with the integral over $u$ in \eqref{eq:multiplicativeconvolution}, and using this integral representation,  we arrive at 
 \begin{equation}
 \EE\left[ 2 \left(\varsigma Z(t)e^{\frac{t}{12}}\right)^{\frac{1+A+B}{2}} K_{1+A+B}\left(2\sqrt{\varsigma Z(t)e^{\frac{t}{12}}}\right)  \right] = g(\varsigma)\Gamma(1+A+B), 
 \end{equation} 
 where $K_{\nu}(z)$ denotes the modified Bessel $K$ function. Note that similar integral transforms involving the modified Bessel function $K$ appear in  \cite[Theorem 2.9]{BCFV}. 
 It is plausible that this expression can be inverted to compute the distribution of $Z(t)$. For $A+B+1=0$, an inversion formula is provided in   \cite[Appendix E]{BCFV}. 
 We will see in Section~\ref{sec:largeTlimit} that we will actually not need to perform this inversion.  
\end{remark}}
The constraint $\sum_{i=1}^{n_s} m_i=n$ in \eqref{znn2} can then be relaxed by reorganizing the series according to the number of strings:
\begin{eqnarray} \label{defg}
g(\varsigma) = 1 +  \sum_{n_s=1}^\infty \frac{1}{n_s!}  {\sf Z}(n_s,\varsigma) 
\end{eqnarray}
where ${\sf Z}(n_s,\varsigma)$ is the partition sum at fixed number of strings $n_s$, calculated below. We now show that one can write the generating function as a Fredholm Pfaffian. It will
be possible thanks to the Schur Pfaffian identity, \eqref{pfid}, given above.  The partition sum at fixed number of strings, expressed in terms of the reduced variables $X_{2p-1} = m_p + 2 \I k_p $ and $X_{2p} = m_p - 2 \I k_p$ for $p\in [1,n_s]$, reads 
\be
{\sf Z}(n_s,\varsigma) = \prod_{p=1}^{n_s} \sum_{m_p \geqslant 1} \int_\mathbb{R} \frac{\rmd k_p}{2 \pi} 
(-\varsigma)^{m_p} \frac{B_{k_p,m_p}}{4 \I k_p} 
e^{-t  m_p k_p^2 + \frac{t}{12} m_p^3} 
\underset{2 n_s \times 2 n_s}{\rm Pf} \left[ \frac{X_i-X_j}{X_i+X_j}\right]
\label{eq:afterRefSummand}
\ee
where $B_{k,m}$ was given in \eqref{B}. The summation over the variables $m_p$ can be done using the Mellin-Barnes summation trick similarly to Refs.~\cite{SasamotoStationary2,SasamotoStationary}. The barrier $A>(n-1)/2$ is overcome exactly as in Ref.~\cite{SasamotoStationary2} (see Lemma.~6 and the discussion therein) from an analytic continuation of Gamma functions included in the $B_{k,m}$ factor, the introduction of a particular contour $C_0$ and a final requirement for the drift $A+1/2>0$. Indeed, define the contour  $C_0 = a + \I \mathbb{R}$ with  $a \in (0,\min\lbrace 2B+1,2A+1,1\rbrace)$, then for any holomorphic function $f$ having sufficient decay at infinity and in particular denoting the summand of Eq.~\eqref{eq:afterRefSummand} by the function $f(m_p)$, we have 

\bea \label{MB} 
\sum_{m \geqslant 1} (-\varsigma)^m f(m) = 
- \int_{C_0}  \frac{\rmd w}{2 \I \pi} \varsigma^w \frac{\pi}{\sin \pi w} f(w)
=
- \int_\mathbb{R} \rmd r \frac{\varsigma}{ \varsigma +e^{-r}}  \int_{C_0}  \frac{\rmd w}{2 \I \pi} f(w) e^{-w r}.
\eea 
For each $m_p$ we therefore introduce two variables $r_p$ and $w_p$ and we redefine the reduced variables $X_{2p}$ and $X_{2p-1}$ under the minimal replacement $m_p \to w_p$ imposed by the Mellin-Barnes formula, which we will apply despite the presence of poles on the right of the contour $C_0$. This is an a priori illegal step, but it will  exactly turn the diverging moment generating series into a well-defined and converging series equal to the Laplace transform.  We refer to \cite[Section 6]{borodin2016directed} where this procedure and its degree of rigor is discussed in great details. This leads to the following rewriting of the coefficient ${\sf Z}(n_s,\varsigma)$  as 
(see section 5 in \cite{krajenbrink2018large} for similar manipulations)
\be
\begin{split}
 {\sf Z}(n_s,\varsigma) =& (-1)^{n_s}\prod_{p=1}^{n_s} \int_\mathbb{R} \rmd r_p \frac{\varsigma}{\varsigma  + e^{ -r_p }} 
 \iint_{C_0^2} \frac{\rmd X_{2p-1}}{4\I\pi}   \frac{\rmd X_{2p}}{4\I\pi}
\frac{\sin(\frac{\pi}{2} (X_{2p}-X_{2p-1}))}{2\pi} \\
& \times  \frac{\Gamma(A+\frac{1}{2}-\frac{X_{2p}}{2})}{\Gamma(A+\frac{1}{2}+\frac{X_{2p}}{2})}\frac{\Gamma(A+\frac{1}{2}-\frac{X_{2p-1}}{2})}{\Gamma(A+\frac{1}{2}+\frac{X_{2p-1}}{2})} \frac{\Gamma(B+\frac{1}{2}-\frac{X_{2p}}{2})}{\Gamma(B+\frac{1}{2}+\frac{X_{2p}}{2})}\frac{\Gamma(B+\frac{1}{2}-\frac{X_{2p-1}}{2})}{\Gamma(B+\frac{1}{2}+\frac{X_{2p-1}}{2})}  \\
&\times\Gamma(X_{2p-1}) \Gamma(X_{2p})  e^{-(X_{2p-1}+X_{2p}) \frac{r_p}{2} + t ( \frac{X_{2p-1}^3}{24} + \frac{X_{2p}^3}{24} )}  \underset{2 n_s \times 2 n_s}{\rm Pf} \left[ \frac{X_i-X_j}{X_i+X_j}\right]
\end{split}
\ee
\begin{remark}
Note that the contour $C_0$ passes to the left of the poles of the Gamma function at $X=2A+1, 2B+1$.
\end{remark}
We observe that the integrals are almost separable in $X_{2p}$ and $X_{2p-1}$ except for the sine function which couples them.  By anti-symmetrization and similarly to \cite[Section 5]{krajenbrink2018large}, we can proceed to the replacement\footnote{This replacement is done using the trigonometric identity $\sin(x-y)=\sin(x)\cos(y)-\sin(y)\cos(x)$ and the antisymmetry of the Schur Pfaffian upon the relabelling $X_{2p}\leftrightarrow X_{2p-1}$.  Note that similarly to \cite{krajenbrink2018large}, we could have the more general decomposition $\sin\left(\frac{\pi}{2} (X_{2p}-X_{2p-1})\right)\to 2\sin\left(\frac{\pi}{2} X_{2p}+\theta \right)\cos\left(\frac{\pi}{2} X_{2p-1}+\theta\right).$ valid for all $\theta \in \R$ but we do not make use of it here.}
\begin{equation}
\sin\left(\frac{\pi}{2} (X_{2p}-X_{2p-1})\right)\to 2\sin\left(\frac{\pi}{2} X_{2p}\right)\cos\left(\frac{\pi}{2} X_{2p-1}\right).
\end{equation}
The last manipulations consist in rescaling all variables $X$ by a factor $2$ and replacing the contours of integration by $C=\frac{a}{2} + i \mathbb{R}$. Hence we have
\be
\begin{split}
 {\sf Z}(n_s,\varsigma) =& (-1)^{n_s}\prod_{p=1}^{n_s} \int_\mathbb{R} \rmd r_p \frac{\varsigma}{\varsigma  + e^{ -r_p }} 
 \iint_{C^2} \frac{\rmd X_{2p-1}}{2\I\pi}  \frac{\rmd X_{2p}}{2\I\pi}
\frac{\sin(\pi X_{2p})\cos(\pi X_{2p-1})}{\pi} \\
&\times 
\frac{\Gamma(A+\frac{1}{2}-X_{2p})}{\Gamma(A+\frac{1}{2}+X_{2p})}\frac{\Gamma(A+\frac{1}{2}-X_{2p-1})}{\Gamma(A+\frac{1}{2}+X_{2p-1})}  \frac{\Gamma(B+\frac{1}{2}-X_{2p})}{\Gamma(B+\frac{1}{2}+X_{2p})}\frac{\Gamma(B+\frac{1}{2}-X_{2p-1})}{\Gamma(B+\frac{1}{2}+X_{2p-1})}\\
&\times  \Gamma(2 X_{2p-1}) \Gamma(2X_{2p})  e^{-(X_{2p-1}+X_{2p}) r_p + t ( \frac{X_{2p-1}^3}{3} + \frac{X_{2p}^3}{3} )}  \underset{2 n_s \times 2 n_s}{\rm Pf} \left[ \frac{X_i-X_j}{X_i+X_j}\right]
\end{split}
\ee 
There are a few last steps before we introduce the Fredholm Pfaffian. First define the functions
\begin{equation}
\begin{split}
&\phi_{2p}(X)=\frac{\sin(\pi X)}{\pi}\Gamma(2X)\frac{\Gamma(A+\frac{1}{2}-X)}{\Gamma(A+\frac{1}{2}+X)}\frac{\Gamma(B+\frac{1}{2}-X)}{\Gamma(B+\frac{1}{2}+X)}e^{ -r_p X + t  \frac{X^3}{3} }\\
&\phi_{2p-1}(X)=\cos(\pi X)\Gamma(2X)\frac{\Gamma(A+\frac{1}{2}-X)}{\Gamma(A+\frac{1}{2}+X)}\frac{\Gamma(B+\frac{1}{2}-X)}{\Gamma(B+\frac{1}{2}+X)}e^{ -r_p X + t  \frac{X^3}{3} }
\end{split}
\end{equation}
Using a known property of Pfaffians (see De Bruijn \cite{de1955some}), we can rewrite the partition sum at fixed number of strings itself as a Pfaffian, i.e. we use that
\begin{equation}
\prod_{\ell=1}^{2n_s}\int_{C}\frac{\mathrm{d}X_{\ell}}{2\I\pi}\Phi_{\ell}(X_{\ell}) \underset{2 n_s \times 2 n_s}{\rm Pf} \left[ \frac{X_i-X_j}{X_i+X_j}\right]=\underset{2 n_s \times 2 n_s}{\rm Pf}\left[\iint_{C^2} \frac{\mathrm{d}w}{2\I\pi}\frac{\mathrm{d}z}{2\I\pi} \Phi_i(w)\Phi_j(z)\frac{w-z}{w+z} \right] 
\label{eq:5.12}
\end{equation}
This leads to the definition of a $2 n_s \times 2 n_s$ matrix $M$ such that 
\begin{equation}
M_{i j} =\iint_{C^2} \frac{\mathrm{d}w}{2\I\pi}\frac{\mathrm{d}z}{2\I\pi} \Phi_i(w)\Phi_j(z)\frac{w-z}{w+z}
\end{equation}
Since a variable $r_p$ will be shared between four elements of this matrix, it is more convenient to view $M$ as 
composed of $2\times 2$ blocks which we denote $K$, whose elements are presented in Eqs.~\eqref{eq:kernelAllA}.
Finally, the string-replicated partition function is given by an infinite series of Pfaffians
\begin{equation}
g(\varsigma)=1+\sum_{n_s=1}^\infty\frac{(-1)^{n_s}}{n_s!} \prod_{p=1}^{n_s} \int_\mathbb{R} \rmd r_p\frac{\varsigma}{\varsigma + e^{-r_p}} \underset{ n_s \times  n_s}{\rm Pf}\left( K(r_k,r_{\ell})\right)
\label{eq:5.14}
\end{equation}
This series is a Fredholm Pfaffian,
\begin{equation}\label{eq:PfaffAllTimesAllA}
g(\varsigma)=\mathbb{E} \left[ \exp(-\varsigma \invgamma e^{H(t)}) \right]={\rm Pf}(J-\sigma_\varsigma K)_{\mathbb{L}^2(\mathbb{R})}, 
\end{equation}
where $K$ is given in \eqref{eq:kernelAllA}, 
the function $\sigma_\varsigma$ is given by $\sigma_\varsigma(r)=\frac{\varsigma}{\varsigma+e^{-r}}$
and the $2\times 2$ kernel $J$ is given by  $J(r,r')=\bigg(\begin{array}{cc}
0 & 1 \\ 
-1 & 0
\end{array} \bigg)\mathds{1}_{r=r'}$. For the precise definition and properties of Fredholm Pfaffians see Section 8 in \cite{rains2000correlation},
as well as e.g. Section 2.2. in \cite{baik2018pfaffian}, Appendix B in 
\cite{Ortmann} and Appendix G in \cite{CLDflat, CLDflat2}.

\section{Large time limit of the Fredholm Pfaffian and the distribution of the KPZ height: crossover kernel}\label{sec:largeTlimit}

We will now study the large time limit of our kernel. To understand the scaling required at large time, let us recall the 
expression of the partition sum at fixed number of strings
\begin{equation}
\begin{split}
 {\sf Z}(n_s,\varsigma) =& (-1)^{n_s}\prod_{p=1}^{n_s} \int_\mathbb{R} \rmd r_p \frac{\varsigma}{\varsigma  + e^{ -r_p }} 
 \iint_{C^2} \frac{\rmd X_{2p-1}}{2\I\pi}  \frac{\rmd X_{2p}}{2\I\pi}
\frac{\sin(\pi X_{2p})\cos(\pi X_{2p-1})}{\pi} \\
&\times 
\frac{\Gamma(A+\frac{1}{2}-X_{2p})}{\Gamma(A+\frac{1}{2}+X_{2p})}\frac{\Gamma(A+\frac{1}{2}-X_{2p-1})}{\Gamma(A+\frac{1}{2}+X_{2p-1})}  \frac{\Gamma(B+\frac{1}{2}-X_{2p})}{\Gamma(B+\frac{1}{2}+X_{2p})}\frac{\Gamma(B+\frac{1}{2}-X_{2p-1})}{\Gamma(B+\frac{1}{2}+X_{2p-1})}\\
&\times \Gamma(2 X_{2p-1}) \Gamma(2X_{2p})  e^{-(X_{2p-1}+X_{2p}) r_p + \frac{t}{3} ( X_{2p-1}^3 + X_{2p}^3 )}  \underset{2 n_s \times 2 n_s}{\rm Pf} \left[ \frac{X_i-X_j}{X_i+X_j}\right]
\end{split}
\end{equation}
At large time, we want to eliminate the time factor in the exponential, hence we perform the change of variables
\begin{equation}
X=t^{-1/3}\tilde{X}, \qquad  r=t^{1/3}\tilde{r}, \qquad A+\frac{1}{2}=\Arescaled t^{-1/3}, \qquad B+\frac{1}{2}=\Brescaled t^{-1/3}.
\label{eq:scalings1}
\end{equation}
In the large time limit, the Gamma, cosine and sine functions simplify using that for small positive argument 
\begin{equation}
\Gamma(x)\simeq \frac{1}{x}, \qquad \cos(x)\simeq 1, \qquad \sin(x) \simeq x.
\label{eq:simpleequivalents}
\end{equation}
Under these simplifications, the partition sum at fixed number of strings reads 
in the limit $t \to +\infty$ (dropping all tildes) 
\be
\begin{split}
 {\sf Z}(n_s,\varsigma) =& (-1)^{n_s}\prod_{p=1}^{n_s} \int_\mathbb{R} \rmd r_p \frac{\varsigma}{\varsigma  + e^{ -t^{1/3}r_p }} 
 \iint_{C^2} \frac{\rmd X_{2p-1}}{2\I\pi}  \frac{\rmd X_{2p}}{2\I\pi} \frac{1}{4X_{2p-1}}
\frac{\Arescaled+X_{2p}}{\Arescaled- X_{2p}}\frac{\Arescaled+X_{2p-1}}{\Arescaled- X_{2p-1}} \\
&\times \frac{\Brescaled+X_{2p}}{\Brescaled - X_{2p}}\frac{\Brescaled +X_{2p-1}}{\Brescaled - X_{2p-1}}  e^{-(X_{2p-1}+X_{2p}) r_p +   \frac{X_{2p-1}^3}{3} + \frac{X_{2p}^3}{3} }  \underset{2 n_s \times 2 n_s}{\rm Pf} \left[ \frac{X_i-X_j}{X_i+X_j}\right]
\end{split}
\ee 
The contours $C$ have now to be understood as $C=\tilde{a}+\I\mathbb{R}$, where $\tilde{a}\in (0,\min\lbrace \Arescaled,\Brescaled\rbrace )$. We emphasize that the contours all lie at the left of the poles at $X=\Arescaled$ and $X=\Brescaled$. We may now use \eqref{eq:5.12} and \eqref{eq:5.14} to write
the Laplace transform $g(\zeta)$
under the scalings in \eqref{eq:scalings1} as the Pfaffian $\mathrm{Pf}[J-\sigma_{\zeta}K^{(a,b)}]_{L^2(\R)} $. The matrix valued kernel $K^{(\Arescaled, \Brescaled)}$ reads in this limit
\begin{equation}\label{eq:largeTimeK}
\begin{split}
&K^{(\Arescaled, \Brescaled)}_{11}(r,r')=\frac{1}{4}\iint_{C^2}\frac{\mathrm{d}w}{2\I\pi}\frac{\mathrm{d}z}{2\I\pi}\frac{w-z}{w+z}\frac{1}{wz}\frac{\Arescaled+w}{\Arescaled-w}\frac{\Arescaled+z}{\Arescaled-z} \frac{\Brescaled+w}{\Brescaled-w}\frac{\Brescaled+z}{\Brescaled-z} e^{-rw-r'z +   \frac{w^3+z^3}{3} }\\
&K^{(\Arescaled, \Brescaled)}_{22}(r,r')=\frac{1}{4}\iint_{C^2} \frac{\mathrm{d}w}{2\I\pi}\frac{\mathrm{d}z}{2\I\pi}\frac{w-z}{w+z}\frac{\Arescaled+w}{\Arescaled-w}\frac{\Arescaled+z}{\Arescaled-z} \frac{\Brescaled+w}{\Brescaled-w}\frac{\Brescaled+z}{\Brescaled-z}e^{ -rw-r'z +  \frac{w^3+z^3}{3} }\\
&K^{(\Arescaled, \Brescaled)}_{12}(r,r')=\frac{1}{4}\iint_{C^2} \frac{\mathrm{d}w}{2\I\pi}\frac{\mathrm{d}z}{2\I\pi}\frac{w-z}{w+z}\frac{1}{w}\frac{\Arescaled+w}{\Arescaled-w}\frac{\Arescaled+z}{\Arescaled-z}\frac{\Brescaled+w}{\Brescaled-w}\frac{\Brescaled+z}{\Brescaled-z} e^{ -rw-r'z +   \frac{w^3+z^3}{3} }
\end{split}
\end{equation}
\begin{remark}
The kernel $K^{(\Arescaled, \Brescaled)}$ has a particular structure, indeed its elements are related through derivative identities: $K^{(\Arescaled, \Brescaled)}_{22}(r,r')=\partial_r \partial_{r'} K^{(\Arescaled, \Brescaled)}_{11}(r,r')$,  $K^{(\Arescaled, \Brescaled)}_{12}(r,r')=-\partial_{r'} K^{(\Arescaled, \Brescaled)}_{11}(r,r')$ and $K^{(\Arescaled, \Brescaled)}_{22}(r,r')=-\partial_r K^{(\Arescaled, \Brescaled)}_{12}(r,r')$.\\
\end{remark}

\begin{remark}
The kernel $K^{(\Arescaled, \Brescaled)}$ can be obtained equivalently from the kernel
\eqref{eq:kernelAllA} by rescaling, as in \cite{AlexLD}. 
\end{remark}

Finally, choosing the variable $\varsigma$ as $\varsigma=e^{-st^{1/3}}$, at large time we have  $\lim_{t\to +\infty} \sigma_{\varsigma}(rt^{1/3})=\Theta(r-s)$, where $\Theta$ is the Theta Heaviside function. The Fredholm Pfaffian formula for the generating function then becomes in the limit
\begin{equation}\label{eq:PfaffianLargeTime}
\lim_{t\to +\infty}g(\varsigma=e^{-st^{1/3}})={\rm Pf}(J-K^{(\Arescaled, \Brescaled)})_{\mathbb{L}^2(s,+\infty)}.
\end{equation}
On the other hand, at large time, the Laplace transform of the distribution of the exponential of the KPZ height converges towards the cumulative probability of the height (see \cite[Lemma 4.1.39]{borodin2014macdonald}), i.e.
\begin{equation}
\begin{split}
g(\varsigma=e^{-st^{1/3}})=\mathbb{E}\left[ \exp(- \invgamma e^{H(t)-st^{1/3}}) \right]&\simeq_{t \to +\infty} \mathbb{E}\left[ \Theta(st^{1/3}-H(t)-\log \invgamma ) \right]\\
&\simeq_{t \to +\infty} \mathbb{P}\left(\frac{H(t)+\log \invgamma }{t^{1/3}}\leqslant s\right)
\end{split}
\end{equation}
where $\Theta$ is the Theta Heaviside function. 

Note that since $\invgamma$ is an inverse Gamma random variable with parameter $A+B+1=t^{-1/3}(\Arescaled+\Brescaled)$, the random variable $\log(\invgamma)/t^{1/3}$ weakly converges to an exponential random variable with parameter $\Arescaled+\Brescaled$. \\

At this point, we obtain that for any $\Arescaled, \Brescaled>0$, 
\begin{equation} 
\lim_{t \to +\infty} \mathbb{P}\left(\frac{H(t) }{t^{1/3}}\leqslant s-\mathcal E\right)={\rm Pf}(J-K^{(\Arescaled, \Brescaled)})_{\mathbb{L}^2(s,+\infty)}
\label{eq:shifteddistribution}
\end{equation}
where $\mathcal E$ is an exponential random variable with parameter $\Arescaled+\Brescaled$ independent from $H(t)$, and the matrix kernel $K^{(\Arescaled, \Brescaled)}$ is given in \eqref{eq:largeTimeK}. {Using the density of the exponential distribution, and denoting $F^{(a,b)}(s) = \lim_{t\to\infty }\mathbb{P}\left(\frac{H(t) }{t^{1/3}}\leqslant s\right)$, we may rewrite \eqref{eq:shifteddistribution} as  
\begin{equation}
\int_{0}^{+\infty} \mathrm d x\,   (a+b)e^{-x(a+b)} F^{(a,b)}(s-x) = {\rm Pf}(J-K^{(\Arescaled, \Brescaled)})_{\mathbb{L}^2(s,+\infty)}.
\label{eq:exponentialtransform}
\end{equation}
Following \cite[Eq. (4.3)]{ferrari2006scaling} (see also Remark \ref{rem:remark5.1}),  we differentiate in $s$ in \eqref{eq:exponentialtransform} and use integration by parts in the left hand side.  We obtain 
\begin{align*}
\partial_s {\rm Pf}(J-K^{(\Arescaled, \Brescaled)})_{\mathbb{L}^2(s,+\infty)} &= (a+b)F^{(a,b)}(s) - (a+b)^2 \int_{0}^{+\infty} \mathrm d x\,   e^{-x(a+b)} F^{(a,b)}(s-x) \\ 
  &= (a+b)F^{(a,b)}(s)- (a+b) {\rm Pf}(J-K^{(\Arescaled, \Brescaled)})_{\mathbb{L}^2(s,+\infty)}.
\end{align*}
Finally, we can write}
\begin{equation}
\lim_{t \to +\infty} \mathbb{P}\left(\frac{H(t) }{t^{1/3}}\leqslant s\right)=\left(1+\frac{\partial_s}{\Arescaled+\Brescaled} \right){\rm Pf}(J-K^{(\Arescaled, \Brescaled)})_{\mathbb{L}^2(s,+\infty)}
:= F^{(\Arescaled,\Brescaled)}(s). 
\label{eq:defFepsiloneta}
\end{equation}
In the next sections, we show that the distribution $F^{(\Arescaled,\Brescaled)}$ interpolates between various known distribution as $\Brescaled$ or $\Arescaled$ goes to infinity. The most interesting case, corresponding to stationary growth, is when $\Brescaled, \Arescaled $ go to zero and will be studied in details in Section~\ref{sec:scalarkernel}.

\begin{remark} \label{r3} 
Performing the large time limit at any fixed $A > -1/2$ and $B>-1/2$ corresponds -- modulo an exchange of limits -- to the scaling considered above with $\Arescaled,\Brescaled=+\infty$. One can show
that in the limit $\Arescaled,\Brescaled\to +\infty $, the kernel 
\eqref{eq:largeTimeK} converges to the GSE matrix kernel, 
by the same manipulations as in \cite[Section 4.1]{AlexLD}. Indeed, as the contours of kernel $K^{(\Arescaled, \Brescaled)}$ are parallel to the imaginary axis and cross the real axis between $0$ and $\min\lbrace \Arescaled,\Brescaled \rbrace$, we can take the limit $\Arescaled,\Brescaled  \to \infty$ in the integrand without affecting the contours. All rational functions involving the parameter $\Arescaled, \Brescaled$ in the large time limit of the kernel $K^{(\Arescaled, \Brescaled)}$ in \eqref{eq:largeTimeK} converge to the value $-1$. Hence in this limit we obtain the kernel $K^\infty$ given in \cite[Eq. (113)]{AlexLD}, which is precisely the kernel  associated to the Gaussian Symplectic Ensemble (GSE) of random matrices as also given in Lemma 2.7. of \cite{baik2018pfaffian}. Hence this shows that the distribution of the height at $x=0$ converges at large time for boundary conditions such that $\Arescaled, \Brescaled \to \infty$ 
(e.g. for any fixed $A,B>-1/2$) to the GSE Tracy-Widom distribution, as we will also
show below.
\end{remark}

\begin{remark} \label{r4} 
In the limit $\Brescaled\to +\infty$, the kernel $K^{(a,b)}$ in \eqref{eq:largeTimeK} converges to the 
kernel $K^{\epsilon}$ with $\epsilon=a$ obtained in \cite[Eq. (64)]{AlexLD} in the study of the droplet initial condition. The Fredholm Pfaffian $F^{(\Arescaled)}(s):={\rm Pf}(J-K^{(\Arescaled,+\infty)})_{\mathbb L^2(s, +\infty)}$ interpolates between the CDF of the GSE Tracy-Widom distribution, $F_4(s)$, (at $\Arescaled\to+\infty$) and CDF of the GOE Tracy-Widom distribution function, $F_1(s)$, (at $\Arescaled=0$). Note that
the GOE kernel obtained in \cite[Eq. (83)]{AlexLD} was found to provide a new representation of the GOE Tracy-Widom distribution. If instead we take $\Arescaled \to +\infty$ for fixed $\Brescaled$, we obtain again the same kernel $K^{\Brescaled}$ because of the symmetry $\Arescaled \leftrightarrow \Brescaled$. Physically, it describes the CDF of the distribution of the rescaled height of the KPZ equation with Brownian initial data and Dirichlet boundary condition.
\end{remark} 

\section{From a matrix valued kernel to a scalar kernel}
\label{sec:scalarkernel}

\subsection{Solution for the KPZ generating function at all times for generic $A,B$ in terms of a scalar kernel}

The general kernel we have obtained in \eqref{eq:kernelAllA} has a particular structure in the form of a Schur Pfaffian. With this structure, the kernel verifies the hypothesis of Proposition B.2 of \cite{krajenbrink2018large} recalled in Lemma~\ref{prop:2D_to_1DMOI} in Appendix~\ref{app:2D_to_1D}. This proposition states that we can transform the Fredholm Pfaffian of \eqref{eq:PfaffAllTimesAllA} which involves a matrix valued kernel, into a Fredholm determinant of a scalar kernel. To proceed, let us first define the functions 
\begin{equation}
\begin{split}
&\ratioGamma(w) = \frac{\Gamma(A+\frac{1}{2}-w)}{\Gamma(A+\frac{1}{2}+w)}\frac{\Gamma(B+\frac{1}{2}-w)}{\Gamma(B+\frac{1}{2}+w)}\Gamma(2w)\\
&f_{\rm odd}(r)=\int_C \frac{\rmd w}{2\I\pi}\ratioGamma(w)\cos(\pi w)e^{-rw +t\frac{w^3}{3}}\\
&f_{\rm even}(r)=\int_C \frac{\rmd z}{2\I\pi}\ratioGamma(z)\frac{\sin(\pi z)}{\pi}e^{-rz +t\frac{z^3}{3}}
\end{split}
\end{equation}
and the kernel $\bar{K}_{t,\varsigma}$ such that for all $(x,y)\in \mathbb{R}_+^2$
\begin{equation}\label{eq:1DkernelAllT}
\begin{split}
\bar{K}_{t,\varsigma}(x,y)&=2\partial_x \int_\mathbb{R} \rmd r \frac{\varsigma}{\varsigma+ e^{-r}}[f_{\rm even}(r+x)f_{\rm odd}(r+y)-f_{\rm odd}(r+x)f_{\rm even}(r+y)]\\
&=2\partial_x \int_\mathbb{R}\rmd r \frac{\varsigma}{\varsigma+ e^{-r}}\iint_{C^2} \frac{\rmd w \rmd z}{(2\I\pi)^2 }\ratioGamma(z)\ratioGamma(w) \frac{\sin(\pi (z-w))}{\pi}e^{-xz-yw -rw-rz + t  \frac{w^3+z^3}{3} }\\
&{=2\partial_x \iint_{C^2} \frac{\rmd w \rmd z}{(2\I\pi)^2 }\ratioGamma(z)\ratioGamma(w) \frac{\sin(\pi (z-w))}{\sin(\pi(z+w))}\varsigma^{w+z}e^{-xz-yw  + t  \frac{w^3+z^3}{3} } }
\end{split}
\end{equation}

Then, the Laplace transform of the one-point distribution of the exponential of the KPZ height admits the following representation:
\begin{equation} \label{res-scalar} 
g(\varsigma)=\mathbb{E}\left[ \exp(-\varsigma \invgamma e^{H(t)}) \right]={\rm Pf}(J-\sigma_\varsigma K)_{\mathbb{L}^2(\mathbb{R})}=\sqrt{\mathrm{Det}(I-\bar{K}_{t,\varsigma})_{\mathbb{L}^2(\mathbb{R}_+)}}.
\end{equation}

\subsection{Large time limit of the scalar kernel}
\label{subsec:largetime} 

 {To deduce the large time asymptotics of $H(t)$ from the determinantal formula  \eqref{res-scalar}, 
one performs the same rescaling as in Sec.~\ref{sec:largeTlimit},
namely one chooses $\varsigma = e^{- t^{1/3} s}$ and one rescales
$(w,z) \to t^{-1/3} (w,z)$, $r \to t^{1/3} r$.  The kernel $\bar{K}_{t,\varsigma}$ becomes 
	\begin{equation}
	\bar{K}^{(\Arescaled, \Brescaled)}(x,y) =\frac{1}{2}  \iint_{C^2} \frac{\rmd w \rmd z}{(2\I\pi)^2 }\frac{\Arescaled+w}{\Arescaled-w}\frac{\Brescaled+w}{\Brescaled-w} \frac{\Arescaled+z}{\Arescaled-z}\frac{\Brescaled+z}{\Brescaled-z}\frac{w-z}{w+z}\frac{1}{w} e^{-xz-yw +   \frac{w^3+z^3}{3} },
	\label{eq:Kbarscalaire}
	\end{equation}
 where the contour $C$ is an upwardly oriented vertical line with real part between $0$ and $\min\lbrace\Arescaled, \Brescaled \rbrace$ as previously. 	Then, using the defintion of $F^{(\Arescaled, \Brescaled)}$ from \eqref{eq:defFepsiloneta}, the determinantal formula \eqref{res-scalar} implies  the following: for any $\Arescaled, \Brescaled>0$, 
$$ F^{(\Arescaled, \Brescaled)}(s) = \lim_{t \to +\infty} \mathbb{P}\left(\frac{H(t) }{t^{1/3}}\leqslant s\right)=\left(1+\frac{\partial_s}{\Arescaled+\Brescaled} \right)\sqrt{\det(I-\bar K^{(\Arescaled, \Brescaled)})_{\mathbb{L}^2(s,+\infty)}}$$ 
where $\bar K^{(\Arescaled,\Brescaled)}$ is defined in \eqref{eq:Kbarscalaire}.  Using $\frac{1}{2} \frac{w-z}{(w+z)w} = \frac{1}{w+z} -\frac{1}{2w}$, 
	we obtain that 
	\begin{equation}
\bar K^{(\Arescaled, \Brescaled)}(x,y) = \int_{0}^{+\infty} \mathrm d \la A^{(\Arescaled, \Brescaled)}(x+\la)A^{(\Arescaled, \Brescaled)}(y+\la)\, -\frac 1 2 A^{(\Arescaled, \Brescaled)}(x) \int_{0}^{+\infty}A^{(\Arescaled, \Brescaled)}(y+\la)\, \mathrm d\la,
\label{eq:GLDform}
	\end{equation}
	where the function $A^{(\Arescaled, \Brescaled)}(x)$ is defined by 
\begin{equation}
A^{(\Arescaled, \Brescaled)}(x) = \int \frac{\mathrm d z}{2\I\pi} \frac{\Arescaled +z}{\Arescaled -z}  \frac{\Brescaled+z}{\Brescaled -z} e^{-xz+ \frac{z^3}{3}},
\label{eq:defA}
\end{equation} 
	where the contour is a vertical line with real part between $0$ and $\min \lbrace \Arescaled, \Brescaled \rbrace $. Note that the function  $A^{(\Arescaled, \Brescaled)}$ has exponential decay at $+\infty$, that is for any $c\in (0, \min \lbrace \Arescaled, \Brescaled \rbrace)$, there exist $C\in \R$ such that $\left\vert A^{(\Arescaled, \Brescaled)}(x)\right\vert\leqslant Ce^{-cx}$. 	Let us introduce an operator $\hat A_s$ acting on $\mathbb L^2(0,+\infty)$ with kernel
\begin{equation} \label{khatA} 
\hat A_s(x,y)= A^{(\Arescaled, \Brescaled)}(x+y+s), 
\end{equation}
and an operator $\bar{K}^{(a,b)}_s$  acting on $\mathbb L^2(0,+\infty)$ with kernel 
$\bar{K}^{(a,b)}_s(x, y) := \bar{K}^{(a,b)}(x+s, y+s)$.

\begin{claim} For any $s\in \R$, and $\Arescaled, \Brescaled>0$, 
\begin{equation}
\sqrt{\det(I - \bar K^{(\Arescaled, \Brescaled)}_s)} = \frac{1}{2} \left(\det(I - \hat A_s) + \det(I + \hat A_s)\right),
\end{equation} 
where all operators act on $\mathbb L^2(0,+\infty)$.
\end{claim}
\begin{proof}
Equation~\ref{eq:GLDform} implies that, as operators acting on $\mathbb L^2(0,+\infty)$, we have  
$$\bar K^{(\Arescaled, \Brescaled)}_s= \hat A_s^2 - \frac{1}{2} \vert\hat A_s \delta \rangle \langle 1 \hat A_s\vert.$$ 
At this point, we recognize that the operator $\bar K^{(\Arescaled, \Brescaled)}_s$ has the same structure as in \cite{gueudre2012directed} and we may apply the same steps as Equations (32)--(35) therein. More precisely, we may use the matrix determinant Lemma to obtain that 
\begin{equation}
\det\left( I - \bar K_s \right) = \det\left( I-\hat A_s^2 \right) \left( 1+\frac{1}{2} \langle 1\vert \frac{\hat A_s^2}{I-\hat A_s^2} \vert\delta \rangle\right). 
\label{eq:afterShermanMorrisson}
\end{equation}
Then, we use the decomposition 
$$ \frac{\hat A_s^2}{I-\hat A_s^2}  = -I+ \frac{1}{2} \left( \frac{I}{I-\hat A_s} +\frac{I}{I+\hat A_s} \right),$$
and recall that $\langle 1\vert I\vert \delta \rangle =1$. 
Using (the proof of) \cite[Proposition 1]{ferrari2005determinantal} we have

\begin{equation}
\langle 1\vert \frac{I}{I\pm\hat A_s} \vert \delta \rangle = \frac{\det\left(I\mp \hat A_s \right)}{\det\left(I\pm \hat A_s \right)}. 
\label{eq:FerrariSpohnLemma}
\end{equation}

Plugging \eqref{eq:FerrariSpohnLemma} into \eqref{eq:afterShermanMorrisson} yields the statement of the Lemma. 
\end{proof}
\begin{remark} \label{r1} 
	The identity \eqref{eq:FerrariSpohnLemma} is true for any kernel of the type $B_s(x,y)=B(x+y+s)$ such that $B$ has sufficient decay at $+\infty$ so that $\det(I\pm B_s)$ and $\langle 1\vert \frac{I}{I+ B_s} \vert \delta \rangle$ converges to $1$ as $s$ goes to $+\infty$ (see the proof of \cite[Proposition 1]{ferrari2005determinantal} for details). In our case, this condition is satisfied due to the exponential decay of the function $A^{(\Arescaled, \Brescaled)}$ for fixed $\Arescaled, \Brescaled>0$). 
\end{remark}

Hence one has 
\be
\label{eq:Det1DNoSquareRoot}
F^{(\Arescaled, \Brescaled)}(s) = \frac{1}{2} \left(1 + \frac{\partial_s}{\Arescaled+\Brescaled}\right)   \left(\det(I - \hat A_s) + \det(I + \hat A_s)\right)
\ee 
}

\subsection{Limit $\Arescaled, \Brescaled \to0$: the critical stationary case}
\label{sec:criticalstationarycase}
\subsubsection{CDF in terms of a Fredholm determinant}

Moving the contour to the right in the definition of $A^{(\Arescaled, \Brescaled)}$ in \eqref{eq:defA}, we obtain 
\begin{equation} 
A^{(\Arescaled, \Brescaled)}(x) = \tilde A^{(\Arescaled, \Brescaled)}(x) + 2 \frac{\Arescaled+\Brescaled}{\Arescaled-\Brescaled} (h_\Brescaled(x) - h_\Arescaled(x) ),
\end{equation}
with  $h_\Brescaled(x)=\Brescaled e^{- x \Brescaled + \Brescaled^3/3}$ and 
\begin{equation}
\tilde A^{(\Arescaled, \Brescaled)}(x) = \int \frac{\rmd z}{2 \I \pi} \frac{\Arescaled+z}{\Arescaled-z} \frac{\Brescaled+z}{\Brescaled-z} e^{-x z + z^3/3} 
= {\rm Ai}(x) + 2 (\Arescaled + \Brescaled) \int_x^{+\infty}\rmd \lambda \, \Ai(\lambda) + \mathcal{O}(\Brescaled^2,\Arescaled^2,\Brescaled \Arescaled), 
\end{equation}
where in the integral over $z$, the contour passes to the right of $\Brescaled, \Arescaled$. We introduce the operator $\tilde A_s$ acting on $\mathbb L^2(0,+\infty)$ with kernel 
$$ \tilde A_s (x,y)= \tilde A^{(\Arescaled, \Brescaled)}(s+x+y).$$

We thus write
\be
\hat A_s(x,y)= \tilde A_s(x,y) + 2 \frac{\Arescaled+\Brescaled}{\Arescaled-\Brescaled} ( |f_\Brescaled(x) \rangle \langle f_\Brescaled(y) | 
-  |f_\Arescaled(x) \rangle \langle f_\Arescaled(y) | ), 
\ee 
with $f_\Brescaled(x) = \sqrt{\Brescaled} e^{\Brescaled^3/6 - \Brescaled s/2 -\Brescaled x}$.
Using the matrix determinant lemma, we have 
\begin{equation} \label{dets} 
\det(I  \mp  \hat A_s) = 
\det(I \mp \tilde A_s) \left( \left(1 \mp 2 \frac{\Arescaled+\Brescaled}{\Arescaled-\Brescaled} I_{\Brescaled,\Brescaled} \right) 
\left(1 \pm 2 \frac{\Arescaled+\Brescaled}{\Arescaled-\Brescaled} I_{\Arescaled,\Arescaled} \right) + 4 \left(\frac{\Arescaled+\Brescaled}{\Arescaled-\Brescaled}\right)^2
I_{\Brescaled,\Arescaled} I_{\Arescaled,\Brescaled} \right) 
\end{equation} 
where 
{\be
I_{\alpha,\beta} = \langle f_\alpha | f_\beta \rangle \pm \langle f_\alpha | \frac{\tilde A_s}{I \mp \tilde A_s}  | f_\beta \rangle.
\ee 
}

We now consider the limit $\Brescaled,\Arescaled \to 0$ with a fixed arbitrary ratio $r=\Arescaled/\Brescaled$. We use the
exact expressions for the scalar products
\be
\langle f_\alpha | f_\beta \rangle = \frac{\sqrt{\alpha \beta}}{\alpha + \beta}  e^{\frac{\alpha^3}{6} + \frac{\beta^3}{6} - \frac{\alpha+\beta}{2} s} 
\ee 
as well as
{
\be
\bra{f_\alpha }\frac{\tilde A_s}{I \mp \tilde A_s} \ket{ f_\beta }
= \sqrt{\alpha \beta} e^{\frac{\alpha^3}{6} + \frac{\beta^3}{6} }
e^{- \frac{\alpha+\beta}{2} s} 
\bra{e^{-\alpha x} } \frac{\tilde A_s}{I \mp \tilde A_s}  \ket{ e^{-\beta x}}
\ee 
}

One has
\bea
&& I_{\Brescaled,\Brescaled}= e^{-\Brescaled s} \left( \frac{1}{2} \pm \Brescaled R^\mp_s + \mathcal{O}(\Brescaled^2) \right) \quad , \quad 
R^\mp_s =\bra{1} 
\frac{\tilde A_s}{I \mp \tilde A_s}  \ket{1} \\
&& I_{\Arescaled,\Arescaled}= e^{-\Arescaled s} \left( \frac{1}{2} \pm \Arescaled R^\mp_s + \mathcal{O}(\Arescaled^2) \right) \\
&& I_{\Brescaled,\Arescaled}= I_{\Arescaled,\Brescaled}= \sqrt{\Arescaled \Brescaled} 
e^{-(\Brescaled+\Arescaled) s/2} \left( \frac{1}{\Arescaled+\Brescaled} \pm R^\mp_s + \mathcal{O}(\Arescaled,\Brescaled) \right) 
\eea
Plugging these asymptotics  in \eqref{dets} we obtain 
\bea
&& \det(I -  \hat A_s) = \mathcal O(\Brescaled^2) \\
&& \det(I + \hat A_s) =  \det(I + \tilde A_s)\times  2\Brescaled(1+r)(2 R^{+}_s+s)+ \mathcal O(\Brescaled^2).
\eea
Furthermore, if we keep only the first order in $\Brescaled$, we may replace  $\tilde A_s(x,y)$ by $ {\rm Ai}(x+y+s)$ as $\Arescaled,\Brescaled \to 0$. Thus, we have found that 
\begin{equation}
F(s) := F^{0,0}(s) = \partial_s  \left[ \det(I + \Ai_s)(2 R^{+,0}_s+s) \right]
\label{eq:defF}
\end{equation}
where
{\be
R^{+,0}_s = \bra{1} \frac{\Ai_s }{I + \Ai_s } \ket{1}
\label{eq:defresolvant}
\ee 
}
and ${\rm Ai}_s$ is the operator with kernel ${\rm Ai}(x+y+s)$. Using the Sherman-Morrison formula, we have that 
{
\begin{equation}
\bra{1} \frac{\Ai_s }{I + \Ai_s } \ket{1} = \bra{1} \frac{I }{I + \Ai_s } \ket{\Ai_s 1} = \frac{\det(I + \Ai_s+ | \Ai_s 1 \rangle \langle   1|)}{\det(I + \Ai_s)}-1.
\end{equation}
}
which yields the following alternative formula in terms of Fredholm determinants 
\begin{equation}
F(s) = \partial_s  \left[  2 \det(I + \Ai_s+| \Ai_s 1 \rangle \langle   1| ) + (s-2)\det(I + \Ai_s)\right].
\label{eq:Fdeterminatalformula}
\end{equation}
which is given in the Introduction in \eqref{eqs}.

\subsubsection{CDF in terms of the solution of the Painlev\'e II equation}

In this Section, we show that the distribution $F(s)$ can also be written in terms of the Tracy-Widom distributions $F_1(s)$ and $F_2(s)$ only. We also provide formulae in terms of the Hastings-McLeod solution of Painlev\'e II equation (see Appendix~\ref{app:Painleve}). Define $q(s)$ to be the solution of the Painlev\'e II equation for $s\in \R$,
\begin{equation}
q''(s)= s q(s) +2q(s)^3 
\end{equation}
which satisfies the asymptotic condition $q(s) \sim_{s \to +\infty} \Ai(s)$. This solution is called the Hastings-McLeod solution of the Painlev\'e II equation \cite{hastings1980boundary}.

Let us first recall some results from \cite{imamura2004fluctuations}. 
Combining the formula 4.18 and the one just below 4.52 in \cite{imamura2004fluctuations} 
we can write
\begin{equation} \label{neweq} 
\begin{split}
\bra{1}\frac{I}{I+\Ai_s} \ket{\delta}=e^{-\int_s^{+\infty} \rmd t \, q(t)},
\end{split}
\end{equation}

\begin{remark} 
One has, see \cite[Eq. 4.22]{imamura2004fluctuations}, 
$q(s)= \bra{\delta}\frac{\Ai_s}{1-K_{\Ai,s}}\ket{\delta}$. 
\end{remark}

Another result is obtained combining Remark \ref{r1} and \eqref{neweq} as
\begin{equation}
\bra{1}\frac{I}{I\pm\Ai_s}\ket{\delta}=\frac{\Det(I\mp \Ai_s)}{\Det(I\pm \Ai_s)}=e^{\mp\int_s^{+\infty} \rmd t \, q(t)}=\left(\frac{F_1(s)^2}{F_2(s)}\right)^{\pm 1}.
\end{equation}
where the last identity is from \cite{baik2008asymptotics}.\\

The next result allows to rewrite the resolvant $R_s^{+,0}$ appearing in \eqref{eq:defF} and \eqref{eq:defresolvant} in terms of $q$. 
\begin{proposition}We have 
\begin{equation} 
\bra{1}\frac{\Ai_s}{I+\Ai_s}\ket{1}=\frac{1}{2}\left(\int_{-\infty}^s \rmd r \, e^{-2\int_r^{+\infty} \rmd t \, q(t)} -s\right).
\end{equation}
\label{prop:computationbraresolvantket}
\end{proposition}

\begin{proof}
	Before proving this proposition, we need the following known lemma. 
	\begin{lemma}[{\cite[Lemma 3]{ferrari2005determinantal}}]\label{Lemma2}
	We have the following relation
		\begin{equation}
		\partial_s\frac{\Ai_s}{I+\Ai_s}=-\frac{\Ai_s}{I-K_{\Ai,s}}  D - \frac{\Ai_s}{I-K_{\Ai,s}}\ket{\delta} \bra{\delta} \frac{I}{I+\Ai_s}
		\end{equation}
		where $D$ is the derivative operator.
	\end{lemma}
Since $D\ket{1}=0$ we have 
\begin{equation}
\begin{split}
\partial_s \bra{1}\frac{\Ai_s}{I+\Ai_s}\ket{1}&=-\bra{1}\frac{\Ai_s}{I-K_{\Ai,s}}\ket{\delta}\bra{\delta}\frac{I}{I+\Ai_s}\ket{1}\\
&= -\frac{1}{2}\bra{1}\left(\frac{I}{I-\Ai_s}-\frac{I}{I+\Ai_s}\right)\ket{\delta}\bra{\delta}\frac{I}{I+\Ai_s}\ket{1}\\
&= -\frac{1}{2}\left( \frac{\Det(I+\Ai_s)}{\Det(I-\Ai_s)}-\frac{\Det(I-\Ai_s)}{\Det(I+\Ai_s)} \right)\times \frac{\Det(I-\Ai_s)}{\Det(I+\Ai_s)}\\
&= -\frac{1}{2}\left( 1-\frac{F_1(s)^4}{F_2(s)^2} \right)
\end{split}
\end{equation}
Using the expressions of $F_1$ and $F_2$ in terms of the Painlevé II equation, we have
\begin{equation}
\begin{split}
\partial_s \bra{1}\frac{\Ai_s}{I+\Ai_s}\ket{1}
&=-\frac{1}{2}\left( 1-e^{-2\int_s^{+\infty} \rmd t \, q(t)} \right)\\
\end{split}
\end{equation}
Recalling the notation $R_s^{+,0}=\bra{1}\frac{\Ai_s}{1+\Ai_s}\ket{1}$, we have 
\begin{equation}
\partial_s(2R_s^{+,0}+s)=e^{-2\int_s^{+\infty} \rmd t \, q(t)} =\frac{F_1(s)^4}{F_2(s)^2}
\end{equation}
 We integrate the last quantity between $-\infty$ and $s$ using that $q(s) \to +\infty$ for $s \to -\infty$ as in 
 \eqref{eq:leftasymptoticsq} and obtain 
\begin{equation}
2R_s^{+,0}+s=\int_{-\infty}^{s} \rmd r \, e^{-2\int_r^{+\infty} \rmd t \, q(t)} + \kappa.
\end{equation}
Since $R_s^{+,0} \to 0$ in the limit $s\to +\infty$, and using \eqref{eq:equivalence} and the  asymptotics \eqref{eq:rightasymptotics1} and \eqref{eq:rightasymptotics2}, we obtain that $\kappa=0$, which concludes the proof. 
\end{proof}

Using Proposition~\ref{prop:computationbraresolvantket} we arrive at the following equivalent expressions for $F(s)$ (defined in \eqref{eq:defF}).  The first one reads
\begin{equation}
F(s)=\partial_s\left[\frac{F_2(s)}{F_1(s)} \int_{-\infty}^s \rmd t \, \frac{F_1(t)^4}{F_2(t)^2} \right]. 
\label{eq:FintermsofF1F2}
\end{equation}
where the first line was given in \eqref{eq:FintermsofF1F2intro}.
The second formula is expressed in terms of the Hastings–McLeod solution of the Painlevé II equation (see Appendix~\ref{app:Painleve}) and reads
\begin{equation} 
\begin{split}
F(s)&=\partial_s\left[e^{ -\frac{1}{2} \int_s^{+\infty} \mathrm d r [ (r-s)q(r)^2-q(r) ]} \int_{-\infty}^{s} \rmd r \, e^{-2\int_r^{+\infty} \rmd t \, q(t)} \right]\\
&=\partial_s\left[e^{ -\frac{1}{2} \int_s^{+\infty} \mathrm d r \,  [(r-s)q(r)^2+3q(r)]}(s-2q'(s)+2q(s)^2) \right] . 
\end{split}
\end{equation}
which was given in \eqref{pasdelabel}. A third formula, useful for the asymptotics, is obtained 
applying the derivative in front of the bracket and using Eqs.~\eqref{eq:equivalence} and \eqref{eq:PrimitivePainleve}
\begin{equation}
\begin{split}
F(s)=&\; e^{ -\frac{1}{2} \int_s^{+\infty} \mathrm d r \, [ (r-s)q(r)^2+3 q(r)]} \\
& \times  \left( 1+\frac{1}{2}(q'(s)^2-sq(s)^2-q(s)^4-q(s))(s-2q'(s)+2q(s)^2)\right)
\end{split}
\label{eq:CDFwithoutDerivative}
\end{equation}

\subsubsection{Properties of $F(s)$: first moments}

We check in Appendix~\ref{app:Painleve} using the formulae \eqref{eq:CDFwithoutDerivative}
that the function $F(s)$ has the behaviour at $ s \to \pm \infty$ that is required for a CDF, i.e. its limit at $s= -\infty$ is $0$ and its limit at $s=+\infty$ is $1$. The detailed asymptotics for $s \to \pm \infty$ is
performed in Appendix \ref{app:F}. The CDF takes the form $F(s)=\partial_s [{\sf F}(s)]$.
Provided ${\sf F}(s) \to 0$ sufficiently fast for $s \to -\infty$ and 
${\sf F}(s) - s \to 0$ sufficiently fast for $s \to +\infty$, conditions which can be checked
from Appendix \ref{app:F}, integration by parts give the following formula for the 
$k$-th positive integer moment $M_k$ of the distribution $F(s)$ 
\be 
M_k = k(k-1)  \int_\R ({\sf F}(s)-\max(s,0)) s^{k-2} \mathrm d s 
\ee 
This formula can be used to obtain the moments and the cumulants through a numerical evaluation of
$F(s)$. One notices that the mean vanishes, $M_1=0$, which indeed must be the case since $\mathbb E[H(t)]=0$ in the critical stationary case (see Section \ref{sec:convergenceloggammaSHE} where we have computed the expectation of $h(x,t)$ using the stationary structure of the log-gamma polymer). We used two
numerical methods to evaluate $F(s)$. The first one uses Eq. \eqref{eq:Fdeterminatalformula}
where the Fredholm determinants are calculated using the method described in Ref.
\cite{Bornemann,Bornemann2}. The second method uses the formula 
\eqref{eq:FintermsofF1F2} and uses the Mathematica routines for $F_{1,2}(s)$, and is
in agreement with the first one. The CDF $F(s)$ and its derivative $F'(s)$ are plotted 
in Fig. \ref{fig:plots}. The mean, variance, skewness and excess kurtosis are given
in Table \ref{table:moments of the distribution}.

\subsection{Limit $\Arescaled, \Brescaled \to +\infty$: convergence to the GSE}

As we already discussed in Remark \ref{r3} the limit $\Arescaled, \Brescaled \to +\infty$ can be performed on the Fredholm pfaffian formula and leads to GSE Tracy-Widom fluctuations. This limit can also be performed
on the formula \eqref{eq:Det1DNoSquareRoot} involving the scalar kernel $\hat A_s$ defined in \eqref{khatA} in terms of the function $A^{(\Arescaled, \Brescaled)}(x)$ defined in \eqref{eq:defA}.
It is clear from the definition of $A^{(\Arescaled, \Brescaled)}(x)$ that if $\Arescaled,\Brescaled \to +\infty$ simultaneously, then $A^{(\Arescaled, \Brescaled)}(x)$
converges to the standard Airy function. The CDF $F^{(a,b)}(s)$ in \eqref{eq:Det1DNoSquareRoot}, for 
$\Arescaled, \Brescaled \to +\infty$ then takes the form of the GSE Tracy-Widom distribution found in \cite[Eq (35)]{GueudrePLD}. This result thus matches smoothly with the result \eqref{GSE1} valid for any fixed $A,B>-1/2$ in the large time limit.

\subsection{Limit $(\Arescaled,\Brescaled) \to (0,+\infty)$: convergence to the GOE}
Another interesting limit, that we call $F^{(\Arescaled)}(s)= \lim_{\Brescaled \to +\infty} F^{(\Arescaled,\Brescaled)}(s)$,
is the limit $\Brescaled \to + \infty$ at fixed $\Arescaled$ and by the $A \leftrightarrow B$ symmetry,  the case $A\to+\infty$ at fixed $\Brescaled $ is similar. In particular we consider now the limit $\Arescaled\to 0$.\\

The manipulations follow closely the ones of Section.~\ref{sec:criticalstationarycase}. We start by moving the contour to the right in the definition of $A^{(\Arescaled, \infty)}$ in \eqref{eq:defA} already taking into account the $b\to+\infty$ limit. We obtain $ A^{(\Arescaled, \infty)}(x) = \breve A^{(\Arescaled, \infty)}(x) + 2  h_\Arescaled(y)$, 
with  $h_\Arescaled(x)=\Arescaled e^{- x \Arescaled + \Arescaled^3/3}$ and 
\begin{equation}
\breve A^{(\Arescaled, \infty)}(x) = \int \frac{\rmd z}{2 \I \pi} \frac{\Arescaled+z}{\Arescaled-z}  e^{-x z + z^3/3} 
=- {\rm Ai}(x) - 2 \Arescaled  \int_x^{+\infty}\rmd \lambda \, \Ai(\lambda) + \mathcal{O}(\Arescaled^2), 
\end{equation}
where in the integral over $z$, the contour passes to the right of $\Arescaled$. We introduce the operator $\breve{A}_s$ acting on $\mathbb L^2(0,+\infty)$ with kernel  $ \breve{A}_s (x,y)= \breve A^{(\Arescaled, \infty)}(s+x+y)$. 
We thus write
\be
\hat A_s(x,y)= \breve{A}_s(x,y) + 2   |f_\Arescaled(x) \rangle \langle f_\Arescaled(y) | , 
\ee 
with $f_\Arescaled(x) = \sqrt{\Arescaled} e^{\Arescaled^3/6 - \Arescaled s/2 -\Arescaled x}$.
Using the matrix determinant lemma, we have 
\begin{equation} \label{dets222} 
\det(I  \pm  \hat A_s) = 
\det(I \pm \breve{A}_s) \left(1 \pm 2  \langle f_a| f_a \rangle -  2 \langle f_a | \frac{\breve{A}_s}{I \pm \breve{A}_s}  | f_a \rangle \right) 
\end{equation} 

We now consider the limit $\Arescaled \to 0$ and we use the
exact expressions for the scalar product $\langle f_a | f_a \rangle = \frac{1}{2}  e^{\frac{a^3}{3} - a s} $ as well as
\be
\bra{f_a }\frac{ \breve{A}_s}{I \pm  \breve{A}_s} \ket{ f_a }
= a e^{\frac{a^3}{3}  -as} 
\bra{e^{-a x} } \frac{\breve{A}_s}{I \pm \breve{A}_s}  \ket{ e^{-a x}}
\ee 
One has
\begin{equation} \label{dets333} 
\det(I  \pm  \hat A_s) = 
\det(I \pm \breve{A}_s) \left(1 \pm e^{-as} -  2 e^{-as} a \bra{1} 
\frac{\breve{A}_s}{I \pm \breve{A}_s}  \ket{1}  +\mathcal{O}(a^2)\right) 
\end{equation} 
Looking at formula \eqref{eq:Det1DNoSquareRoot} in the limit $b \to +\infty$
we see that we need only the following estimates from \eqref{dets333}  up to $\mathcal O(a)$ as $a \to 0$
\bea
&& \det(I -  \hat A_s) = \mathcal O(\Arescaled) \\
&& \det(I + \hat A_s) =  \det(I-\Ai_s)(2+\mathcal{O}(a)).
\eea
which, inserted in \eqref{eq:Det1DNoSquareRoot}, lead to
\begin{equation}
F^{(0)}(s):=F^{(0,\infty)}(s)=F^{(\infty,0)}(s)=\det(I-\Ai_s) = F_1(s)
\end{equation}
This coincides with the determinantal representation of the GOE Tracy-Widom CDF. 
This formula matches smoothly with the result \eqref{GOE11} which states that the CDF of the one-point KPZ height field is given by the GOE Tracy-Widom distribution for $A=-1/2$ and any $B>-1/2$.

\newpage
\appendix
\begin{center}
{\bf \Large Appendix}
\end{center}

\addcontentsline{toc}{section}{Appendix}

\section{Overlap of the half-line Bethe states with the Brownian initial condition} 
\label{app:overlap}

Here we give some details on the calculation of the overlap 
$\langle \Psi_\mu | \Phi_0 \rangle$ between the half-line Bethe states and the Brownian
initial condition. We recall that $\Phi_0$, given in \eqref{initial}, is a fully symmetric function of its arguments,
and that in the sector $0 \leqslant x_1\leqslant \dots \leqslant x_n$ it equals
\be 
\Phi_0(x_1,\dots,x_n)  = \exp\left(\frac{1}{2}\sum_{j=1}^n (2n-2j+1)x_j - (1/2+B) x_j \right), 
\ee
where $b$ is the drift of the Brownian. Since we will find that the
overlap is real, we will instead calculate its complex conjugate and use that
$\langle \Psi_\mu | \Phi_0 \rangle^*= \langle \Phi_0 |\Psi_\mu  \rangle = \langle \Psi_\mu | \Phi_0 \rangle$.
Since $\Psi_\mu$ is also a symmetric
function of its arguments, by definition the overlap
can thus be written as
\be
\langle \Phi_0 |\Psi_\mu  \rangle =
n! 
\int_{0<y_1<y_2<\dots<y_n} \, \rmd y_1\dots \rmd y_n \, \Psi_\mu(y_1,\dots,y_n) 
e^{\sum_{j=1}^n \frac{1}{2}(2n +1-2j)y_j - (1/2+B) y_j}\
\ee 
Inserting the explicit form of the Bethe eigenstate \eqref{wave} as a superposition of plane waves
we obtain
\bea \label{over1} 
&& \langle \Phi_0 |\Psi_\mu  \rangle =
\frac{n!}{(2 \I)^{n}}  
\sum_{P \in S_n} 
\prod_{p=1}^n \left( \sum_{\varepsilon_p=\pm 1} \varepsilon_p  \right) \prod_{\ell=1}^n \left(1+\I\frac{\varepsilon_\ell \lambda_{P(\ell)}}{A}\right) \\
&& \times 
\prod_{k<\ell}(1+\frac{\I}{\varepsilon_{\ell} \lambda_{P(\ell)}-
\varepsilon_{k}  \lambda_{P(k)}})\, (1+\frac{\I}{\varepsilon_{\ell} \lambda_{P(\ell)} +
\varepsilon_{k}  \lambda_{P(k)}})
G_{n,w}(\varepsilon_{1}  \lambda_{P(1)},\dots, \varepsilon_{n}\lambda_{P(n)}) \nonumber 
\eea 
where we have defined the integrals
\be
G_{n,w}(\lambda_1,\dots,\lambda_n) =\int_{0<y_1<y_2<\dots<y_n}
\rmd y_1\dots \rmd y_p \, e^{\sum_{j=1}^n (-B-1/2+\I \lambda_j)y_j+\frac{1}{2}(2n+1-2j)y_j}
\ee
These integrals can be explicitly evaluated
\be
G_{n,w}(\lambda_1,\dots,\lambda_n) =\prod_{j=1}^n \frac{-1}{-j (B+1/2)+\I\lambda_n+\dots+\I \lambda_{n+1-j}+ j^2/2}
\ee
Now in \eqref{over1} for each permutation $P$ we can relabel all the $\varepsilon_p \to \varepsilon_{P(p)}$ 
and denoting by $\sum_{\varepsilon = \lbrace \pm 1\rbrace^n}$ the operation of summation over
all the variables $\varepsilon_i$ (an operation independent of their labeling) we can rewrite \eqref{over1} as
\bea
&& \langle \Phi_0 |\Psi_\mu  \rangle =
\frac{n!}{(2 \I)^{n}} 
\sum_{\varepsilon = \lbrace \pm 1\rbrace^n} \prod_{\ell=1}^n \varepsilon_\ell \left(1+\I\frac{\varepsilon_\ell \lambda_{\ell}}{A}\right)
\prod_{k<\ell}\left(1+\frac{\I}{\varepsilon_k \lambda_k+\varepsilon_\ell \lambda_\ell}\right)   \\
&& 
\sum_{P\in S_n}\prod_{k<\ell}\left(1+\frac{\I}{\varepsilon_{P(\ell)} \lambda_{P(\ell)}-
\varepsilon_{P(k)}  \lambda_{P(k)}}\right)\, G_{n,w}(\varepsilon_{P(1)}  \lambda_{P(1)},\dots, \varepsilon_{P(n)}\lambda_{P(n)})
\nonumber 
\eea 
where we have used the fact that the products
\be
\prod_{k<\ell} \left(1+\frac{\I}{\varepsilon_{P(\ell)} \lambda_{P(\ell)} +
\varepsilon_{P(k)}  \lambda_{P(k)}}\right) = \prod_{k<\ell}\left(1+\frac{\I}{\varepsilon_k \lambda_k+\varepsilon_\ell \lambda_\ell}\right) 
\ee 
and
\begin{equation}
\prod_{\ell=1}^n \left(1+\I\frac{\varepsilon_{P(\ell)} \lambda_{P(\ell)}}{A}\right)=\prod_{\ell=1}^n \left(1+\I\frac{\varepsilon_{\ell} \lambda_{\ell}}{A}\right)
\end{equation}
are independent of the permutation $P$.  Now we use the following symmetrization identity, given in \cite{SasamotoHalfBrownian} (this is a limit of \cite[Eq. (9)]{tracy2009asep} which was a slight generalization of \cite[Eq.~(1.6)]{tracy2008integral}),
\begin{equation} \label{miracle} 
\sum_{P\in S_n}\prod_{k<\ell}\left(1+\frac{\I}{\lambda_{P(\ell)}-\lambda_{P(n)}}\right) \, G_{n,w}(\lambda_{P(1)},\dots,\lambda_{P(n)})=\prod_{j=1}^n \frac{1}{B-\I\lambda_j}
\end{equation}
Applying it to the set $\{ \varepsilon_k \lambda_k \}$ we obtain
\be \label{over2}
\langle \Phi_0 |\Psi_\mu  \rangle =
\frac{n!}{(2 \I)^{n}}  
\sum_{\varepsilon = \lbrace \pm 1\rbrace^n} \prod_{k<\ell}\left(1+\frac{\I}{\varepsilon_k \lambda_k+\varepsilon_\ell \lambda_\ell}\right) 
\prod_{j=1}^n \frac{\varepsilon_j+\I\frac{\lambda_j}{A}}{B-\I \varepsilon_j \lambda_j}.
\ee
It remains to perform the symmetrization over $\varepsilon$.

\begin{lemma}
For any set of
complex numbers $\lbrace \lambda_j\rbrace_{1\leqslant j \leqslant n}$, 
\begin{equation}
\sum_{\varepsilon = \lbrace \pm 1\rbrace^n}\prod_{k<\ell}\left(1+\frac{\I}{\varepsilon_k \lambda_k+\varepsilon_\ell \lambda_\ell}\right)\prod_{j=1}^n \frac{\varepsilon_j+\I\frac{\lambda_j}{A}}{B-\I\varepsilon_j \lambda_j}=\left(\frac{2\I}{A}\right)^n \frac{\Gamma(A+B+1)}{\Gamma(A+B-n+1)} \prod_{j=1}^n \frac{\lambda_j}{B^2+\lambda_j^2}.
\label{eq:symmetrizationweuse}
\end{equation}
\label{lem:symetrization}
\end{lemma}
\begin{remark}
	In the limit $A\to +\infty$, this formula yields back the overlap for Dirichlet boundary condition, see \cite{krajenbrink2018large}.
\end{remark}
\begin{proof}
	This identity was essentially known in the context of Hall-Littlewood polynomials of type BC. In particular, a generalization of it was proved in \cite[Theorem 2.6]{venkateswaran2015symmetric}. We refer to \cite[Eq. (54)]{borodin2016directed} for more details about how to degenerate Venkateswaran's symmetrization identity to the one we need (trigonometric to rational limit). For any parameters $A,B,C \in \mathbb C$ and complex variables $(z_i)_{1\leqslant i\leqslant n}$, we have 
	\begin{multline}
		\sum_{\varepsilon \in \lbrace \pm1\rbrace^n} \sum_{P\in S_n} \prod_{k<l} \frac{\varepsilon_{P(k)}z_{P(k)}+\varepsilon_{P(l)}z_{P(l)}-C}{\varepsilon_{P(k)}z_{P(k)}+\varepsilon_{P(l)}z_{P(l)}}\frac{\varepsilon_{P(k)}z_{P(k)}-\varepsilon_{P(l)}z_{P(l)}-C}{\varepsilon_{P(k)}z_{P(k)}-\varepsilon_{P(l)}z_{P(l)}}\\ \times \prod_{j=1}^n \frac{(\varepsilon_{P(j)}z_{P(j)}+A)(\varepsilon_{P(j)}z_{P(j)}+B)}{\varepsilon_{P(j)}z_{P(j)}} = 2^n n! \prod_{j=0}^{n-1}(A+B-j).
		\label{eq:BCsymetrization}
	\end{multline}
	Let us perform first the symmetrization over $P\in S_n$. Since the first product in the left hand side of \eqref{eq:BCsymetrization} is $S_n$-invariant, and using the symmetrization identity (\!\cite[Chap. III, Eq. (1.4)]{macdonald1995symmetric}) 
	\begin{equation*}
		\sum_{P\in S_n}  \frac{\varepsilon_{P(k)}z_{P(k)}-\varepsilon_{P(l)}z_{P(l)}-C}{\varepsilon_{P(k)}z_{P(k)}-\varepsilon_{P(l)}z_{P(l)}} = n!, 
	\end{equation*} 
	we obtain that 
	\begin{equation}
		\sum_{\varepsilon \in \lbrace \pm1\rbrace^n} \prod_{k<l} \frac{\varepsilon_kz_k+\varepsilon_lz_l-C}{\varepsilon_kz_k+\varepsilon_lz_l}\prod_{j=1}^n \frac{(\varepsilon_jz_j+A)(\varepsilon_jz_j+B)}{\varepsilon_jz_j} = 2^n   \prod_{j=0}^{n-1}(A+B-j).
		\label{eq:signsymetrization}
	\end{equation}
	We obtain \eqref{eq:symmetrizationweuse} using the substitutions $C\to 1, z_j\to \I\lambda_j $ in \eqref{eq:signsymetrization}. 
\end{proof}

Applying Lemma~\ref{lem:symetrization} in \eqref{over2}, we obtain the formula for the overlap \eqref{over0} given in the text. \\

We note that the overlap formula it is a priori valid before the insertion of the solution of the Bethe equations,
i.e. it is valid for any set of complex $\lambda_j$ such that the overlap integral converges. The condition for that to be true can be read
from \eqref{miracle} as
\be
{\rm Re}(\I \lambda_j) < B
\ee 
for all $j \in [1,n]$. Inserting a string state, labeled by $\{ k_j, m_j\}_{j=1,\dots,n_s}$, the condition becomes
$\max_j \frac{1}{2}(2 m_j-1) \leq \frac{1}{2}(2 n -1) < B$, which leads to the condition
$\frac{n}{2} < B+ \frac{1}{2}$. 

\section{From two-dimensional kernels to scalar kernels}\label{app:2D_to_1D}

We present in this section an equivalent representation of a class of Fredholm Pfaffians with $2\times 2$ block kernels in terms of a Fredholm determinant with a scalar valued kernel. Consider a measure $\rmd \mu$ on a contour $C$ in the complex plane and another measure $\rmd \nu_\varsigma$ on the real line $\mathbb{R}$, depending on a real parameter $\varsigma$. Consider the quantity $g(\varsigma)$ defined by the series
\begin{equation}
g(\varsigma)=1+\sum_{n_s= 1}^\infty \frac{(-1)^{n_s}}{n_s!} Z(n_s,\varsigma) 
\end{equation}
and
\begin{equation}
\begin{split}
 Z(n_s,\varsigma) = \prod_{p=1}^{n_s} \int_{\mathbb{R}}\rmd \nu_\varsigma( r_p)
 &\iint_{C^2} \rmd  \mu(X_{2p-1}) \rmd \mu( X_{2p})\\
& \phi_{\mathrm{odd}}(X_{2p-1}) \phi_{\mathrm{even}}(X_{2p})e^{-r_p[X_{2p-1}+X_{2p}]  } \;  {\rm Pf} \left[ \frac{X_i-X_j}{X_i+X_j}\right]_{i,j=1}^{2n_s} 
\end{split}
\end{equation}
Then, following Ref.~ \cite{krajenbrink2018large}, we have the following equivalent representations for $g(\varsigma)$.

\begin{lemma}[Fredholm Schur Pfaffian]\label{lemma:BRUIJN}
$g(\varsigma)$ is equal to a Fredholm Pfaffian with a $2\times2$ matrix valued skew-symmetric kernel
\begin{equation}
g(\varsigma)=\mathrm{Pf}( J-K)_{\mathbb{L}^2(\mathbb{R}, \nu_\varsigma)}
\end{equation}
For $(r,r')\in \mathbb{R}^2$ the matrix kernel $K$ is given by
\begin{equation}
\label{app:K_block}
\begin{split}
&K_{11}(r,r')=\iint_{C^2} \rmd \mu(v)\rmd \mu(w) \frac{v-w}{v+w}\phi_{\mathrm{odd}}(v)\phi_{\mathrm{odd}}(w)e^{ -rv-r'w }\\
&K_{22}(r,r')=\iint_{C^2}\rmd \mu(v)\rmd \mu(w) \frac{v-w}{v+w}\phi_{\mathrm{even}}(v)\phi_{\mathrm{even}}(w)e^{ -rv-r'w}\\
&K_{12}(r,r')=\iint_{C^2}\rmd \mu(v)\rmd \mu(w) \frac{v-w}{v+w}\phi_{\mathrm{odd}}(v)\phi_{\mathrm{even}}(w)e^{ -rv-r'w  }\\
&K_{21}(r,r')=\iint_{C^2}\rmd \mu(v)\rmd \mu(w) \frac{v-w}{v+w}\phi_{\mathrm{even}}(v)\phi_{\mathrm{odd}}(w)e^{ -rv-r'w  }
\end{split}
\end{equation}
and the matrix kernel $J$ is defined by $J(r,r')=\bigg(\begin{array}{cc}
0 & 1 \\ 
-1 & 0
\end{array} 
\bigg)\mathds{1}_{r=r'}$.
\end{lemma}

\begin{lemma}[Scalar Fredholm determinant - Proposition B.2 of Ref.~\cite{krajenbrink2018large}]\label{prop:2D_to_1DMOI}
$g(\varsigma)$ is equal to the square root of a Fredholm determinant with scalar valued kernel
\begin{equation}
g(\varsigma)=\sqrt{\mathrm{Det}( I-\bar{K})_{\mathbb{L}^2(\mathbb{R_+})}}
\end{equation}
where $\mathbb{L}^2(\mathbb{R}_+)$ is considered with the Lebesgue measure on $\R_+$. Introducing the functions
\begin{equation}\label{eq:def_f_even_odd}
\begin{split}
f_{\mathrm{odd}}(r)=\int_{C}\rmd \mu(v) \, \phi_{\mathrm{odd}}(v)e^{ -rv }, \quad f_{\mathrm{even}}(r)=\int_{C}\rmd \mu(v) \, \phi_{\mathrm{even}}(v)e^{ -rv }
\end{split}
\end{equation}
which are assumed to be in $\mathbb{L}^2(\mathbb{R}_+)$, the scalar kernel $\bar{K}$ is given, for $(x,y)\in \mathbb{R}^2_+$, by
\begin{equation} \label{eq:1D_kernel}
\bar{K}(x,y)=2\partial_x\int_{\mathbb{R}}\rmd \nu_\varsigma( r)\, \left[\mathrm{f}_{\mathrm{even}}(x+r) \mathrm{f}_{\mathrm{odd}}(r+y)- \mathrm{f}_{\mathrm{odd}}(x+r)\mathrm{f}_{\mathrm{even}}(r+y) \right] 
\end{equation}
and the scalar kernel $I$ is the identity kernel $I(x,y)=\mathds{1}_{x=y}$.
\end{lemma}

\section{Painlev\'e II equation and Tracy-Widom distributions}
\label{app:Painleve}

\subsection{The Hastings–McLeod solution of the Painlev\'e II equation}
Define $q(s)$ to be the solution of the Painlev\'e II equation for $s\in \R$,
\begin{equation}
q''(s)=sq(s)+2q(s)^3 
\end{equation}
which satisfies the asymptotic condition $q(s) \sim_{s \to +\infty} \Ai(s)$. The unique smooth solution of the Painlev\'e II equation with that asymptotic condition is called the Hastings-McLeod solution
\cite{hastings1980boundary}.  Let us indicate the two following formula which we use in the manuscript. The first one
is given in Ref.~\cite[Eq. (2.18)]{png}
\begin{equation}
\label{eq:equivalence}
s-2q'(s)+2q(s)^2=e^{2\int_s^\infty \rmd r \, q(r)}\int_{-\infty}^s \rmd r \, e^{-2\int_r^\infty \rmd r \, q(r)}
\end{equation}
The second one is given in Ref.~\cite[Eqs. (2.5), (2.6)]{png}
\begin{equation}
\label{eq:PrimitivePainleve}
\int_s^{+\infty}\rmd r \, q(r)^2=q'(s)^2-sq(s)^2-q(s)^4 
\end{equation}

\subsubsection{Left asymptotics}
\label{sec:leftasymptoticsq}
Useful asymptotics are given \cite{baik2008asymptotics} and

\begin{equation} 
q(s)=_{s\to -\infty}\sqrt{-\frac{s}{2}}\left(1+\frac{1}{8s^3}-\frac{73}{128s^6}+\frac{10657}{1024 s^9}+o\left(\frac{1}{s^9}\right)\right)
\label{eq:leftasymptoticsq}
\end{equation}

\subsubsection{Right asymptotics}

We also needed the following right asymptotics in the proof of Proposition~\ref{prop:computationbraresolvantket}, see \cite[Section 1.3]{nadalRight}.
\begin{equation}
\int_s^{+\infty} \rmd r \,(r-s) q(r)^2=_{s\to +\infty}\frac{e^{-4/3 s^{3/2}}}{15\pi s^{3/2}}\left(1-\frac{35}{24s^{3/2}}+\mathcal{O}\left(\frac{1}{s^3}\right)\right)
\label{eq:rightasymptotics1}
\end{equation}
and
\begin{equation}
\int_s^{+\infty}\rmd r \, q(r)=_{s\to +\infty}\frac{e^{-2/3s^{3/2}}}{2\sqrt{\pi}s^{3/4}}\left(1-\frac{41}{48s^{3/2}}+\mathcal{O}\left(\frac{1}{s^3}\right)\right)
\label{eq:rightasymptotics2}
\end{equation}

\subsection{Relations between the Tracy-Widom, Baik-Rains distributions and the Hastings-McLeod solution of the Painlev\'e II equation}
 The Tracy-Widom distributions for $\beta=2,1,4$ and the Baik-Rains distribution (denoted BR)  are given by
 \cite{png,ferrari2005determinantal}
\begin{itemize}
\item For $\beta=2$
\begin{equation}
\begin{split}
F_2(s)&=\exp \left( -\int_s^{+\infty} \rmd r (r-s) q^2(r) \right)=\Det(I-\Ai_s^2)_{\mathbb{L}^2(\R_+)}
\end{split}
 \end{equation} 
\item For $\beta=1$ 
\begin{equation}
\begin{split}
F_1(s)&=\exp\left(-\frac{1}{2}\int_s^{+\infty} \rmd r [(r-s) q^2(r)+q(r)]\right)=\Det(I-\Ai_s)_{\mathbb{L}^2(\R_+)}
\end{split}
\end{equation}
\item For $\beta=4$ (everywhere in the paper we use the conventions of \cite{baik2008asymptotics})
\begin{equation}
\begin{split}
F_4(
s)&=\exp \left( -\frac{1}{2}\int_s^{+\infty} \rmd r (r-s) q^2(r) \right)\cosh\left(\frac{1}{2}\int_s^{+\infty} \rmd r \, q(r)\right)\\
&=\frac{1}{2}(\Det(I-\Ai_s)_{\mathbb{L}^2(\R_+)}+\Det(I+\Ai_s)_{\mathbb{L}^2(\R_+)})
\end{split}
\end{equation}
\item For the Baik-Rains distribution 
\begin{equation}
\begin{split}
&F_{\rm BR}(s)\\
&=(1+(s-2q'(s)+2q^2(s))(\int_s^{+\infty} \rmd r q(r)^2))\exp \left( -\int_s^{+\infty} \rmd r [(r-s) q^2(r)+2q(r)] \right)\\
&=\partial_s \left[\int_{-\infty}^{s} \rmd t\,  \exp \left( -2\int_t^{+\infty} \rmd r \,q(r) \right)\exp \left( -\int_s^{+\infty} \rmd r (r-s) q^2(r) \right)\right]
\end{split}
\end{equation}
\end{itemize}
It allows to obtain the following relation between the distributions:
\begin{equation}
F_4(
s)=\frac{1}{2}\left(F_1(s)+\frac{F_2(s)}{F_1(s)}\right)
\end{equation}
and 
\begin{equation}
F_{\rm BR}(s)=\partial_s \left[ F_2(s) \int_{-\infty}^s \rmd t \, \frac{F_1(t)^4}{F_2(t)^2} \right].
\label{eq:defBaikRains}
\end{equation}
The next sections provide tail asymptotics for $F_1$ and $F_2$. 
\subsubsection{Left asymptotics of $F_1$ and $F_2$}
The left asymptotics of the Tracy-Widom $\beta=1,2$ are given by (see e.g. \cite{nadalLeft} and references therein)
\begin{equation}
\label{eq:leftasymptoticsF1F2}
\begin{split}
F_1(s) & =  2^{-11/48}\,e^{\zeta'(-1)/2}\,\exp\Big[ -\frac{|s|^3}{24} - \frac{|s|^{3/2}}{3\sqrt{2}} - \frac{\log |s|}{16} - \frac{|s|^{-3/2}}{24\sqrt{2}}  \\
 & \phantom{2^{-11/48}\,e^{\zeta'(-1)/2}\,\exp}\, + \frac{3|s|^{-3}}{128} - \frac{73|s|^{-9/2}}{1152\sqrt{2}} + \frac{63|s|^{-6}}{512} + \mathcal{O}\big(|s|^{-15/2}\big)\Big] \\
F_2(s) & =  2^{1/24}\,e^{\zeta'(-1)}\,\exp\Big[-\frac{|s|^3}{12} - \frac{\log |s|}{8} + \frac{3|s|^{-3}}{64} + \frac{63|s|^{-6}}{256} + \mathcal{O}\big(|s|^{-9}\big)\Big] 
\end{split}
\end{equation}
\subsubsection{Right asymptotics of $F_1$}
The right asymptotics of the Tracy-Widom $\beta=1,2$ are given by (see e.g. \cite{nadalRight} and references therein)
\begin{equation}
\label{eq:TWright1}
\begin{split}
1 - F_{1}(s) & =  \frac{e^{-\frac{2s^{3/2}}{3}}}{4\sqrt{\pi}\,s^{3/4}}\Big[1 - \frac{41}{2^4\cdot 3}\,s^{-3/2} + \frac{9241}{2^9\cdot 3^2}\,s^{-3} - \frac{5075225}{2^{13}\cdot 3^4}\,s^{-9/2} + \frac{5153008945}{2^{19}\cdot 3^5}\,s^{-6}  \\
 & \phantom{\frac{e^{-\frac{2s^{3/2}}{3}}}{2\sqrt{\pi},s^{3/4}}}\, - \frac{1674966309205}{2^{23}\cdot 3^6}\,s^{-15/2} + \frac{3985569631633205}{2^{28}\cdot 3^8}\,s^{-9} + \mathcal{O}(s^{-21/2})\Big]  
 \end{split}
\end{equation}
and
\begin{equation}
\label{eq:TWright2}
\begin{split}
1 -F_{2}(s)  = & \frac{e^{-\frac{4s^{3/2}}{3}}}{16\pi\,s^{3/2}}\Big[1 - \frac{35}{2^3\cdot 3}\,s^{-3/2} + \frac{3745}{2^7\cdot 3^2}\,s^{-3} - \frac{805805}{2^{10}\cdot 3^4}\,s^{-9/2} + \frac{289554265}{2^{15}\cdot 3^5}\,s^{-6}  \\
 & \phantom{\frac{e^{-\frac{4s^{3/2}}{3}}}{16\pi\,s^{3/2}}}\, - \frac{31241084875}{2^{18}\cdot 3^6}\,s^{-15/2} + \frac{23604769513325}{2^{22}\cdot 3^8}\,s^{-9} + \mathcal{O}(s^{-21/2})\Big] 
\end{split}
\end{equation}

\section{Asymptotics of $F(s)$}
\label{app:F} 
We compute the asymptotics of $F(s)$ for large positive and negative values of $s$ and plot in Fig.~\ref{fig:plotsTails} the overlap between the complete PDF $F'(s)$ and the asymptotics obtained.
\subsection{Right tail using the determinantal formula}
Let us perform the trace expansion for $s \to +\infty$ on the form
\begin{equation}
F(s) = \partial_s  \left[  2 \det(I + \Ai_s+| \Ai_s 1 \rangle \langle   1| ) + (s-2)\det(I + \Ai_s)\right].
\end{equation}
Here we can perform an expansion in powers of Airy functions that is
in $e^{- \frac{2 k}{3} s^{3/2}}$, $k=1,2,\dots$ using the first two orders of the trace expansion of the Fredholm determinant
\be
\det(I + M) = 1 + {\rm Tr} M + \frac{1}{2} (  ({\rm Tr} M)^2 -  {\rm Tr} M^2 ) + \mathcal{O}(M^3)
\ee 
We have, where $\mathcal{O}(K^\ell)$ indicate the higher order traces in the expansion. Up to to the second order in powers of Airy functions, we obtain
\be
\begin{split}
 2 \det(I +& \Ai_s+| \Ai_s 1 \rangle \langle   1| ) + (s-2)\det(I + \Ai_s) \\
& = s + s {\rm Tr} \Ai_s + 2 {\rm Tr}  | \Ai_s 1 \rangle \langle   1|    + \mathcal{O}(\Ai^2) \\
& = s + s \int_0^{+\infty} \rmd x \Ai(2 x+s) + 2 \int_0^{+\infty} \int_0^{+\infty} \rmd x \rmd y \Ai(s+x+y) + \mathcal{O}(\Ai^2)
\end{split}
\ee
Differentiating allows to find back the CDF $F(s)$
\be
\begin{split}
 F(s)& = 1 + \int_0^{+\infty} \rmd x \Ai(2 x+s) - \frac{s}{2} \Ai(s)  - 2 \int_0^{+\infty} \rmd x \Ai(x+s) + \mathcal{O}(\Ai^2) \\
& = 1 - \frac{s}{2} \Ai(s)  - \frac{3}{2} \int_0^{+\infty} \rmd x \Ai(x+s)  + \mathcal{O}(\Ai^2)
\end{split}
\label{eq:RightTailAiry}
\ee
\subsection{Right tail using the asymptotics of $q(s)$ and the Tracy-Widom distributions}
We can now compare with the formula \eqref{eq:CDFwithoutDerivative}
\begin{equation}
\begin{split}
F(s)&=e^{ -\frac{1}{2} \int_s^{+\infty} \mathrm d r \, [ (r-s)q(r)^2+3 q(r)]} \\
& \times \left( 1+\frac{1}{2}(q'(s)^2-sq(s)^2-q(s)^4-q(s))(s-2q'(s)+2q(s)^2)\right) 
\end{split}
\end{equation}
which we recast into
\begin{equation}
\begin{split}
F(s)&=\frac{F_1(s)^3}{F_2(s)} \left( 1+\frac{1}{2}(q'(s)^2-sq(s)^2-q(s)^4-q(s))(s-2q'(s)+2q(s)^2)\right) \\
\end{split}
\end{equation}
Since $q(s)$ behaves asymptotically for large positive $s$ as the Airy function, we obtain at first order in $q(s)$ the expansion
\begin{equation}
\begin{split}
F(s)&=\frac{F_1(s)^3}{F_2(s)} \left( 1-\frac{s}{2}q(s) \right) +\mathcal{O}(q(s)^2)\\
\end{split}
\end{equation}
Using the asymptotics of the Tracy-Widom distributions \eqref{eq:TWright1} and \eqref{eq:TWright2} and the asymptotics of the Airy function, we obtain 

\begin{equation}
\label{eq:appRight}
\begin{split}
1 -& F(s)  \\
&= \frac{s^{3/4}e^{-\frac{2 s^{3/2}}{3}}}{4\sqrt{\pi}} \left[1+\frac{139s^{-3/2}}{48  }-\frac{11423s^{-3}}{4608  }+\frac{3907027s^{-9/2}}{663552  }-\frac{2886147455s^{-6}}{127401984  }+o(s^{-6})\right]
 \end{split}
\end{equation}

\subsection{Left tail using the asymptotics of $q(s)$ and the Tracy-Widom distributions}

Now we investigate the behaviour of $F(s)$ when $s\to -\infty$.  Using once again that
\begin{equation}
\begin{split}
F(s)&=\frac{F_1(s)^3}{F_2(s)} \left( 1+\frac{1}{2}(q'(s)^2-sq(s)^2-q(s)^4-q(s))(s-2q'(s)+2q(s)^2)\right) \\
\end{split}
\end{equation}
and reading the asymptotics of the Painlevé transcendent \eqref{eq:leftasymptoticsq} and of the Tracy-Widom distributions \eqref{eq:leftasymptoticsF1F2} we obtain the left tail of $F(s)$ as 

\begin{equation}
\label{eq:appLeft}
\begin{split}
F(s)&=2^{-203/48}e^{\zeta'(-1)/2} \exp \big[-\frac{\left| s\right| ^3}{24}-\frac{\left| s\right| ^{3/2}}{\sqrt{2}}+\frac{23}{16}\log\abs{s}+\frac{91}{8 \sqrt{2}
   \left| s\right| ^{3/2}}\\
   & \hspace*{5cm} -\frac{3957}{128 \, \abs{s} ^3}+\frac{28717}{128 \sqrt{2} \left| s\right| ^{9/2}}-\frac{469683}{512 \left| s\right|
   ^6}+o(s^{-6})\big]
   \end{split}
\end{equation}

\begin{figure}[t!]
\begin{center}
		\includegraphics[width=11cm]{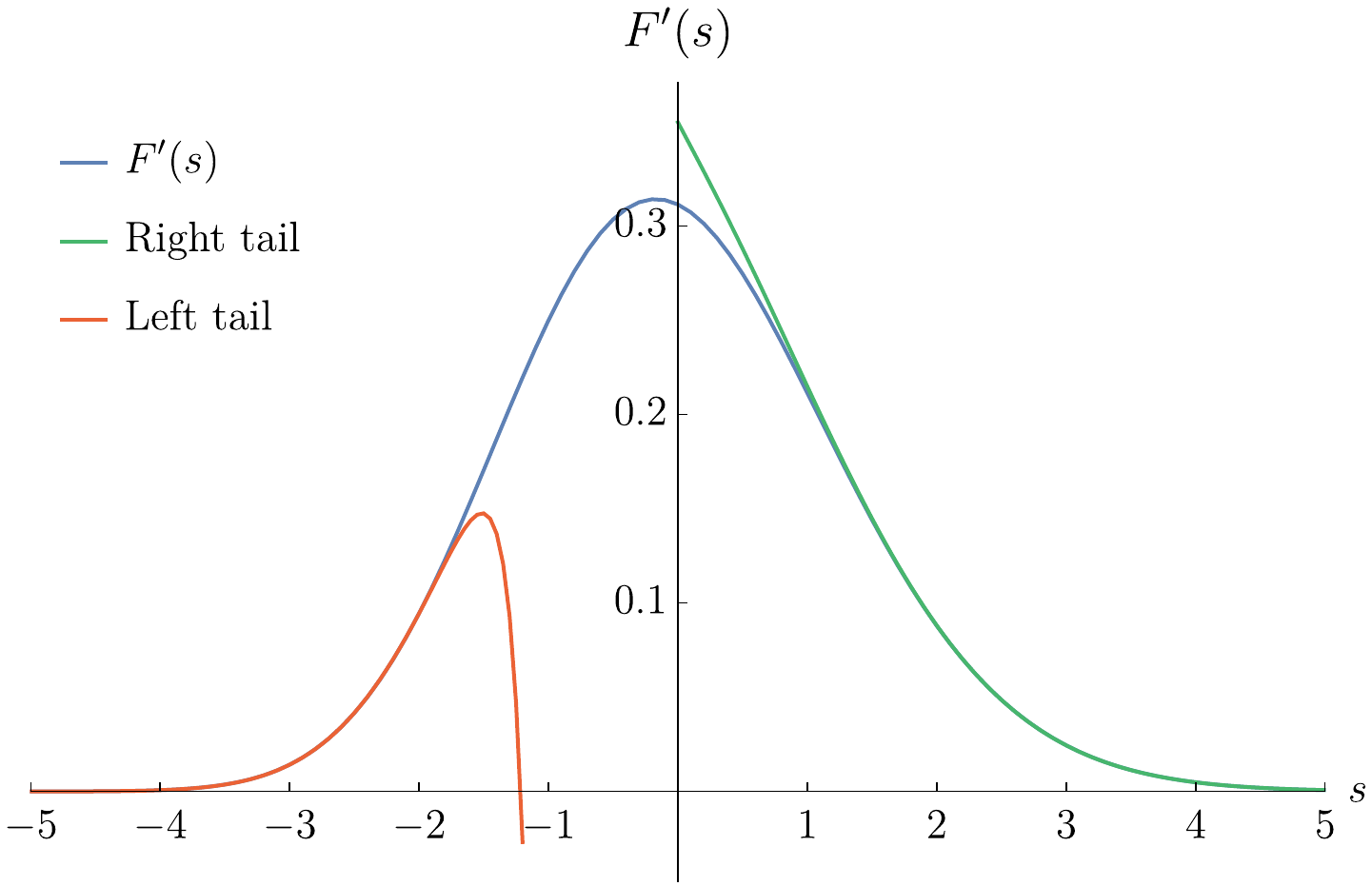}
		\includegraphics[width=11cm]{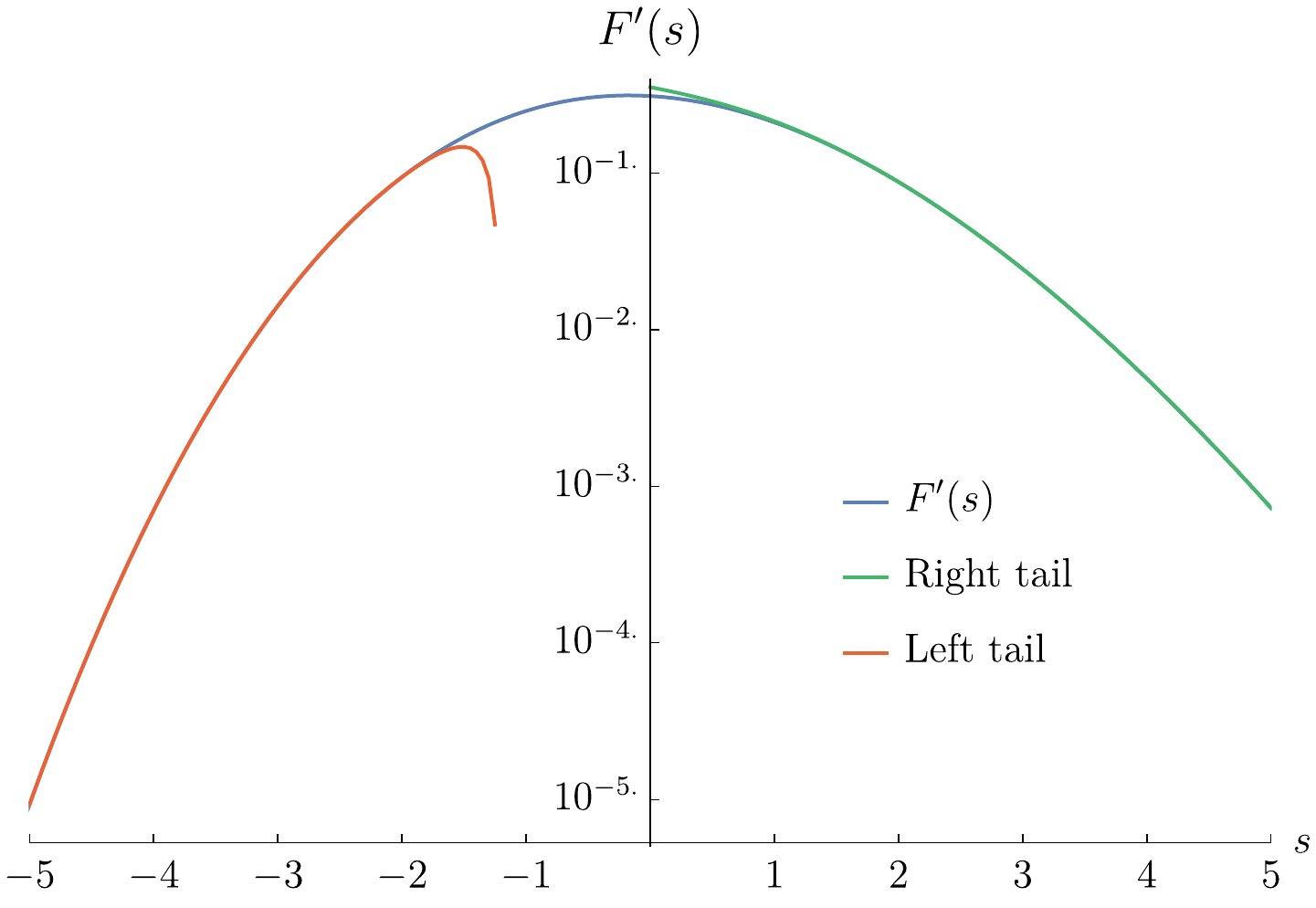}
\end{center}
\caption{Overlap of the left and right tails of the PDF of the critical stationary case (derivative of Eqs.~\eqref{eq:appLeft} and \eqref{eq:RightTailAiry}) with the complete PDF (derivative of Eq.~\eqref{eq:FintermsofF1F2}). \textbf{Top}. True scale. \textbf{Bottom}. Logarithmic scale on the vertical axis.}
\label{fig:plotsTails}
\end{figure}

\section{Extended kernel and the Kadomtsev--Petviashvili equation}
\label{app:KP} 

Recently it was shown that the height CDF at the (large time) KPZ fixed point for the full space problem is related to scale-invariant solutions of the Kadomtsev-Petviashvili (KP) equation \cite{quastel2019kp}. A related observation was made for the periodic KPZ fixed point \cite{prolhac2020}. This connection to the KP equation extends to the generating function at arbitrary time for the KPZ equation in full space, for some particular initial conditions, droplet, half-Brownian \cite{quastel2019kp} and Brownian \cite{doussal2019KP}. 
A similar relation was also obtained for a class of linear statistics associated to the Airy process \cite{doussal2019KP}. However at present no such relation is known for the half-space problem.

Here we provide an extended version of our kernel which can be related to the KP equation. It involves however an additional variable which plays the role of a "fictitious space". Although at this stage
we could not see a physical interpretation for this variable, we believe this fact is curious enough to be reported.

We recall that for the half-space problem at all times, the generating function of the exponential of the KPZ height at the origin, $x=0$, reads \eqref{res-scalar}
\begin{equation} 
\mathbb{E}\left[ \exp(-\varsigma \invgamma e^{H(t)}) \right]=\sqrt{\mathrm{Det}(I-\bar{K}_{t,\varsigma})_{\mathbb{L}^2(\mathbb{R}_+)}}.
\end{equation}
Denoting $\varsigma=e^{-r}$, one rewrites the kernel $\bar{K}$ as 
\begin{equation}
\begin{split}
\bar{K}_{t,e^{-r}}(x,y)
&=2\partial_x \iint_{C^2} \frac{\rmd w \rmd z}{(2\I\pi)^2 }G(z) G(w) \frac{\sin(\pi (z-w))}{\sin(\pi (z+w))}e^{-(x+r)z-(y+r)w  + t  \frac{w^3+z^3}{3} }\, ,
\end{split}
\end{equation}
where the function $G$ reads
\be
\ratioGamma(z) = \frac{\Gamma(A+\frac{1}{2}-z)}{\Gamma(A+\frac{1}{2}+z)}\frac{\Gamma(B+\frac{1}{2}-z)}{\Gamma(B+\frac{1}{2}+z)} \Gamma(2z) \, .
\ee
The kernel $\bar{K}_{t,e^{-r}}$ verifies two simple identities
\begin{equation}
\begin{cases}
\partial_t \bar{K}_{t,e^{-r}}(x,y)&=-\frac{1}{3}[\partial_x^3+\partial_y^3]\bar{K}_{t,e^{-r}}(x,y)\\
&~\\
\partial_r \bar{K}_{t,e^{-r}}(x,y)&=[\partial_x+\partial_y]\bar{K}_{t,e^{-r}}(x,y)
\end{cases}
\label{eq:DiffEqKernel}
\end{equation}
We can now extend the kernel by introducing a fictitious variable $u$ in the following way
\begin{equation}\label{eq:1DkernelAllTnew}
\begin{split}
\mathcal{K}_{t,e^{-r},u}(x,y)&=2\partial_x \iint_{C^2} \frac{\rmd w \rmd z}{(2\I\pi)^2 }G(z) G(w) \frac{\sin(\pi (z-w))}{\sin(\pi (z+w))}e^{-(x+r)z-(y+r)w+\frac{u}{2}(w^2-z^2)  + t  \frac{w^3+z^3}{3} }
\end{split}
\end{equation}
so that
\be
\mathcal{K}_{t,e^{-r},u=0}(x,y) = \bar{K}_{t,e^{-r}}(x,y)
\ee
The new kernel $\mathcal{K}_{t,e^{-r},u}$ verifies the same set of differential equations \eqref{eq:DiffEqKernel} and in addition verifies a third one
\begin{equation} \label{cond3} 
\partial_u \mathcal{K}_{t,e^{-r},u}(x,y)=\frac{1}{2}[\partial_y^2-\partial_x^2]\mathcal{K}_{t,e^{-r},u}(x,y)
\end{equation}
It was shown in \cite{quastel2019kp},
and known earlier in \cite{poppe1989general} (see discussion in \cite[Appendix E]{doussal2019KP}),
that the three conditions \eqref{eq:DiffEqKernel} and \eqref{cond3} imply that the Fredholm determinant associated to $\mathcal{K}$ is a $\tau$-function of the KP equation.
Thus, denoting $\mathcal{F}(r,t,u)=\sqrt{\mathrm{Det}(I-\mathcal{K}_{t,e^{-r},u})_{\mathbb{L}^2(\mathbb{R}_+)}}$, the function $\phi(r,t,u):=2\partial_r^2 \log(\mathcal{F}(r,t,u))$ solves the KP equation for $(r,t,u)$
\begin{equation}
\partial_t \phi+\phi \partial_r \phi +\frac{1}{12}\partial_r^3 \phi +\partial_r^{-1}\partial_u^2 \phi=0
\end{equation}
Note that the knowledge of $\phi(r,t,u=0)$ is equivalent to the knowledge of the half-space generating 
function for the KPZ equation. For the full-space KPZ equation, the variable $u$ was interpreted as a spatial variable whereas in our case, it is a fictitious variable with no obvious direct interpretation.


\newpage


\begin{thebibliography}{130}

\bibitem{KPZ} M. Kardar, G. Parisi and Y-C. Zhang, 
\textit{Dynamic Scaling of Growing Interfaces}, 
\doidoi{10.1103/PhysRevLett.56.889}{Physical Review Letters {\bf 56}, 889}, (1986).

\bibitem{baik2001symmetrized} J.~Baik and E.~M. Rains,
\newblock \textit{Symmetrized random permutations},
\newblock \href{https://arxiv.org/abs/math/9910019}{arXiv:math/9910019} and in Random matrix models and their applications, vol.~40 of {Math. Sci. Res. Inst. Publ.}, pp. 1--19. Cambridge Univ. Press, (2001).

\bibitem{SpohnTASEPGSE} M. Prahofer and H. Spohn
\textit{Current fluctuations for the totally asymmetric simple exclusion process}, 
\doidoi{10.1007/978-1-4612-0063-5_7}{Progress in Probability, Vol. 51, edited by V. Sidoravicius (Birkhauser,Boston, 2002) 185, 
arXiv:cond-mat/0101200}.

\bibitem{spohn2006exact} H.~Spohn,
\newblock \textit{Exact solutions for {KPZ}-type growth processes, random matrices, and equilibrium shapes of crystals},
\newblock \doidoi{10.1016/j.physa.2006.04.006}{Physica A: Stat. Mech. Appl. \textbf{369}(1), 71}, (2006).

\bibitem{quastel2012lectures} J.~Quastel,
\newblock \textit{Introduction to {KPZ}},
\newblock available online \href{https://www.math.toronto.edu/quastel/survey.pdf}{on the website of J.~Quastel} (2012).

\bibitem{corwin2012kardar} I.~Corwin, \textit{The Kardar-Parisi-Zhang equation and universality class}, Random Matrices: Theory Appl. 01(2012), 113 0001

\bibitem{corwin2014macdonald} I.~Corwin,
\textit{Macdonald processes, quantum integrable systems and the Kardar-Parisi-Zhang universality class}, 
\href{https://arxiv.org/abs/1403.6877}{Proceedings of the ICM, arXiv:1403.6877}.

\bibitem{borodin2012lectures} A.~Borodin and V.~Gorin,
\newblock \textit{Lectures on integrable probability},
\newblock Lecture notes, \href{https://arxiv.org/abs/1212.3351}{arXiv:1212.3351},  (2012).

\bibitem{borodin2014integrable} A.~Borodin and L.~Petrov,
\newblock \textit{Integrable probability: From representation theory to Macdonald processes},
\doidoi{10.1214/13-PS225}{Probab. Surveys \textbf{11}, 1}, (2014).

\bibitem{borodin2016lectures} A.~Borodin and L.~Petrov,
\newblock \textit{Lectures on integrable probability: stochastic vertex models and symmetric functions},
\newblock \href{https://arxiv.org/abs/1605.01349}{arXiv:1605.01349}, (2016).

\bibitem{ferrari2010interacting} P.~L.~ Ferrari, \textit{From interacting particle systems to random matrices}, J. Stat. Mech. (2010), P10016.

\bibitem{quastel2017totally} J.~Quastel and K.~Matetski,
\newblock \textit{From the totally asymmetric simple exclusion process to the {KPZ} fixed point},
\newblock \href{https://arxiv.org/abs/1710.02635}{arXiv:1710.02635},  (2017).

\bibitem{quastel2015one} J.~Quastel and H.~Spohn, The one-dimensional KPZ equation and its universality class, J. Stat.Phys. 160 (2015), 965–984.

\bibitem{takeuchi2016appetizer} K.~A.~Takeuchi, \textit{An appetizer to modern developments on the Kardar–Parisi–Zhang universality class}, Physica A 504 (2016), 77–105.

\bibitem{we} P. Calabrese, P. Le Doussal, A. Rosso, 
\textit{Free-energy distribution of the directed polymer at high temperature}, 
\doidoi{10.1209/0295-5075/90/20002}{Europhys. Lett. {\bf 90}, 20002}, (2010).

\bibitem{dotsenko} V. Dotsenko, 
\textit{Bethe ansatz derivation of the Tracy-Widom distribution for one-dimensional directed polymers} 
\doidoi{10.1209/0295-5075/90/20003}{EPL {\bf 90}, 20003} (2010);
\bibitem{dotsenko22} V.~Dotsenko
\textit{Replica Bethe ansatz derivation of the Tracy–Widom distribution of the free energy fluctuations in one-dimensional directed polymers}, 
\doidoi{10.1088/1742-5468/2010/07/P07010}{J. Stat. Mech. P07010}, (2010);  
\bibitem{dotsenko33}
V. Dotsenko and B. Klumov, 
\textit{Bethe ansatz solution for one-dimensional directed polymers in random media} 
\doidoi{10.1088/1742-5468/2010/03/P03022}{J. Stat. Mech. P03022}, (2010).

\bibitem{spohnKPZEdge} T. Sasamoto, H. Spohn, 
\textit{One-dimensional Kardar-Parisi-Zhang equation: an exact solution and its universality}, 
\doidoi{10.1103/PhysRevLett.104.230602}{Physical Review Letters {\bf 104}, 230602}, (2010).
\bibitem{spohnKPZEdge2}
T. Sasamoto and H. Spohn, 
\textit{Exact height distributions for the KPZ equation with narrow wedge initial condition}, 
\doidoi{10.1016/j.nuclphysb.2010.03.026}{Nucl. Phys. B {\bf 834} 523, arXiv:1002.1879}, (2010) ;
\bibitem{spohnKPZEdge3}
T. Sasamoto and H. Spohn,
\textit{The crossover regime for the weakly asymmetric simple exclusion process} 
\doidoi{	10.1007/s10955-010-9990-z}{J. Stat. Phys. {\bf 140} 209 , arXiv:1002.1873}, (2010).

\bibitem{amir2011probability} G.~Amir, I.~Corwin and J.~Quastel,
\textit{Probability distribution of the free energy of the continuum directed random polymer in 1 + 1 dimensions},
 \doidoi{10.1002/cpa.20347}{Comm. Pure and Appl. Math. {\bf 64}, 466}, (2011).

\bibitem{sineG} P. Calabrese, M. Kormos and P. Le Doussal, 
\textit{From the sine-Gordon field theory to the Kardar-Parisi-Zhang growth equation}, 
\doidoi{  10.1209/0295-5075/107/10011}{arXiv:1405.2582, EPL 107 10011}, (2014).

\bibitem{CLDflat} P. Calabrese, P. Le Doussal,
\textit{Exact solution for the Kardar-Parisi-Zhang equation with flat initial conditions},
\doidoi{10.1103/PhysRevLett.106.250603}{Physical Review Letters {\bf 106}, 250603}, (2011)

\bibitem{CLDflat2} P. Le Doussal, P. Calabrese, 
\textit{The KPZ equation with flat initial condition and the directed polymer with one free end},  
\doidoi{10.1088/1742-5468/2012/06/P06001}{J. Stat. Mech. P06001}, (2012).

\bibitem{PLDCrossoverDropFlat} P. Le Doussal, 
\textit{Crossover from droplet to flat initial conditions in the KPZ equation from the replica Bethe ansatz} 
\doidoi{10.1088/1742-5468/2014/04/P04018}{arXiv:1401.1081, J. Stat. Mech. P04018}, (2014).

\bibitem{crissingprob2} A. De Luca and P. Le Doussal,
 \textit{Crossing probability for directed polymers in random media: exact tail of the distribution}, 
 \doidoi{	10.1103/PhysRevE.93.032118}{Phys. Rev. E \textbf{93}, 032118,  arXiv:1511.05387}, (2016)
 
\bibitem{tracy1994level} C.~A.~Tracy, H.~Widom, \textit{Level-spacing distributions and the Airy kernel}, Comm. Math. Phys. 159(1), 151-174 (1994).

\bibitem{tracy1996orthonormal} C.~A.~Tracy, H.~Widom, \textit{On orthogonal and symplectic matrix ensembles}, Comm. Math. Phys., 177(3), 727-754 (1996).

\bibitem{dotsenkoGOE} V. Dotsenko, 
\textit{Replica Bethe ansatz derivation of the GOE Tracy-Widom distribution in one-dimensional directed polymers with free boundary conditions} 
\doidoi{	10.1088/1742-5468/2012/11/P11014}{ J. Stat. Mech. P11014}, (2012)

\bibitem{Quastelflat} J. Ortmann, J. Quastel and D. Remenik, 
\textit{Exact formulas for random growth with half-flat initial data} 
\doidoi{	10.1214/15-AAP1099}{Ann. Appl. Probab. \textbf{26} 507} (2016)

\bibitem{SasamotoStationary} T. Imamura, T. Sasamoto,
\textit{Exact solution for the stationary Kardar-Parisi-Zhang equation},
 \doidoi{10.1103/PhysRevLett.108.190603}{Physical Review Letters {\bf 108}, 190603}, (2012).
 

\bibitem{baik2000limiting} J.~Baik and E.M.~Rains, \textit{Limiting distributions for a polynuclear growth model with external sources}, J. Stat. Phys.100 (2000), 523–542.

\bibitem{baik2010limit} J.~Baik, P.~L.~Ferrari, and S.~P\'ech\'e, \textit{Limit process of stationary TASEP near the characteristic line}, Comm. Pure Appl. Math.63(2010), 1017–1070

\bibitem{ferrari2006scaling} P.~L.~Ferrari and H.~Spohn, \textit{Scaling limit for the space-time covariance of the stationary totally asymmetric simple exclusion process}, Comm. Math. Phys. 265(2006), 1–44.

\bibitem{BCFV} A. Borodin, I. Corwin, P. L. Ferrari. B. Veto, 
\textit{Height fluctuations for the stationary KPZ equation},
 \doidoi{10.1007/s11040-015-9189-2}{Math. Phys. Anal. Geom. \textbf{18}, 20}, (2015).

\bibitem{aggarwal2018current} A. Aggarwal, Current Fluctuations of the Stationary ASEP and Six-Vertex Model, Duke Math.J. 167(2018), 269–384

\bibitem{KPZFixedPoint} I.~{Corwin}, J.~{Quastel}, et D.~{Remenik}, 
\textit{Renormalization fixed  point of the KPZ universality class}, 
 \doidoi{10.1007/s10955-015-1243-8}{J Stat Phys 160: 815}, (2015).

\bibitem{exp4} K.~A. Takeuchi, M.~Sano, 
\textit{Universal fluctuations of growing interfaces: Evidence in turbulent liquid crystals},
   \doidoi{10.1103/PhysRevLett.104.230601}{Physical Review Letters {\bf 104}  230601}, (2010).

   \bibitem{exp444}
K.~A. Takeuchi, M.~Sano, T.~Sasamoto, H.~Spohn, 
\textit{Growing interfaces uncover  universal fluctuations behind scale invariance}, 
\doidoi{10.1038/srep00034}{Scientific Reports 1 34}, (2011).

\bibitem{Takeuchi} K.~A. Takeuchi, M.~Sano, 
\textit{Evidence for geometry-dependent universal fluctuations  of the {Kardar-Parisi-Zhang} interfaces in liquid-crystal turbulence}, 
\doidoi{10.1007/s10955-012-0503-0}{J.  Stat. Phys. 147  853--890}, (2012).

\bibitem{TakeuchiCrossover} K.~A. Takeuchi, 
\textit{Crossover from growing to stationary interfaces in the {Kardar-Parisi-Zhang} class}, 
\doidoi{10.1103/PhysRevLett.110.210604}{Physical Review Letters {\bf 110} 210604}, (2013).

\bibitem{TakeuchiHHLReview} T. Halpin-Healy, K.~A. Takeuchi,
\textit{A KPZ cocktail-shaken, not stirred: Toasting 30 years of kinetically roughened surfaces}
 \doidoi{10.1007/s10955-015-1282-1}{J. Stat. Phys. {\bf 160}, 794 }, (2015).

\bibitem{deNardisPLDTakeuchi} J. De Nardis, P. Le Doussal, K. A. Takeuchi, 
\textit{Memory and universality in interface growth}, 
\doidoi{	10.1103/PhysRevLett.118.125701}{Phys. Rev. Lett. 118, 125701 }, (2017).

\bibitem{TakeuchiItoPng} I. Ito, K. A. Takeuchi,
\textit{When fast and slow interfaces grow together: connection to the half-space problem of the Kardar-Parisi-Zhang class},
\doidoi{10.1103/PhysRevE.97.040103}{arXiv:1802.10284, Phys. Rev. E 97, 040103} (2018).

\bibitem{KardarTransition} M. Kardar, 
\textit{Depinning by quenched randomness}. 
\doidoi{10.1103/PhysRevLett.55.2235}{Physical Review Letters \textbf{55}, 2235,} (1985).

\bibitem{sasamotohalfspace} T. Sasamoto, T. Imamura,
 \textit{Fluctuations of a one-dimensional polynuclear growth model in a half space} 
 \doidoi{10.1023/B:JOSS.0000022374.73462.85}{J. Stat. Phys. \textbf{115} 749, arXiv:cond-mat/0307011}, (2004)

\bibitem{baik2018pfaffian} J.~Baik, G.~Barraquand, I.~Corwin and T.~Suidan,
\textit{Pfaffian Schur processes and last passage percolation in a half-quadrant}, 
\newblock \doidoi{10.1214/17-AOP1226}{Ann. Probab. \textbf{46}(6), 3015}, (2018).

\bibitem{baik2018pfaffian22}
J.~Baik, G.~Barraquand, I.~Corwin and T.~Suidan,
\textit{Facilitated exclusion process},  Computation and Combinatorics in Dynamics, Stochastics and Control. Abelsymposium 2016. Abel Symposia, vol 13. Springer, 
\href{https://arxiv.org/abs/1707.01923}{arXiv:1707.01923} (2017).

\bibitem{borodin2005eynard} A.~Borodin and E.~M. Rains.
\newblock {Eynard--Mehta theorem, Schur process, and their Pfaffian analogs}.
\newblock \doidoi{10.1007/s10955-005-7583-z}{ J. Stat. Phys., 121(3-4):291--317}, (2005).

\bibitem{TWhalf} C. Tracy, H. Widom, 
\textit{The Bose gas and asymmetric simple exclusion process on the half-line}
\doidoi{10.1007/s10955-012-0686-4}{J. Stat. Phys., 150:1}, (2013).

\bibitem{tracy2013asymmetric} C.~A. Tracy and H.~Widom.
\newblock The asymmetric simple exclusion process with an open boundary.
\newblock \doidoi{10.1063/1.4822418}{ J. Math. Phys., 54(10):103301}, (2013).

\bibitem{halfASEPBarraquand} G.~Barraquand, A.~Borodin, I.~Corwin, M.~Wheeler,
 \newblock \textit{Stochastic six-vertex model in a half-quadrant and half-line open ASEP}, 
  \doidoi{10.1215/00127094-2018-0019}{Duke Math. J. 167, no. 13, 2457-2529}, (2018).

\bibitem{OConnelSymmetrized} N. O'Connell, T. Sepp\"al\"ainen, N. Zygouras, 
\textit{Geometric RSK correspondence, Whittaker functions and symmetrized random polymers},  
\doidoi{10.1007/s00222-013-0485-9}{Invent. math. vol. 197, p. 361–416}, (2014).

\bibitem{barraquand2018half} G.~Barraquand, A.~Borodin and I.~Corwin,
\newblock \textit{Half-space Macdonald processes}.
\newblock  \href{https://arxiv.org/abs/1802.08210}{arXiv:1802.08210}, to appear in \doidoi{10.1017/fmp.2020.3}{Forum Math. Pi} (2018).

\bibitem{tracy2005matrix} C.~Tracy, H.~Widom, \textit{Matrix kernels for the Gaussian orthogonal and symplectic ensembles}, Annales de l'institut Fourier Vol. 55, No. 6, pp. 2197-2207, (2005).

\bibitem{bisi2020geometric} E.~Bisi, N.~O'Connell and N.~Zygouras,
\newblock \textit{The geometric Burge correspondence and the partition function
	of polymer replicas},
\newblock \href{https://arxiv.org/abs/2001.09145}{arXiv:2001.09145},  (2020).

\bibitem{GueudrePLD} T.~Gueudr{\'e}, P.~Le~Doussal,
 \textit{Directed polymer near a hard wall and KPZ equation in the half-space},
\doidoi{10.1209/0295-5075/100/26006}{Europhys. Lett. \textbf{100}, 26006,} (2012).

\bibitem{krajenbrink2018large} A.~Krajenbrink, P.~Le Doussal,
\newblock \textit{Large fluctuations of the KPZ equation in a half-space.}
\newblock \doidoi{10.21468/SciPostPhys.5.4.032}{SciPost Phys. {\bf 5}, 032}, (2018).

\bibitem{forrester2000painlev}
PJ~Forrester.
\newblock \textit{Painleve transcendent evaluation of the scaled
  distribution of the smallest eigenvalue in the Laguerre orthogonal and
  symplectic ensembles.}
\newblock arXiv nlin/0005064, (2000).



\bibitem{borodin2016directed} A.~Borodin, A.~Bufetov, I.~Corwin, 
\textit{Directed random polymers via nested contour integrals}.
 \doidoi{10.1016/j.aop.2016.02.001}{Annals of Physics, \textbf{368} 191--247}, (2016).

\bibitem{AlexLD} A.~Krajenbrink, P.~Le~Doussal, 
\textit{Replica Bethe Ansatz solution to the Kardar-Parisi-Zhang equation on the half-line}. 
\doidoi{10.21468/SciPostPhys.8.3.035}{SciPost Phys. 8, 035}, (2020).

\bibitem{deNardisPLDTT} J. de Nardis, A. Krajenbrink, P. Le Doussal, T. Thiery, 
\textit{Delta-Bose gas on a half-line and the KPZ equation: boundary bound states and unbinding transitions},
\href{https://arxiv.org/pdf/1911.06133.pdf}{arXiv:1911.06133}, (2019).

\bibitem{seppalainen2012scaling} T.~Sepp{\"a}l{\"a}inen,
\newblock \textit{Scaling for a one-dimensional directed polymer with boundary conditions},
\newblock \doidoi{10.1214/10-aop617}{Ann. Probab. \textbf{40}(1), 19}, (2012).

\bibitem{betea2019stationary} D.~Betea, P.~L. Ferrari and A.~Occelli,
\newblock \textit{Stationary half-space last passage percolation}, Comm. Math. Phys. 377, pages421–467 (2020)
\newblock \href{https://arxiv.org/abs/1905.08582}{arXiv:1905.08582},  (2019).

\bibitem{corwin2016open} I.~Corwin and H.~Shen,
\newblock \textit{Open ASEP in the weakly asymmetric regime},
\newblock \doidoi{10.1002/cpa.21744}{Comm. Pure Appl. Math. \textbf{71}(10), 2065}, (2018).

\bibitem{ghosal2018moments} P.~Ghosal.
\textit{Moments of the SHE under delta initial measure}, 
\href{https://arxiv.org/abs/1808.04353}{arXiv:1808.04353}, (2018).

\bibitem{parekh2019kpz123} S.~Parekh.
\newblock \textit{The KPZ limit of ASEP with boundary}.
\newblock \doidoi{10.1007/s00220-018-3258-x}{Communications in Mathematical Physics, 365(2):569--649}, (2019).

\bibitem{wu2018intermediate} X.~Wu,
\newblock \textit{Intermediate disorder regime for half-space directed polymers},
\newblock \href{https://arxiv.org/abs/1804.09815}{arXiv:1804.09815},  (2018).

\bibitem{parekh2019} S.~Parekh.
 \textit{Positive random walks and an identity for half-space SPDEs}.
 \href{https://arxiv.org/abs/1901.09449}{arXiv:1901.09449}, (2019).

\bibitem{liggett1975ergodic} T.~M. Liggett,
\newblock \textit{Ergodic theorems for the asymmetric simple exclusion process},
\newblock \doidoi{10.1090/s0002-9947-1975-0410986-7}{Trans. Amer. Math. Soc. \textbf{213}, 237}, (1975).

\bibitem{derrida1993exact} B.~Derrida, M.~R. Evans, V.~Hakim and V.~Pasquier,
\newblock \textit{Exact solution of a {1D} asymmetric exclusion model using a matrix formulation},
\newblock \doidoi{10.1088/0305-4470/26/7/011}{J. Phys. A \textbf{26}(7), 1493}, (1993).

\bibitem{grosskinsky} S. Grosskinsky
\textit{Phase transitions in nonequilibrium stochastic particle systems with local conservation laws}
\href{https://homepages.warwick.ac.uk/~masgav/diss.pdf}{PhD thesis, TU Munich}, (2004)

\bibitem{derrida2004asymmetric} B.~Derrida, C.~Enaud, and J.~Lebowitz.
\newblock \textit{The asymmetric exclusion process and Brownian excursions}.
\newblock \doidoi{10.1023/b:joss.0000019833.35328.b4}{ Journal of Statistical Physics, 115(1-2):365--382}, (2004).

\bibitem{bryc2019limit} W.~Bryc, Y.~Wang.
\newblock \textit{Limit fluctuations for density of asymmetric simple exclusion
  processes with open boundaries}.
\newblock \doidoi{10.1214/18-aihp945}{  Annales de l'Institut Henri Poincar{\'e}, Probabilit{\'e}s et
  Statistiques, volume~55, pages 2169--2194},   (2019).

\bibitem{baik2001asymptotics} J.~Baik and E.~M. Rains,
\newblock \textit{The asymptotics of monotone subsequences of involutions},
\newblock \doidoi{10.1215/S0012-7094-01-10921-6}{Duke Math. J. \textbf{109}(2), 205}, (2001).




\bibitem{wang2009largest} D.~Wang,
\newblock \textit{The largest sample eigenvalue distribution in the rank 1 quaternionic spiked model of {W}ishart ensemble},
\newblock \doidoi{10.1214/08-aop432}{Ann. Probab. pp. 1273--1328} (2009).

\bibitem{Halpin2014} T.~Halpin-Healy and Y.~Lin.
\newblock \textit{Universal aspects of curved, flat, and stationary-state Kardar-Parisi-Zhang statistics}.
\newblock \doidoi{10.1103/PhysRevE.89.010103}{Physical Review E, 89(1):010103}, (2014).

\bibitem{borodin2014macdonald} A.~Borodin and I.~Corwin,
\textit{ Macdonald processes}, 
\doidoi{10.1007/s00440-013-0482-3}{Prob. Theor. Rel. Fields \textbf{158}, no. 1-2, 225--400, arXiv:1111.4408}, (2014).

\bibitem{borodin2012free} A.~Borodin, I.~Corwin and P.~Ferrari,
\newblock \textit{Free energy fluctuations for directed polymers in random media in 1+ 1 dimension},
\newblock \doidoi{10.1002/cpa.21520}{Comm. Pure Appl. Math. \textbf{67}(7), 1129} (2014).

\bibitem{borodin2014spectral} A.~Borodin, I.~Corwin, L.~Petrov and T.~Sasamoto,
\newblock \textit{Spectral theory for interacting particle systems solvable by coordinate Bethe ansatz},
\newblock \doidoi{10.1007/s00220-015-2424-7}{Comm. Math. Phys. \textbf{339}(3), 1167}, (2015).

\bibitem{borodin2015spectral} A.~Borodin, I.~Corwin, L.~Petrov and T.~Sasamoto,
\newblock \textit{Spectral theory for the q-Boson particle system},
\newblock \doidoi{10.1112/S0010437X14007532}{Compos. Math. \textbf{151}, 1}, (2015).

\bibitem{kardareplica} M. Kardar, 
\textit{Replica Bethe ansatz studies of two-dimensional interfaces with quenched random impurities}. 
\doidoi{10.1016/0550-3213(87)90203-3}{Nucl.  Phys. B {\bf 290}, 582}, (1987).

\bibitem{ll} E. H. Lieb and W. Liniger, 
\textit{Exact Analysis of an Interacting Bose Gas. I. The General Solution and the Ground State}, 
\doidoi{10.1103/PhysRev.130.1605}{Phys. Rev. \textbf{130}, 1605}, (1963).

\bibitem{GaudinHardWall} M. Gaudin, 
\textit{Boundary Energy of a Bose Gas in One Dimension }, 
\doidoi{10.1103/PhysRevA.4.386}{Phys. Rev. A 4 386,} (1971).

\bibitem{BAhardwall} N. Oelkers, M.T. Batchelor, M. Bortz, X.W. Guan,
\textit{Bethe Ansatz study of one-dimensional Bose and Fermi gases with periodic and hard wall boundary conditions}, 
\doidoi{  10.1088/0305-4470/39/5/005}{arXiv:cond-mat/0511694, J. Phys. A {\bf 39} 1073} (2006).

\bibitem{gaudin2014bethe} M.~Gaudin and J-S Caux.
 \textit{The Bethe Wavefunction}.
 \doidoi{10.1017/CBO9781107053885}{Cambridge University Press}, (2014).

\bibitem{gueudre2012directed} T.~Gueudr{\'e}, P.~Le~Doussal,
\newblock \textit{Directed polymer near a hard wall and KPZ equation in the half-space},
\newblock Europhys. Lett. \textbf{100}, 26006 (2012), \doi{10.1209/0295-5075/100/26006}.

\bibitem{Castillo} I. P. Castillo, T. Dupic, 
\textit{Reunion Probabilities of $N$ One-Dimensional Random Walkers with Mixed Boundary Conditions}, 
\doidoi{10.1007/s10955-014-1017-8}{J. Stat Phys \textbf{3 } 156:606--616, arXiv:1311.0654} (2014).

\bibitem{VanDiejen} J.F. van Diejen, E. Emsiz,
\textit{Orthogonality of Bethe Ansatz eigenfunctions for the Laplacian on a hyperoctahedral Weyl alcove}
\doidoi{10.1007/s00220-016-2719-3}{Commun. Math. Phys. 350, no. 3, 1017}, (2017),
\bibitem{VanDiejen22}
J.F. van Diejen, E. Emsiz, I.N. Zurrian, 
\textit{Completeness of the Bethe Ansatz for an open q-boson system with integrable boundary interactions},
\doidoi{10.1007/s00023-018-0658-6}{arXiv:1611.05922, Ann. Henri Poincaré  19: 1349}, (2018).

\bibitem{EmsizComplete} E. Emsiz,
\textit{Completeness of the Bethe ansatz on Weyl alcoves}, 
\doidoi{10.1007/s11005-009-0359-7}{Lett. Math. Phys. 91, 61--70} (2010)

\bibitem{GutkinSutherland} E. Gutkin, B. Sutherland, 
\textit{Completely integrable systems and groups generated by reflections}, 
\doidoi{10.1073/pnas.76.12.6057}{PNAS, 76:6057,} (1979).

\bibitem{HeckmanOpdam} G. J. Heckman, E. M. Opdam, 
\textit{Yang's System of Particles and Hecke Algebras} 
\doidoi{10.2307/2951825}{Ann. Math., 145:139--173}, (1997).

\bibitem{PLD1} P. Le Doussal (2016) Unpublished notes

\bibitem{ChineseBA} Y. Hao, Y. Zhang, J. Q. Liang and S. Chen, 
\textit{Ground-state properties of one-dimensional ultracold Bose gases in a hard-wall trap}, 
\doidoi{10.1103/PhysRevA.73.063617}{Phys. Rev. A {\bf 73},063617} (2006).

\bibitem{m-65} J. B. McGuire,
 \textit{Study of Exactly Soluble One-Dimensional N-Body Problems}.
 \doidoi{10.1063/1.1704156}{J. Math. Phys. {\bf 5}, 622}, (1964).

\bibitem{CLR10} P. Calabrese, P. Le Doussal, A. Rosso, 
\textit{Free-energy distribution of the directed polymer at high temperature}, 
\doidoi{10.1209/0295-5075/90/20002}{Europhys. Lett. {\bf 90}, 20002}, (2010).

\bibitem{SasamotoStationary2} T. Imamura, T. Sasamoto,
\textit{Stationary correlations for the 1D KPZ equation},
\doidoi{10.1007/s10955-013-0710-3}{J. Stat. Phys. \textbf{150}, 908-939 }, (2013).

\bibitem{Knuth_1995} D.~E. Knuth.
\newblock \textit{Overlapping Pfaffians}.
\newblock \doidoi{10.37236/1263}{The Electronic Journal of Combinatorics}, (1995).

\bibitem{alberts2014intermediate} T.~Alberts, K.~Khanin and J.~Quastel,
\newblock \textit{The intermediate disorder regime for directed polymers in dimension $1+ 1$},
\newblock \doidoi{10.1214/13-aop858}{Ann. Probab. \textbf{42}(3), 1212} (2014).

\bibitem{oconnell2001brownian} N.~O'Connell and M.~Yor,
\newblock \textit{Brownian analogues of Burke's theorem},
\newblock \doidoi{10.1016/S0304-4149(01)00119-3}{Stochastic Process. Appl. \textbf{96}(2), 285} (2001).

\bibitem{baryshnikov2001gues} Y.~Baryshnikov,
\newblock \textit{GUEs and queues},
\newblock \doidoi{10.1007/pl00008760}{Probab. Theory Rel. Fields \textbf{119}(2), 256} (2001).

\bibitem{gravner2001limit} J.~Gravner, C.~A. Tracy and H.~Widom,
\newblock \textit{Limit theorems for height fluctuations in a class of discrete space and time growth models},
\newblock \doidoi{10.1023/a:1004879725949}{J. Stat. Phys. \textbf{102}(5-6), 1085} (2001).


\bibitem{tracy2008integral} C.~A. Tracy and H.~Widom, 
\textit{Integral formulas for the asymmetric simple	exclusion process}, 
\doidoi{10.1007/s00220-008-0443-3}{Comm. Math. Phys. \textbf{279} , no.~3, 815--844}, (2008).

\bibitem{tracy2009asep} C.~A. Tracy and H. Widom, 
\textit{On ASEP with step bernoulli initial condition}, 
\doidoi{10.1007/s10955-009-9867-1}{J. Stat. Phys. \textbf{137}, no.~5-6, 825}, (2009).

\bibitem{macdonald1995symmetric} I.~G. Macdonald,
\newblock \textit{Symmetric functions and Hall polynomials}, vol. 354,
\newblock \href{https://global.oup.com/academic/product/symmetric-functions-and-hall-polynomials-9780198739128?cc=fr&lang=en&}{Clarendon press Oxford}, (1995).

\bibitem{venkateswaran2015symmetric} V.~Venkateswaran,
\newblock \textit{Symmetric and nonsymmetric Hall-Littlewood polynomials of type BC},
\newblock \doidoi{10.1007/s10801-015-0583-4}{J. Algebr. Comb. \textbf{42}, 331} (2015).

\bibitem{de1955some} N.~De~Bruijn,
 \textit{On some multiple integrals involving determinants.} 
 \href{https://pure.tue.nl/ws/files/1920642/597510.pdf}{J. Indian Math. Soc, \textbf{19} 133--151}, (1955).

\bibitem{rains2000correlation} E.~M. Rains,
\newblock \textit{Correlation functions for symmetrized increasing subsequences},
\href{https://arxiv.org/abs/math/0006097}{arXiv:math/0006097}, (2000).

\bibitem{Ortmann} J. Ortmann, J. Quastel, D. Remenik, 
\textit{A Pfaffian representation for flat ASEP}, 
\doidoi{10.1002/cpa.21644}{Comm. Pure Appl. Math. \textbf{70}  1, 3}, (2015).

\bibitem{hastings1980boundary} S.~P. Hastings and J.~B. Mcleod.
\newblock \textit{A boundary value problem associated with the second Painlev{\'e} transcendent and the Korteweg-De Vries equation},
\newblock \doidoi{10.1007/bf00283254}{Archive for Rational Mechanics and Analysis, 73(1):31--51}, (1980).

\bibitem{imamura2004fluctuations} T.~Imamura and T.~Sasamoto.
\newblock \textit{Fluctuations of the one-dimensional polynuclear growth model with external sources}.
\newblock \doidoi{10.1016/j.nuclphysb.2004.07.030}{Nuclear Physics B, 699(3):503--544}, (2004).

\bibitem{baik2008asymptotics} J.~Baik, R.~Buckingham, and J.~DiFranco,
\newblock \textit{Asymptotics of Tracy-Widom distributions and the total integral of a Painlev{\'e} II function},
\newblock \doidoi{10.1007/s00220-008-0433-5}{Communications in Mathematical Physics, \textbf{280} (2):463--497}, (2008), .

\bibitem{Bornemann} F.~Bornemann, 
\textit{On the Numerical Evaluation of Fredholm Determinants},
\doidoi{10.1090/s0025-5718-09-02280-7}{Math. Comp. \textbf{79} 871} (2010).

\bibitem{Bornemann2} F.~Bornemann.
\newblock \textit{On the numerical evaluation of distributions in random matrix theory:  a review.}
\newblock \href{http://math-mprf.org/journal/articles/id1229/}{Markov Processes Relat. Fields 16 (2010) 803-866}, (2009).

\bibitem{SasamotoHalfBrownian} T. Imamura, T. Sasamoto,
\textit{Replica approach to the KPZ equation with half Brownian motion initial condition},
\doidoi{10.1088/1751-8113/44/38/385001}{J. Phys. A: Math. Theor. 44 385001}, (2011).

\bibitem{png} J.~Baik, E.~M. Rains, 
\textit{Limiting distributions for a polynuclear growth model  with external sources}, 
\doidoi{10.1023/A:1018615306992}{J. Stat. Phys. 100 523--541}, (2000).

\bibitem{ferrari2005determinantal} P.~L. Ferrari and H.~Spohn,
\newblock \textit{A determinantal formula for the GOE Tracy-Widom distribution},
\newblock \doidoi{10.1088/0305-4470/38/33/L02}{Journal of Physics A: Mathematical and General, \textbf{38} (33):L557}, (2005).

\bibitem{nadalLeft} G.~Borot, B.~Eynard, S.~N. Majumdar, and C.~Nadal.
\newblock \textit{Large deviations of the maximal eigenvalue of random matrices}.
\newblock \doidoi{10.1088/1742-5468/2011/11/P11024}{Journal of Statistical Mechanics: Theory and Experiment, (11):P11024}, (2011).

\bibitem{nadalRight} G.~Borot and C.~Nadal.
\newblock \textit{Right tail asymptotic expansion of Tracy--Widom beta laws}.
\newblock \doidoi{10.1142/s2010326312500062}{Random Matrices: Theory and Applications, 1(03):1250006}, (2012).

\bibitem{quastel2019kp} J.~Quastel and D.~Remenik.
\newblock \textit{KP governs random growth off a one dimensional substrate}.
\newblock \href{https://arxiv.org/abs/1908.10353}{arXiv:1908.10353}, (2019).

\bibitem{prolhac2020} S.~Prolhac.
\newblock \textit{Riemann surfaces for KPZ with periodic boundaries}.
\newblock \doidoi{10.21468/SciPostPhys.8.1.008}{SciPost Physics, 8(1)} (2020).

\bibitem{doussal2019KP} P.~Le Doussal.
\newblock \textit{Large deviations for the KPZ equation from the KP equation}.
\newblock \href{https://arxiv.org/abs/1910.03671}{arXiv:1910.03671}, (2019).

\bibitem{poppe1989general} C.~Poppe.
\newblock \textit{General determinants and the tau function for the
  Kadomtsev--Petviashvili hierarchy}.
\newblock \doidoi{10.1088/0266-5611/5/4/012}{ Inverse Problems, 5(4):613}, (1989).


\end{thebibliography}
\end{document}